\documentclass[pdftex,pagesize=auto,paper=letter,11pt,abstract=on,bibliography=oldstyle]{scrartcl}

\usepackage{etex}
\usepackage[utf8]{inputenc}
\usepackage[pdftex]{xcolor}
\usepackage{amsthm}
\usepackage{amssymb,amsmath}
\usepackage{xspace}
\usepackage{rotating}
\usepackage{multirow}
\usepackage{booktabs}
\usepackage{url}
\usepackage[protrusion=true,expansion=true]{microtype}
\usepackage{lmodern}
\usepackage[T1]{fontenc}
\usepackage{tikz}
\usepackage{mathtools}
\usepackage{paralist}
\usepackage{accents}
\usepackage{pgfplots}
\pgfplotsset{compat=1.9}

\usepackage{authblk}
\usepackage[margin=2.5cm]{geometry}

\usepackage[numbers,sort]{natbib}

\usepackage{subcaption}
\usepackage[ruled, vlined, linesnumbered]{algorithm2e}
\DontPrintSemicolon

\SetCommentSty{commentfont}
\newcommand{\assign}{\longleftarrow}

\SetKwFunction{queueIsEmpty}{isEmpty}%
\SetKwFunction{queueIsNotEmpty}{isNotEmpty}%
\SetKwFunction{queueMinElement}{minElement}%
\SetKwFunction{queueDeleteMin}{deleteMin}%
\SetKwFunction{queueDeleteMax}{deleteMax}%
\SetKwFunction{queueInsert}{insert}%
\SetKwFunction{queueUpdate}{update}%
\SetKwFunction{setInsert}{insert}%
\SetKwFunction{setDeleteDominatedLabels}{deleteLabelsDominatedBy}%
\SetKw{KwNotDominates}{does not dominate}%

\definecolor{thesisblue}         {rgb}{0.337, 0.592, 0.773}
\definecolor{thesisblue-dark}    {rgb}{0.212, 0.471, 0.655}
\definecolor{thesisblue-light}   {rgb}{0.490, 0.710, 0.867}
\definecolor{thesisblue-vlight}  {rgb}{0.871, 0.937, 0.988}

\definecolor{thesisred}          {rgb}{0.776, 0.357, 0.396}
\definecolor{thesisred-dark}     {rgb}{0.663, 0.235, 0.271}
\definecolor{thesisred-light}    {rgb}{0.867, 0.502, 0.533}
\definecolor{thesisred-vlight}   {rgb}{0.984, 0.871, 0.878}

\definecolor{thesisgreen}        {rgb}{0.337, 0.765, 0.235}
\definecolor{thesisgreen-dark}   {rgb}{0.267, 0.612, 0.184}
\definecolor{thesisgreen-light}  {rgb}{0.443, 0.871, 0.341}
\definecolor{thesisgreen-vlight} {rgb}{0.710, 0.984, 0.643}

\definecolor{thesisyellow}       {rgb}{0.808, 0.659, 0.263}
\definecolor{thesisyellow-dark}  {rgb}{0.667, 0.529, 0.180}
\definecolor{thesisyellow-light} {rgb}{0.898, 0.765, 0.412}
\definecolor{thesisyellow-vlight}{rgb}{0.992, 0.933, 0.776}

\definecolor{black}              {rgb}{0,0,0}
\definecolor{black70}            {rgb}{0.30,0.30,0.30}
\definecolor{black50}            {rgb}{0.50,0.50,0.50}
\definecolor{black30}            {rgb}{0.70,0.70,0.70}
\definecolor{black15}            {rgb}{0.85,0.85,0.85}
\definecolor{black7}             {rgb}{0.93,0.93,0.93}

\colorlet{primarycolor-dark}      {thesisblue-dark}
\colorlet{primarycolor}           {thesisblue}
\colorlet{primarycolor-light}     {thesisblue-light}
\colorlet{primarycolor-vlight}    {thesisblue-vlight}
\colorlet{secondarycolor-dark}    {thesisred-dark}
\colorlet{secondarycolor}         {thesisred}
\colorlet{secondarycolor-light}   {thesisred-light}
\colorlet{secondarycolor-vlight}  {thesisred-vlight}

\usetikzlibrary{arrows}
\usetikzlibrary{decorations.text}
\usetikzlibrary{decorations.pathmorphing}
\usetikzlibrary{through}
\usetikzlibrary{intersections}
\tikzset{>=stealth'}

\tikzstyle{figure}=[font=\small]
\tikzstyle{vertex}=[circle, inner sep=0pt, minimum size=6pt, draw=black, fill=primarycolor-vlight]
\tikzstyle{bigvertex}=[vertex, minimum size=11pt]
\tikzstyle{textvertex}=[bigvertex, font=\footnotesize]
\tikzstyle{highlightedvertex}=[vertex, fill=secondarycolor-light]
\tikzstyle{edgelabel}=[inner sep=1pt, fill=white]
\tikzstyle{edges}=[shorten >=1pt, shorten <=1pt]
\tikzstyle{edgesbackground}=[line width=2pt, draw=white]

\tikzstyle{function}=[draw=primarycolor,very thick,line cap=rect]
\tikzstyle{functionterm}=[text=primarycolor,edgelabel]
\tikzstyle{discontinuity}=[circle, inner sep=0pt, minimum size=3.5pt,very thick]
\tikzstyle{discontinuityblank}=[discontinuity,draw=primarycolor,fill=white]
\tikzstyle{discontinuityfilled}=[discontinuity,draw=primarycolor,fill=primarycolor]
\tikzstyle{plotgrid}=[color=black30,very thin,dash pattern=on 2pt off 1pt]
\tikzstyle{chargingfunction}=[function,draw=secondarycolor]
\tikzstyle{helperline}=[black]
\tikzstyle{breakpoint}=[fill,primarycolor,circle,inner sep=0pt,minimum size=3.5pt,very thick]
\tikzstyle{plotpoint}=[fill=black,circle,inner sep=0pt,minimum size=3.5pt]
\tikzstyle{tangentline}=[shorten >=-20pt,shorten <= -15pt]

\tikzstyle{markSign} = [mark=*]
\tikzstyle{heuristicMarkSign} = [mark=o]
\tikzstyle{exactMarkSign} = [mark=*]
\tikzstyle{shortenLines} = [shorten <= 3.5pt,shorten >= 3.5pt]

\pgfdeclarelayer{background}
\pgfdeclarelayer{foreground}
\pgfsetlayers{background,main,foreground}

\pgfplotsset{
  /pgfplots/legend line with mark/.style={
    legend image code/.code={
        \draw[##1,no markers,shorten >= 3.5pt,line width=1pt]
         plot coordinates {
         (0cm,0cm)
         (0.3cm,0cm)
        };
        \draw[##1,only marks,markSign,line width=1pt]
         plot coordinates {
         (0.3cm,0cm)
        };
        \draw[##1,no markers,shorten <= 3.5pt,line width=1pt]
         plot coordinates {
         (0.3cm,0cm)
         (0.6cm,0cm)
        };
    }
  }
}

\pgfplotsset{
  /pgfplots/legend line with heuristic mark/.style={
    legend image code/.code={
        \draw[##1,no markers,shorten >= 3.5pt,line width=1pt]
         plot coordinates {
         (0cm,0cm)
         (0.3cm,0cm)
        };
        \draw[##1,only marks,heuristicMarkSign,line width=1pt]
         plot coordinates {
         (0.3cm,0cm)
        };
        \draw[##1,no markers,shorten <= 3.5pt,line width=1pt]
         plot coordinates {
         (0.3cm,0cm)
         (0.6cm,0cm)
        };
    }
  }
}

\pgfplotsset{
  /pgfplots/legend line with exact mark/.style={
    legend image code/.code={
        \draw[##1,no markers,shorten >= 3.5pt,line width=1pt]
         plot coordinates {
         (0cm,0cm)
         (0.3cm,0cm)
        };
        \draw[##1,only marks,exactMarkSign,line width=1pt]
         plot coordinates {
         (0.3cm,0cm)
        };
        \draw[##1,no markers,shorten <= 3.5pt,line width=1pt]
         plot coordinates {
         (0.3cm,0cm)
         (0.6cm,0cm)
        };
    }
  }
}

\newcommand{\eg}{e.\,g.\xspace}
\newcommand{\ie}{i.\,e.\xspace}
\newcommand{\wrt}{wrt.\xspace}
\newcommand{\cf}{c.\,f.\xspace}
\newcommand{\case}[1]{\emph{#1:}}

\DeclareMathOperator{\mergeop}{merge}
\DeclareMathOperator{\linkop}{link}

\DeclareMathOperator{\convertop}{convert}

\DeclareMathOperator{\boundop}{bound}

\DeclareRobustCommand{\ubar}[1]{\ensuremath{\underaccent{\bar}{#1}}}
\newcommand{\obar}[1]{\ensuremath{\bar{#1}}}

\newcommand{\graph}{\ensuremath{G}}

\newcommand{\vertices}{\ensuremath{V}}
\newcommand{\arcs}{\ensuremath{A}}

\newcommand{\vertex}{\ensuremath{v}}
\newcommand{\vertexa}{\ensuremath{u}}
\newcommand{\vertexb}{\ensuremath{v}}
\newcommand{\vertexc}{\ensuremath{w}}

\newcommand{\arc}{\ensuremath{a}}
\newcommand{\apath}{\ensuremath{P}}

\newcommand{\source}{\ensuremath{s}}

\newcommand{\target}{\ensuremath{t}}

\newcommand{\potential}{\ensuremath{\pi}}
\newcommand{\backwardgraph}{\ensuremath{\obar{\graph}}}
\newcommand{\backwardarcs}{\ensuremath{\obar{\arcs}}}

\newcommand{\functionspace}{\ensuremath{\mathbb{F}}}
\newcommand{\consumptionfunctionspace}{\ensuremath{\mathbb{F}}}

\newcommand{\interval}{\ensuremath{I}}
\newcommand{\convexlowerboundfunction}{\ensuremath{\varphi}}
\newcommand{\convexfunctionbreakpoints}{\ensuremath{\Phi}}

\newcommand{\drivingtimepotential}{\ensuremath{\potential_\drivingtimefunction}}

\newcommand{\convexweightfunction}{\ensuremath{\varphi}}
\newcommand{\convexpotential}{\ensuremath{\potential_\convexweightfunction}}

\newcommand{\costfunction}{\ensuremath{z}}

\newcommand{\drivingtimefunction}{\ensuremath{d}}
\newcommand{\consumptionfunction}{\ensuremath{c}}

\newcommand{\drivingtimelowerboundfunction}{\ensuremath{\ubar{d}}}
\newcommand{\consumptionlowerboundfunction}{\ensuremath{\ubar{c}}}

\newcommand{\distance}{\text{dist}}
\newcommand{\maxbattery}{\ensuremath{M}}
\newcommand{\soc}{\ensuremath{b}}

\newcommand{\socprofile}{\ensuremath{f}}

\newcommand{\labelset}{\ensuremath{L}}

\newcommand{\alabel}{\ensuremath{\ell}}
\DeclareMathOperator{\queuekey}{key}

\newcommand{\queue}{\ensuremath{Q}}

\newcommand{\reals}{\ensuremath{\mathbb{R}}}
\newcommand{\posreals}{\ensuremath{\reals_{\ge0}}}
\newcommand{\strictposreals}{\ensuremath{\reals_{>0}}}
\newcommand{\negreals}{\ensuremath{\reals_{\le0}}}
\newcommand{\naturals}{\ensuremath{\mathbb{N}}}
\newcommand{\bigO}{\ensuremath{\mathcal{O}}}
\newcommand{\NP}{\ensuremath{\mathcal{NP}}}

\newcommand{\pointa}{\ensuremath{p}}
\newcommand{\pointb}{\ensuremath{q}}
\newcommand{\pointc}{\ensuremath{r}}

\newcommand{\aline}{\ensuremath{\ell}}

\newcommand{\speedconsumptionfunction}{\ensuremath{h}}
\newcommand{\velocity}{\ensuremath{v}}
\newcommand{\slope}{\ensuremath{s}}
\newcommand{\length}{\ensuremath{\ell}}

\newcommand{\atime}{\ensuremath{x}}

\newcommand{\acoefficient}{\ensuremath{\lambda}}
\newcommand{\tcinner}{\ensuremath{g}}
\newcommand{\tcfunction}{\ensuremath{c}}
\newcommand{\tcfunctionupperbound}{\ensuremath{\obar{g}}}
\newcommand{\aparam}{\ensuremath{\alpha}}
\newcommand{\bparam}{\ensuremath{\beta}}
\newcommand{\cparam}{\ensuremath{\gamma}}
\newcommand{\aconst}{\ensuremath{\acoefficient_1}}
\newcommand{\bconst}{\ensuremath{\acoefficient_2}}
\newcommand{\cconst}{\ensuremath{\acoefficient_3}}
\newcommand{\leftintervalborder}{\ensuremath{\ubar{\tau}}}
\newcommand{\rightintervalborder}{\ensuremath{\obar{\tau}}}
\newcommand{\cvalue}{\ensuremath{\Delta}}
\newcommand{\cfunction}{\ensuremath{\Delta_{\operatorname{opt}}}}
\newcommand{\cfunctionsecond}{\ensuremath{\bar{\Delta}_{\operatorname{opt}}}}
\newcommand{\pparam}{\ensuremath{\lambda}}
\newcommand{\qparam}{\ensuremath{\mu}}
\newcommand{\tcfunctionderivative}{\tcinner'}
\newcommand{\tcfunctionleftderivative}{\tcinner'}
\newcommand{\tcfunctionrightderivative}{\tcinner'}

\newcommand{\tcfunctionshortcut}{\ensuremath{\tcfunction_{\smash{\arc}}}}
\newcommand{\leftbordershortcut}{\ensuremath{\leftintervalborder}}
\newcommand{\rightbordershortcut}{\ensuremath{\rightintervalborder}}
\newcommand{\intervalborder}{\ensuremath{\tau}}
\newcommand{\accumulateddrivingtime}{\ensuremath{\tau}}
\newcommand{\accumulatedsoc}{\ensuremath{\beta}}

\newcommand{\vertical}[1]{\begin{rotate}{90}#1\end{rotate}}

\newcommand{\instanceNoAux}{PG\xspace}
\newcommand{\instanceAux}{AX\xspace}
\newcommand{\instanceEurNoAux}{Eur-\instanceNoAux}
\newcommand{\instanceGerNoAux}{Ger-\instanceNoAux}
\newcommand{\instanceEurAux}{Eur-\instanceAux}
\newcommand{\instanceGerAux}{Ger-\instanceAux}

\newcommand{\firstleftboundarytermfull}{\ensuremath{\leftintervalborder_1 + \bparam_2 +  \left( \leftintervalborder_1 - \bparam_1 \right) \sqrt[3]{\frac{\aparam_2}{\aparam_1}}}}
\newcommand{\secondleftboundarytermfull}{\ensuremath{\leftintervalborder_2 + \bparam_1 +  \left( \leftintervalborder_2 - \bparam_2 \right) \sqrt[3]{\frac{\aparam_1}{\aparam_2}}}}
\newcommand{\firstrightboundarytermfull}{\ensuremath{\rightintervalborder_1 + \bparam_2 +  \left( \rightintervalborder_1 - \bparam_1 \right) \sqrt[3]{\frac{\aparam_2}{\aparam_1}}}}
\newcommand{\secondrightboundarytermfull}{\ensuremath{\rightintervalborder_2 + \bparam_1 + \left( \rightintervalborder_2 - \bparam_2 \right) \sqrt[3]{\frac{\aparam_1}{\aparam_2}}}}

\newcommand{\firstleftboundaryterm}{\ensuremath{\leftintervalborder^*_1}}
\newcommand{\secondleftboundaryterm}{\ensuremath{\leftintervalborder^*_2}}
\newcommand{\firstrightboundaryterm}{\ensuremath{\rightintervalborder^*_1}}
\newcommand{\secondrightboundaryterm}{\ensuremath{\rightintervalborder^*_2}}
\newcommand{\ExampleInftyOffset}{1}
\newcommand{\SoCProfileExampleInftyOffset}{1}

\newcommand{\PlotXAxisName}[3]
{
 \node[edgelabel,above=4pt] at (#1,#2) {#3};
}

\newcommand{\PlotXAxisNameFlipped}[3]
{
 \node[edgelabel,below=4pt] at (#1,#2) {#3};
}

\newcommand{\PlotYAxisName}[3]
{
 \node[edgelabel,right=2pt] at (#1,#2) {#3};
}

\newcommand{\DrawDiscontinuitySymbolWithLines}[1]
{
 \pgfgettransformentries{\mya}{\myb}{\myc}{\myd}{\mys}{\myt}
 \pgfmathsetmacro{\grow}{0.75/\mya}
 \begin{scope}[scale=\grow,rotate=#1]
  \clip (-0.15,-0.5/\grow) rectangle (0.15,0.5/\grow);
   
  \draw (0,0.06) -- (0,0.5/\grow);
  \draw (0,-0.06) -- (0,-0.5/\grow);

  \draw[line cap=rect,decorate,decoration={snake,segment length=6pt, amplitude=1.25pt}] (-0.435,0.06) -- (0.5,0.06);
  \draw[line cap=rect,decorate,decoration={snake,segment length=6pt, amplitude=1.25pt}] (-0.435,-0.06) -- (0.5,-0.06);  
 \end{scope}
} 

\newcommand{\ExampleDrawCoordinateSystem}[6]
{
 \begin{scope}
  \begin{pgfinterruptboundingbox}
   \clip (#5,#3-1) rectangle (#2+1,#4+1);
   \draw[plotgrid] (#1-1,#3) grid (#2,#4);
  \end{pgfinterruptboundingbox}
 \end{scope}
 
 \begin{pgfonlayer}{foreground}
  \begin{scope}
   \begin{pgfinterruptboundingbox}
    \clip (#5,#3-1) rectangle (#2+1,#4+1);
    \draw[color=black] (#1-1,#6) -- (#2,#6);
    \draw[shift={(#1-1,#6)},color=black] (0pt,2pt) -- (0pt,-2pt);  
   \end{pgfinterruptboundingbox}
  \end{scope}

  \draw[color=black] (#5,#3) -- (#5,#4);
 \end{pgfonlayer}
 
 \node at (#2,#4) {};
}

\newcommand{\ExampleDrawCoordinateSystemWithTicks}[6]
{
 \ExampleDrawCoordinateSystem{#1}{#2}{#3}{#4}{#5}{#6}

 \begin{pgfonlayer}{foreground}
  \foreach \x in {#1,...,#2}
   \draw[shift={(\x,#6)}] (0pt,2pt) -- (0pt,-2pt) node[below=2pt,inner sep=1pt,fill=white] {\vphantom{$\x$}};
  \foreach \x in {#1,...,#2}
   \draw[shift={(\x,#6)}] (0pt,2pt) -- (0pt,-2pt) node[below=2pt,inner sep=1pt] {$\x$};
   
  \foreach \y in {#3,...,#4}
   \draw[shift={(#5,\y)}] (2pt,0pt) -- (-2pt,0pt) node[left] {$\y$};
 \end{pgfonlayer}
}

\newcommand{\ExampleDrawCoordinateSystemWithPositiveInfty}[6]
{
 \ExampleDrawCoordinateSystemWithTicks{#1}{#2}{#3}{#4}{#5}{#6}

 \begin{pgfonlayer}{foreground}
  \draw[shift={(#5,#4+\ExampleInftyOffset)}] (2pt,0pt) -- (-2pt,0pt) node[left] {$\infty$};
  \draw[shift={(#5,#4+\ExampleInftyOffset)}] (0pt,2pt) -- (0pt,-2pt);
 
  \begin{scope}[shift={(#5,#4+0.5*\ExampleInftyOffset)}]
   \DrawDiscontinuitySymbolWithLines{0}
  \end{scope}
 \end{pgfonlayer}
}

\newcommand{\ExampleDrawCoordinateSystemWithNegativeInfty}[6]
{
 \ExampleDrawCoordinateSystemWithTicks{#1}{#2}{#3}{#4}{#5}{#6}

 \begin{pgfonlayer}{foreground}
  \draw[shift={(#5,#3-\ExampleInftyOffset)}] (2pt,0pt) -- (-2pt,0pt) node[left] {$-\infty$};
  \draw[shift={(#5,#3-\ExampleInftyOffset)}] (0pt,2pt) -- (0pt,-2pt);
 
  \begin{scope}[shift={(#5,#3-0.5*\ExampleInftyOffset)}]
   \DrawDiscontinuitySymbolWithLines{0}
  \end{scope}
 \end{pgfonlayer}
}

\newcommand{\ExampleDrawCoordinateSystemFlipX}[6]
{
 \ExampleDrawCoordinateSystem{#1}{#2}{#3}{#4}{#5}{#6}

 \begin{pgfonlayer}{foreground}
  \foreach \x in {#1,...,#2}
   \draw[shift={(\x,#6)}] (0pt,2pt) -- (0pt,-2pt) node[above=4pt,inner sep=1pt,fill=white] {\vphantom{$\x$}};
  \foreach \x in {#1,...,#2}
   \draw[shift={(\x,#6)}] (0pt,2pt) -- (0pt,-2pt) node[above=4pt,inner sep=1pt] {$\x$};
   
  \foreach \y in {#3,...,#4}
   \draw[shift={(#5,\y)}] (2pt,0pt) -- (-2pt,0pt) node[left] {$\y$};
 \end{pgfonlayer}
}

\newcommand{\ExampleDrawCoordinateSystemWithPositiveInftyFlipX}[6]
{
 \ExampleDrawCoordinateSystemFlipX{#1}{#2}{#3}{#4}{#5}{#6}

 \begin{pgfonlayer}{foreground}
  \draw[shift={(#5,#4+\ExampleInftyOffset)}] (2pt,0pt) -- (-2pt,0pt) node[left] {$\infty$};
  \draw[shift={(#5,#4+\ExampleInftyOffset)}] (0pt,2pt) -- (0pt,-2pt);
 
  \begin{scope}[shift={(#5,#4+0.5*\ExampleInftyOffset)}]
   \DrawDiscontinuitySymbolWithLines{0}
  \end{scope}
 \end{pgfonlayer}
}

\newcommand{\ExampleDrawCoordinateSystemWithPositiveInftyNoTicks}[6]
{
 \ExampleDrawCoordinateSystem{#1}{#2}{#3}{#4}{#5}{#6}

 \begin{pgfonlayer}{foreground}
  \draw[shift={(#5,#4+\ExampleInftyOffset)}] (2pt,0pt) -- (-2pt,0pt) node[left] {$\infty$};
  \draw[shift={(#5,#4+\ExampleInftyOffset)}] (0pt,2pt) -- (0pt,-2pt);
 
  \begin{scope}[shift={(#5,#4+0.5*\ExampleInftyOffset)}]
   \DrawDiscontinuitySymbolWithLines{0}
  \end{scope}
 \end{pgfonlayer}
}

\newcommand{\ConstrainedExampleDrawParallelEdgesDominance}
{
\begin{tikzpicture}[figure,xscale=2.5,yscale=2.1]

 \node[textvertex] (s) at (0,0) {$\source$};
 \node[textvertex] (v) at (1,0) {$\vertex$};
 \node[textvertex] (t) at (2,0) {$\target$};
  
 \path[draw opacity=0] (s) edge[bend right=50] node[edgelabel] {\phantom{$4,2$}} (t);
  
 \path[edges] (s) edge[bend left=50] node[edgelabel] {$10,7$} (v)
              (s) edge[dashed,black50] node[edgelabel] {$20,4$} (v)
              (s) edge[bend right=50] node[edgelabel] {$30,1$} (v)
              (v) edge[bend left=25] node[edgelabel] {$1,2$} (t)
              (v) edge[bend right=25] node[edgelabel] {$11,1$} (t);
               
 \node[inner sep=1pt,pin={[pin distance=15pt,pin edge={black}]90,align=center,inner sep=1pt:$(0,10)$}] at (s.north) {};
 \node[inner sep=1pt,pin={[pin distance=15pt,pin edge={black}]90,align=center,inner sep=1pt:$(10,3)$\\$(30,9)$}] at (v.north) {}; 
 \node[inner sep=1pt,pin={[pin distance=15pt,pin edge={black}]90,align=center,inner sep=1pt:$(11,1)$\\$(21,2)$\\$(31,7)$\\$(41,8)$}] at (t.north) {};
               
\end{tikzpicture}
}

\newcommand{\ConstrainedExampleDrawParallelEdgesContraction}
{
\begin{tikzpicture}[figure,xscale=2.5,yscale=2.1]

 \node[textvertex] (s) at (0,0) {$\source$};
 \node[textvertex] (v) at (1,0) {$\vertex$};
 \node[textvertex] (t) at (2,0) {$\target$};
  
 \path[edges] (s) edge[bend left=25] node[edgelabel] {$1,2$} (v)     
              (s) edge[bend right=25] node[edgelabel] {$2,1$} (v)
              (v) edge[bend left=25] node[edgelabel] {$1,2$} (t)
              (v) edge[bend right=25] node[edgelabel] {$2,1$} (t)
              (s) edge[dashed,black50,bend left=50] node[edgelabel] {$2,4$} (t)
              (s) edge[dashed,black50,bend right=50] node[edgelabel] {$4,2$} (t);

 \node[inner sep=1pt,pin={[pin distance=15pt,pin edge={black}]90,align=center,inner sep=1pt:$(0,5)$}] at (s.north) {};
 \node[inner sep=1pt,pin={[pin distance=15pt,pin edge={black}]90,align=center,inner sep=1pt:$(2,1)$\\$(3,2)$\\$(4,3)$}] at (t.north) {};
               
\end{tikzpicture}
}

\newcommand{\ConstrainedExampleSimpleTradeoffScale}{0.7}

\newcommand{\ConstrainedExampleDrawNegativeTradeoffFunction}
{
\begin{tikzpicture}[figure,scale=\ConstrainedExampleSimpleTradeoffScale]
 \ExampleDrawCoordinateSystemWithPositiveInftyFlipX{1}{3}{-4}{0}{0}{0}
 
 \draw[draw opacity=0,shift={(0,-4)}] (0pt,2pt) -- (0pt,-2pt) node[below=2pt,inner sep=1pt] {\phantom{1}};
 
 \begin{scope}
  \begin{pgfinterruptboundingbox}
   \clip (0,-4) rectangle (3,2+\ExampleInftyOffset);
   \draw [function] (0,\ExampleInftyOffset) -- (1,\ExampleInftyOffset);
   \draw [function] (1,-1) -- (2,-3) -- (3,-3);
  \end{pgfinterruptboundingbox}
 \end{scope}

 \begin{pgfonlayer}{foreground}
  \node[discontinuityblank] (d1) at (1,\ExampleInftyOffset) {};
  \node[discontinuityfilled] (d2) at (1,-1) {};
 \end{pgfonlayer}
 
 \PlotXAxisNameFlipped{3}{0}{$\atime_2$}
 \PlotYAxisName{0}{-4}{$\tcfunction_{(\vertexb,\vertexc)}(\atime_2)$}
\end{tikzpicture}
}

\newcommand{\ConstrainedExampleDrawPositiveTradeoffFunction}
{
\begin{tikzpicture}[figure,scale=\ConstrainedExampleSimpleTradeoffScale]
 \ExampleDrawCoordinateSystemWithPositiveInfty{1}{3.5}{0}{4}{0.5}{0}
 
 \begin{scope}
  \begin{pgfinterruptboundingbox}
   \clip (0,0) rectangle (3.5,5+\ExampleInftyOffset);
   \draw [function] (0.5,4+\ExampleInftyOffset) -- (1,4+\ExampleInftyOffset);
   \draw [function] (1,2) -- (3,1) -- (4,1);
  \end{pgfinterruptboundingbox}
 \end{scope}

 \begin{pgfonlayer}{foreground}
  \node[discontinuityblank] (d1) at (1,4+\ExampleInftyOffset) {};
  \node[discontinuityfilled] (d2) at (1,2) {};
 \end{pgfonlayer}
 
 \PlotXAxisName{3}{0}{$\atime_1$}
 \PlotYAxisName{0.5}{4}{$\tcfunction_{(\vertexa,\vertexb)}(\atime_1)$} 
\end{tikzpicture}
}

\newcommand{\ConstrainedExampleDrawTradeoffFunctionPath}
{
\begin{tikzpicture}[figure,scale=\ConstrainedExampleSimpleTradeoffScale]
 \fill[draw=black30,fill=black7] (2,1) -- (3,-1) -- (5,-2) -- (4,0) -- cycle;

 \ExampleDrawCoordinateSystemWithPositiveInftyNoTicks{2}{6}{-3}{1}{1}{0}

 \draw[draw opacity=0,shift={(0,-3)}] (0pt,2pt) -- (0pt,-2pt) node[below=2pt,inner sep=1pt] {\phantom{1}};
 
 \begin{pgfonlayer}{foreground}
  \draw[shift={(2,0)}] (0pt,2pt) -- (0pt,-2pt) node[below=2pt,inner sep=1pt,fill=white] {$2$};
  \draw[shift={(3,0)}] (0pt,2pt) -- (0pt,-2pt) node[below=2pt,inner sep=1pt,minimum size=2pt,rectangle,fill=black7] {\vphantom{3}} node[below=2pt,inner sep=1pt,fill=none] {$3$};
  \draw[shift={(4,0)}] (0pt,2pt) -- (0pt,-2pt) node[below=2pt,inner sep=1pt,minimum size=2pt,rectangle,fill=black7] {\vphantom{4}} node[below=2pt,inner sep=1pt,fill=none] {$4$};
  \draw[shift={(5,0)}] (0pt,2pt) -- (0pt,-2pt) node[below=2pt,inner sep=1pt,fill=white] {$5$};
  \draw[shift={(6,0)}] (0pt,2pt) -- (0pt,-2pt) node[below=2pt,inner sep=1pt,fill=white] {$6$};
   
  \foreach \y in {-3,...,1}
   \draw[shift={(1,\y)}] (2pt,0pt) -- (-2pt,0pt) node[left] {$\y$};
 \end{pgfonlayer}
 
 \begin{scope}
  \begin{pgfinterruptboundingbox}
   \clip (1,-3) rectangle (6,3+\ExampleInftyOffset);
   \draw [function] (1,1+\ExampleInftyOffset) -- (2,1+\ExampleInftyOffset);
   \draw [function] (2,1) -- (3,-1) -- (5,-2) -- (6,-2);
  \end{pgfinterruptboundingbox}
 \end{scope}
 
 \begin{pgfonlayer}{foreground}
  \node[discontinuityblank] (d1) at (2,1+\SoCProfileExampleInftyOffset) {};
  \node[discontinuityfilled] (d2) at (2,1) {};
 \end{pgfonlayer}
 
 \PlotXAxisName{6}{0}{$\atime$}
 \PlotYAxisName{1}{-3}{$\tcfunction_{\apath}(\atime=\atime_1+\atime_2)$}
\end{tikzpicture}
}

\newcommand{\ConstrainedExampleDrawBivariateTimeSoCTradeoffFunction}
{
\begin{tikzpicture}[figure,scale=0.7]
 \ExampleDrawCoordinateSystemWithNegativeInfty{2}{5}{0}{4}{1}{0}
 
 \colorlet{color3}{thesisblue-dark!80!black}
 \colorlet{color4}{thesisblue!50!thesisblue-dark}
 \colorlet{color2}{thesisblue-light!50!thesisblue}
 \colorlet{color1}{thesisblue-vlight!10!thesisblue-light}

 \PlotXAxisName{5}{0}{$\atime$}
 \PlotYAxisName{1}{4.3}{$\socprofile_\apath(\atime,\soc)$}
 
 \begin{scope}
  \begin{pgfinterruptboundingbox}
   \clip (1,-3) rectangle (5,5);
   \draw [function,draw=color1] (0,-\ExampleInftyOffset) -- (4,-\ExampleInftyOffset);
   \draw [function,draw=color1] (4,1) -- (5,3);
   \draw [function,draw=color2] (0,-\ExampleInftyOffset) -- (2,-\ExampleInftyOffset);
   \draw [function,draw=color3] (2,3) -- (2.5,4) -- (5,4);
   \draw [function,draw=color4] (2,2) -- (3,4) -- (5,4);
   \draw [function,draw=color2] (2,1) -- (3,3) -- (5,4);
  \end{pgfinterruptboundingbox}
 \end{scope}
 
 \begin{pgfonlayer}{foreground}
  \node[discontinuityblank,draw=color2] (d1) at (2,-\ExampleInftyOffset) {};
  \node[discontinuityblank,draw=color1] (d2) at (4,-\ExampleInftyOffset) {};
  \node[discontinuityfilled,color2] (d3) at (2,1) {};
  \node[discontinuityfilled,color4] (d4) at (2,2) {};
  \node[discontinuityfilled,color3] (d5) at (2,3) {};
  \node[discontinuityfilled,color1] (d6) at (4,1) {};   
 \end{pgfonlayer}
 
 \begin{scope}[shift={(5.5,2)}]
  \draw[function,draw=color3,line width=2pt]  (0,1.8) -- (0.5,1.8) node [right] {$\soc=4$};
  \draw[function,draw=color4,line width=2pt]   (0,1.2) -- (0.5,1.2) node [right] {$\soc=3$};
  \draw[function,draw=color2,line width=2pt] (0,0.6) -- (0.5,0.6) node [right] {$\soc=2$};
  \draw[function,draw=color1,line width=2pt]    (0,0) -- (0.5,0) node [right] {$\soc=1$};
 \end{scope}
\end{tikzpicture}
}

\newcommand{\ConstrainedExampleDrawBivariateSoCSoCTradeoffFunction}
{
\begin{tikzpicture}[figure,scale=0.7]
 \ExampleDrawCoordinateSystemWithNegativeInfty{1}{4}{0}{4}{0}{0}
 
 \colorlet{color1}{thesisred-dark!80!black}
 \colorlet{color2}{thesisred!50!thesisred-dark}
 \colorlet{color3}{thesisred-light!50!thesisred}
 \colorlet{color4}{thesisred-vlight!10!thesisred-light}

 \PlotXAxisName{4}{0}{$\soc$}
 \PlotYAxisName{0}{4.3}{$\socprofile_\apath(\atime,\soc)$}
 
 \begin{scope}
  \begin{pgfinterruptboundingbox}
   \clip (0,-2) rectangle (4,5);
   \draw [function,draw=color4] (0,-\ExampleInftyOffset) -- (2,-\ExampleInftyOffset);
   \draw [function,draw=color3] (0,-\ExampleInftyOffset) -- (1.5,-\ExampleInftyOffset);
   \draw [function,draw=color2] (0,-\ExampleInftyOffset) -- (1,-\ExampleInftyOffset);
   \draw [function,draw=color4] (2,1) -- (4,3);
   \draw [function,draw=color1] (1,3) -- (2,4) -- (4,4);
   \draw [function,draw=color2] (1,1) -- (2,3.5) -- (2.5,4) -- (4,4);   
   \draw [function,draw=color3] (1.5,1) -- (2,3) -- (3,4) -- (4,4);
  \end{pgfinterruptboundingbox}
 \end{scope}
 
 \begin{pgfonlayer}{foreground}
  \node[discontinuityblank,draw=color2] (d1) at (1,-\ExampleInftyOffset) {};
  \node[discontinuityblank,draw=color3] (d2) at (1.5,-\ExampleInftyOffset) {};
  \node[discontinuityblank,draw=color4] (d3) at (2,-\ExampleInftyOffset) {};
  \node[discontinuityfilled,color4] (d4) at (2,1) {};
  \node[discontinuityfilled,color3] (d5) at (1.5,1) {};
  \node[discontinuityfilled,color2] (d6) at (1,1) {};
  \node[discontinuityfilled,color1] (d7) at (1,3) {};   
 \end{pgfonlayer}
 
 \begin{scope}[shift={(4.5,2)}]
  \draw[function,draw=color1,line width=2pt]  (0,1.8) -- (0.5,1.8) node [right] {$\atime=5$}; 
  \draw[function,draw=color2,line width=2pt]   (0,1.2) -- (0.5,1.2) node [right] {$\atime=4$};
  \draw[function,draw=color3,line width=2pt] (0,0.6) -- (0.5,0.6) node [right] {$\atime=3$};
  \draw[function,draw=color4,line width=2pt]    (0,0) -- (0.5,0) node [right] {$\atime=2$};
 \end{scope}
\end{tikzpicture}
}
\newcommand{\TradeoffFunctionExampleScale}{0.7}

\newcommand{\ConstrainedExampleDrawRealisticTradeoffFunction}
{
\begin{tikzpicture}[figure,scale=\TradeoffFunctionExampleScale]
 \ExampleDrawCoordinateSystemWithPositiveInfty{1}{8}{0}{5}{0}{0}
 
 \PlotXAxisName{8}{0}{$\atime$}
 \PlotYAxisName{0}{5}{$\tcfunction(\atime)$}
 
 \PlotXAxisName{2}{0}{$\leftintervalborder$}
 \PlotXAxisName{6}{0}{$\rightintervalborder$}
 
 \begin{scope}
  \begin{pgfinterruptboundingbox}
   \clip (0,0) rectangle (8,6+\ExampleInftyOffset);
   \draw[function] (0,5+\ExampleInftyOffset) -- (2,5+\ExampleInftyOffset);
   \draw[function,smooth,samples=100,domain=2:6,line cap=round] plot(\x, {3/((\x-1)*(\x-1))+1});
   \draw[function,line cap=round] (6,1.12) -- (8,1.12);
  \end{pgfinterruptboundingbox}
 \end{scope}
  
 \begin{pgfonlayer}{foreground}
  \node[discontinuityblank] (d1) at (2,5+\ExampleInftyOffset) {};
  \node[discontinuityfilled] (d2) at (2,4) {};
 \end{pgfonlayer}
 
 \node[functionterm] (f) at (4.4,2.3) {$\frac{3}{(x-1)^2}+1$};
\end{tikzpicture}
}

\newcommand{\ConstrainedExampleDrawRealisticTradeoffFunctionFirstEdge}
{
\begin{tikzpicture}[figure,scale=\TradeoffFunctionExampleScale]
 \ExampleDrawCoordinateSystemWithPositiveInftyFlipX{2}{5}{-1}{4}{1}{0}
 
 \draw[draw opacity=0,shift={(0,-1)}] (0pt,2pt) -- (0pt,-2pt) node[below=2pt,inner sep=1pt] {\phantom{1}};
 
 \PlotXAxisNameFlipped{5}{0.1}{$\atime_1$}
 \PlotYAxisName{1}{4}{$\tcfunction_1(\atime_1)$}
  
 \begin{scope}
  \begin{pgfinterruptboundingbox}
   \clip (1,-1) rectangle (5,6+\ExampleInftyOffset);
   \draw[function] (1,4+\ExampleInftyOffset) -- (2,4+\ExampleInftyOffset);
   \draw[function,smooth,samples=100,domain=2:4,line cap=round] plot(\x, {4/((\x-1)*(\x-1))-1});
   \draw[function,line cap=round] (4,-5/9) -- (5,-5/9);
  \end{pgfinterruptboundingbox}
 \end{scope}
 
 \begin{pgfonlayer}{foreground}
  \node[discontinuityblank] (d1) at (2,4+\ExampleInftyOffset) {};
  \node[discontinuityfilled] (d2) at (2,3) {};
 \end{pgfonlayer}
 
 \node[functionterm] (f) at (3.6,1.8) {$\frac{4}{(x-1)^2}-1$};
\end{tikzpicture}
}

\newcommand{\ConstrainedExampleDrawRealisticTradeoffFunctionSecondEdge}
{
\begin{tikzpicture}[figure,scale=\TradeoffFunctionExampleScale]
 \ExampleDrawCoordinateSystemWithPositiveInfty{2}{5.5}{-1}{4}{1.5}{0}
 
 \draw[draw opacity=0,shift={(0,-1)}] (0pt,2pt) -- (0pt,-2pt) node[below=2pt,inner sep=1pt] {\phantom{1}};

 \PlotXAxisName{5}{-0.1}{$\atime_2$}
 \PlotYAxisName{1.5}{4}{$\tcfunction_2(\atime_2)$}
  
 \begin{scope}
  \begin{pgfinterruptboundingbox}
   \clip (1.5,-1) rectangle (5.5,5+\ExampleInftyOffset);
   \draw[function] (1.5,4+\ExampleInftyOffset) -- (2,4+\ExampleInftyOffset);
   \draw[function,smooth,samples=100,domain=2:5,line cap=round] plot(\x, {0.5/((\x-1)*(\x-1))+1});
   \draw[function,line cap=round] (5,33/32) -- (5.5,33/32);
  \end{pgfinterruptboundingbox}
 \end{scope}
  
 \begin{pgfonlayer}{foreground}
  \node[discontinuityblank] (d1) at (2,4+\ExampleInftyOffset) {};
  \node[discontinuityfilled] (d2) at (2,1.5) {};
 \end{pgfonlayer}
 
 \node[functionterm] (f) at (3.8,1.7) {$\frac{0.5}{(x-1)^2}+1$};
\end{tikzpicture}
}

\newcommand{\ConstrainedExampleDrawRealisticTradeoffFunctionLinkedFunction}
{
\begin{tikzpicture}[figure,scale=\TradeoffFunctionExampleScale]
 \ExampleDrawCoordinateSystemWithPositiveInfty{4}{9.5}{0}{5}{3}{0}
 
 \PlotXAxisName{8.5}{-0.2}{$\atime=\atime_1+\atime_2$}
 \PlotYAxisName{3}{5}{$\tcfunction(\atime)$}
 
 \node[functionterm,font=\tiny] (f1) at (4,0.8) {$\frac{4}{(x-3)^2}+\frac{1}{2}$};
 \node[functionterm,font=\tiny] (f2) at (5.75,1.74) {$\frac{13.5}{(x-2)^2}$};
 \node[functionterm,font=\tiny] (f3) at (8,1) {$\frac{0.5}{(x-5)^2}+\frac{4}{9}$};
  
  \path (6.5,4.0) -- (9.0,4.0) node [edgelabel,pos=0.6,below,sloped,text=secondarycolor] {$x_1=\cfunction(\atime)$};
  
 \begin{scope}
  \begin{pgfinterruptboundingbox}
   \clip (4,0) rectangle (9,6+\ExampleInftyOffset);
   \draw[function,thick,draw=secondarycolor-light,text=secondarycolor] (4.0,2.0) -- (5.0,3.0) -- (6.5,4.0) -- (9.0,4.0);
  \end{pgfinterruptboundingbox}
 \end{scope} 
  
 \begin{scope}
  \begin{pgfinterruptboundingbox}
   \clip (3,0) rectangle (9.5,6+\ExampleInftyOffset);
   \draw[function,draw=primarycolor-dark] (3,5+\ExampleInftyOffset) -- (4,5+\ExampleInftyOffset);
  \end{pgfinterruptboundingbox}
 \end{scope}
 \begin{scope}
  \begin{pgfinterruptboundingbox}
   \clip (3,0) rectangle (5,6+\ExampleInftyOffset);
   \draw[function,draw=primarycolor-light,smooth,samples=100,domain=4:5,line cap=round] plot(\x, {4/((\x-3)*(\x-3))+0.5});
  \end{pgfinterruptboundingbox}
 \end{scope}
 \begin{scope}
  \begin{pgfinterruptboundingbox}
   \clip (5,0) rectangle (6.5,6+\ExampleInftyOffset);
   \draw[function,draw=primarycolor-dark,smooth,samples=100,domain=5:6.5,line cap=round] plot(\x, {13.5/((\x-2)*(\x-2))});
  \end{pgfinterruptboundingbox}
 \end{scope}
 \begin{scope}
  \begin{pgfinterruptboundingbox}
   \clip (6.5,0) rectangle (9,6+\ExampleInftyOffset);
   \draw[function,draw=primarycolor-light,smooth,samples=100,domain=6.5:9,line cap=round] plot(\x, {0.5/((\x-5)*(\x-5))+4/9});
  \end{pgfinterruptboundingbox}
 \end{scope}
 \begin{scope}
  \begin{pgfinterruptboundingbox}
   \clip (9,0) rectangle (9.5,6+\ExampleInftyOffset);
   \draw[function,line cap=round] (9,33/32-5/9) -- (9.5,33/32-5/9);
  \end{pgfinterruptboundingbox}
 \end{scope}

 \draw (5,0) -- (5,5) node[above,edgelabel] {$\secondleftboundaryterm$};
 \draw (6.5,0) -- (6.5,5) node[above,edgelabel] {$\firstrightboundaryterm$};

 \begin{pgfonlayer}{foreground}
  \node[discontinuityblank,draw=primarycolor-dark] (d1) at (4,5+\ExampleInftyOffset) {};
  \node[discontinuityfilled,primarycolor-light] (d2) at (4,4.5) {};
 \end{pgfonlayer}
\end{tikzpicture}
}

\newcommand{\ConstrainedExampleDrawRealisticTradeoffFunctionLinkedFunctionWithCFunctions}
{
\begin{tikzpicture}[figure,scale=\TradeoffFunctionExampleScale]
 \ExampleDrawCoordinateSystemWithPositiveInfty{4}{9.5}{0}{5}{3}{0}
 
 \PlotXAxisName{9}{-0.15}{$\atime$}
 \PlotYAxisName{3}{5}{$\tcfunction(\atime)$}
  
 \begin{scope}
  \begin{pgfinterruptboundingbox}
   \clip (4,0) rectangle (9,6+\ExampleInftyOffset);
   \draw[function,thick,draw=secondarycolor-light,text=secondarycolor] (4.0,2.0) -- (5.0,3.0) -- (6.5,4.0) -- (9.0,4.0) node [edgelabel,pos=0.34,above=1pt,sloped] {$\cfunction(\atime)$};
   \draw[function,thick,draw=secondarycolor-light,text=secondarycolor] (4.0,2.0) -- (5.0,2.0) -- (6.5,2.5) -- (9.0,5.0) node [edgelabel,pos=0.3,below=1.5pt,sloped] {$\atime-\cfunction(\atime)$};
   \draw[function,thick,draw=secondarycolor-light] (6.5,4.0) -- (9.0,4.0);
  \end{pgfinterruptboundingbox}
 \end{scope} 
  
 \begin{scope}
  \begin{pgfinterruptboundingbox}
   \clip (3,0) rectangle (9.5,6+\ExampleInftyOffset);
   \draw[function,draw=primarycolor-dark] (3,5+\ExampleInftyOffset) -- (4,5+\ExampleInftyOffset);
  \end{pgfinterruptboundingbox}
 \end{scope}
 \begin{scope}
  \begin{pgfinterruptboundingbox}
   \clip (3,0) rectangle (5,6+\ExampleInftyOffset);
   \draw[function,draw=primarycolor-light,smooth,samples=100,domain=4:5,line cap=round] plot(\x, {4/((\x-3)*(\x-3))+0.5});
  \end{pgfinterruptboundingbox}
 \end{scope}
 \begin{scope}
  \begin{pgfinterruptboundingbox}
   \clip (5,0) rectangle (6.5,6+\ExampleInftyOffset);
   \draw[function,draw=primarycolor-dark,smooth,samples=100,domain=5:6.5,line cap=round] plot(\x, {13.5/((\x-2)*(\x-2))});
  \end{pgfinterruptboundingbox}
 \end{scope}
 \begin{scope}
  \begin{pgfinterruptboundingbox}
   \clip (6.5,0) rectangle (9,6+\ExampleInftyOffset);
   \draw[function,draw=primarycolor-light,smooth,samples=100,domain=6.5:9,line cap=round] plot(\x, {0.5/((\x-5)*(\x-5))+4/9});
  \end{pgfinterruptboundingbox}
 \end{scope}
 \begin{scope}
  \begin{pgfinterruptboundingbox}
   \clip (9,0) rectangle (9.5,6+\ExampleInftyOffset);
   \draw[function,line cap=round] (9,33/32-5/9) -- (9.5,33/32-5/9);
  \end{pgfinterruptboundingbox}
 \end{scope}

 \draw (5,0) -- (5,5) node[above,edgelabel] {$\secondleftboundaryterm$};
 \draw (6.5,0) -- (6.5,5) node[above,edgelabel] {$\firstrightboundaryterm$};

 \begin{pgfonlayer}{foreground}
  \node[discontinuityblank,draw=primarycolor-dark] (d1) at (4,5+\ExampleInftyOffset) {};
  \node[discontinuityfilled,primarycolor-light] (d2) at (4,4.5) {};
 \end{pgfonlayer}
\end{tikzpicture}
}

\newcommand{\ConstrainedExampleDrawTradeoffFunctionDominance}
{
\begin{tikzpicture}[figure,scale=1.05*\TradeoffFunctionExampleScale]
 \fill[color=black7] plot[domain=1.2:5.2](\x, {1.5/((\x-0.2)*(\x-0.2))+2.1}) -- (6,2.16) -- (6,5) -- (1.2,5) -- cycle;
 \fill[color=black7] plot[domain=2.5:4](\x, {4/((\x-1.5)*(\x-1.5))+0.2}) -- (6,0.84) -- (6,5) -- (2.5,5) -- cycle;

 \ExampleDrawCoordinateSystemWithPositiveInfty{1}{6}{0}{5}{0}{0}
 
 \PlotXAxisName{6}{0}{$\atime$}
 \PlotYAxisName{0}{5}{$\tcfunction(\atime)$}
  
 \begin{scope}
  \begin{pgfinterruptboundingbox}
   \clip (0,0) rectangle (6,6+\ExampleInftyOffset);
   \draw[function,secondarycolor] (0,5+\ExampleInftyOffset) -- (2.5,5+\ExampleInftyOffset);
   \draw[function] (0,5+\ExampleInftyOffset) -- (2,5+\ExampleInftyOffset);
   \draw[function,secondarycolor] (0,5+\ExampleInftyOffset) -- (1.2,5+\ExampleInftyOffset);
   
   \draw[function,secondarycolor,smooth,samples=100,domain=1.2:5.2,line cap=round] plot(\x, {1.5/((\x-0.2)*(\x-0.2))+2.1});
   \draw[function,secondarycolor,line cap=round] (5.2,2.16) -- (6,2.16);
   
   \draw[function,secondarycolor,smooth,samples=100,domain=2.5:4,line cap=round] plot(\x, {4/((\x-1.5)*(\x-1.5))+0.2});
   \draw[function,secondarycolor,line cap=round] (4,0.84) -- (6,0.84);
   
   \draw[function,smooth,samples=100,domain=2:6,line cap=round] plot(\x, {2.5/((\x-1)*(\x-1))+1.8});
  \end{pgfinterruptboundingbox}
 \end{scope}
  
 \begin{pgfonlayer}{foreground}
  \node[discontinuityblank,draw=secondarycolor] (d1) at (1.2,5+\ExampleInftyOffset) {};
  \node[discontinuityblank] (d2) at (2,5+\ExampleInftyOffset) {};
  \node[discontinuityblank,draw=secondarycolor] (d3) at (2.5,5+\ExampleInftyOffset) {};
  \node[discontinuityfilled,secondarycolor] (d4) at (1.2,3.6) {};
  \node[discontinuityfilled,secondarycolor] (d5) at (2.5,4.2) {};
  \node[discontinuityfilled] (d6) at (2,4.3) {};  
 \end{pgfonlayer}
\end{tikzpicture}
}

\newcommand{\ConstrainedExampleDrawTradeoffFunctionPartialDominance}
{
\begin{tikzpicture}[figure,scale=1.05*\TradeoffFunctionExampleScale]
 \fill[color=black7] plot[domain=1.5:4](\x, {1.5/((\x-0.5)*(\x-0.5))+2}) -- (6,2.12245) -- (6,5) -- (1.5,5) -- cycle;

 \ExampleDrawCoordinateSystemWithPositiveInfty{1}{6}{0}{5}{0}{0}
 
 \PlotXAxisName{6}{0}{$\atime$}
 \PlotYAxisName{0}{5}{$\tcfunction(\atime)$}
 
 \draw (2.5,0) -- (2.5,5) node[above,edgelabel] {$\leftintervalborder$};
 \draw (3.23567,0) -- (3.23567,5) node[above,edgelabel] {\hphantom{$\leftintervalborder$}\clap{$\leftintervalborder^*$}};
 \draw (5,0) -- (5,5) node[above,edgelabel] {$\rightintervalborder$};
  
 \begin{scope}
  \begin{pgfinterruptboundingbox}
   \clip (0,0) rectangle (6,6+\ExampleInftyOffset);
   \draw[function] (0,5+\ExampleInftyOffset) -- (2.5,5+\ExampleInftyOffset);
   \draw[function,secondarycolor] (0,5+\ExampleInftyOffset) -- (1.5,5+\ExampleInftyOffset);
   
   \draw[function,secondarycolor,smooth,samples=100,domain=1.5:4,line cap=round] plot(\x, {1.5/((\x-0.5)*(\x-0.5))+2});
   \draw[function,secondarycolor,line cap=round] (4,2.12245) -- (6,2.12245);
      
   \draw[function,smooth,samples=100,domain=2.5:5,line cap=round] plot(\x, {6/((\x-1)*(\x-1))+1});
   \draw[function,line cap=round] (5,1.375) -- (6,1.375);
  \end{pgfinterruptboundingbox}
 \end{scope}
  
  \begin{pgfonlayer}{foreground}
   \node[discontinuityblank,draw=secondarycolor] (d1) at (1.5,5+\ExampleInftyOffset) {};
   \node[discontinuityblank] (d2) at (2.5,5+\ExampleInftyOffset) {};
   \node[discontinuityfilled,secondarycolor] (d3) at (1.5,3.5) {};
   \node[discontinuityfilled] (d4) at (2.5,3.66667) {};
  \end{pgfonlayer} 
\end{tikzpicture}
}

\newcommand{\ConstrainedExampleDrawTradeoffFunctionSimplification}
{
\begin{tikzpicture}[figure,scale=\TradeoffFunctionExampleScale]
 \ExampleDrawCoordinateSystemWithTicks{1}{10}{1}{7}{0}{1}
 
 \PlotXAxisName{10}{1}{$\tcfunction$}
 \PlotYAxisName{0}{7}{$\atime$}
 
 \PlotYAxisName{0}{2}{$\leftintervalborder$}
 \PlotYAxisName{0}{6}{$\rightintervalborder$}
  
 \coordinate (cut) at (2.076923,2.76923);
 \coordinate (fval) at (2.076923,4.163332);
 \draw[tangentline] (1,6) -- (cut);
 \draw[tangentline] (9,2) -- (cut);
 \draw[dashed] (cut) -- (fval) node[pos=0.4,right] {$\varepsilon$};
  
 \draw[function,smooth,samples=100,domain=1:9,line cap=round] plot(\x, {6/sqrt(\x)});
 
 \node[edgelabel] (f) at (5.7,3.1) {$\tcfunction_{(\vertexa,\vertexb)}(\atime)$};
 
 \node[edgelabel] (l1) at (2.6,1.4) {$\aline_1$};
 \node[edgelabel] (l2) at (0.6,2.9) {$\aline_2$};
 
 \node[plotpoint,label={[edgelabel,label distance=2pt]right:$\pointa$}] (p) at (1,6) {};
 \node[plotpoint,label={[edgelabel,label distance=2pt]above:$\pointb$}] (q) at (9,2) {};
 \node[plotpoint,label={[label distance=-1pt]225:$\pointc$}] (r) at (cut) {};
 \node[plotpoint,label={[label distance=-1pt]45:$\pointc'$}] (s) at (fval) {};
\end{tikzpicture}
}

\newcommand{\ConstrainedExampleDrawRealisticSoCFunctionFirstEdge}
{
\begin{tikzpicture}[figure,scale=\TradeoffFunctionExampleScale]
 \ExampleDrawCoordinateSystemWithPositiveInfty{2}{5.5}{2}{8}{1.5}{2}
 
 \draw[draw opacity=0,shift={(1.5,3)}] (-2pt,0pt) -- (2pt,0pt) node[left] {\phantom{$-1$}};
 
 \PlotXAxisName{5}{2}{$\atime_1$}
 \PlotYAxisName{1.5}{8}{$\tcfunction_1(\atime_1)$}
  
 \begin{scope}
  \begin{pgfinterruptboundingbox}
   \clip (1.5,1) rectangle (5.5,9+\ExampleInftyOffset);
   \draw[function] (1.5,8+\ExampleInftyOffset) -- (2,8+\ExampleInftyOffset);
   \draw[function,smooth,samples=100,domain=2:5,line cap=round] plot(\x, {4/((\x)*(\x))+41/9});
   \draw[function,line cap=round] (5,4.71556) -- (5.5,4.71556);
  \end{pgfinterruptboundingbox}
 \end{scope}
  
  \begin{pgfonlayer}{foreground}
   \node[discontinuityblank] (d1) at (2,8+\ExampleInftyOffset) {};
   \node[discontinuityfilled] (d2) at (2,50/9) {};
  \end{pgfonlayer}
  
  \node[functionterm] (f) at (3.8,5.8) {$\frac{4}{x^2}+\frac{41}{9}$};
\end{tikzpicture}
}

\newcommand{\ConstrainedExampleDrawRealisticSoCFunctionSecondEdge}
{
\begin{tikzpicture}[figure,scale=\TradeoffFunctionExampleScale]
 \ExampleDrawCoordinateSystemWithPositiveInftyFlipX{2}{5.5}{-5}{1}{1.5}{0}
 
 \draw[draw opacity=0,shift={(1,-5)}] (0pt,-2pt) -- (0pt,2pt) node[below] {\phantom{$1$}};
 
 \PlotXAxisNameFlipped{5}{0}{$\atime_2$}
 \PlotYAxisName{1.5}{-5}{$\tcfunction_2(\atime_2)$}
  
 \begin{scope}
  \begin{pgfinterruptboundingbox}
   \clip (1,-5) rectangle (5.5,2+\ExampleInftyOffset);
   \draw[function] (1.5,1+\ExampleInftyOffset) -- (2,1+\ExampleInftyOffset);
   \draw[function,smooth,samples=100,domain=2:4,line cap=round] plot(\x, {4/((\x-1)*(\x-1))-5});
   \draw[function,line cap=round] (4,-4.55556) -- (5.5,-4.55556);
  \end{pgfinterruptboundingbox}
 \end{scope}
  
  \begin{pgfonlayer}{foreground}
   \node[discontinuityblank] (d1) at (2,1+\ExampleInftyOffset) {};
   \node[discontinuityfilled] (d2) at (2,-1) {};
  \end{pgfonlayer}
  
  \node[functionterm] (f) at (3.9,-2.2) {$\frac{4}{(x-1)^2}-5$};
\end{tikzpicture}
}

\newcommand{\ConstrainedExampleDrawRealisticSoCFunctionLink}
{
\begin{tikzpicture}[figure,scale=\TradeoffFunctionExampleScale]
 \filldraw[draw=black30,fill=black7] plot[domain=4:5](\x, {4/((\x-3)*(\x-3))+0.55556}) --
  plot[domain=5:7](\x, {32/((\x-1)*(\x-1))-0.44444}) --
  plot[domain=7:9](\x, {4/((\x-4)*(\x-4))}) --
  plot[domain=9:7](\x, {4/((\x-6)*(\x-6))-0.28444}) --
  plot[domain=7:4](\x, {4/((\x-2)*(\x-2))+32/9});

 \ExampleDrawCoordinateSystemWithPositiveInfty{4}{9.5}{0}{6}{3}{0}

 \PlotXAxisName{9}{0}{$\atime$}

 \draw[thick,draw=black50] (3,5) -- (9.5,5) node[very near end,above=2pt,edgelabel] {$\soc=5$};
 \draw[->] (5,4.8) -- (5,4.2);

 \node[functionterm,text=primarycolor] (fphantom) at (7.5,2.1) {\phantom{$\tcfunction(\atime)=\frac{4}{(\atime-4)^2}$}};
 \begin{scope}
 \clip (fphantom.north east) rectangle (fphantom.south west);
  \filldraw[draw=black30,fill=black7] plot[domain=4:5](\x, {4/((\x-3)*(\x-3))+0.55556}) --
   plot[domain=5:7](\x, {32/((\x-1)*(\x-1))-0.44444}) --
   plot[domain=7:9](\x, {4/((\x-4)*(\x-4))}) --
   plot[domain=9:7](\x, {4/((\x-6)*(\x-6))-0.28444}) --
   plot[domain=7:4](\x, {4/((\x-2)*(\x-2))+32/9});
 \end{scope} 

 \node[functionterm,text=secondarycolor-light] (f1) at (6,5.9) {$\tcfunction^+(\atime=\atime_1+\atime_2)$};
 \node[functionterm,text=secondarycolor-dark] (f2) at (8.2,4.2) {$\tcfunction^+_\soc(\atime)$};
 \node[functionterm,fill=none] (f3) at (7.5,2.1) {$\tcfunction(\atime)=\frac{4}{(\atime-4)^2}$};

 \begin{scope}
  \begin{pgfinterruptboundingbox}
   \clip (3,0) rectangle (9.5,7+\ExampleInftyOffset);
   \draw[function] (3,6+\ExampleInftyOffset) -- (5,6+\ExampleInftyOffset);
   \draw[function,secondarycolor-light] (3,6+\ExampleInftyOffset) -- (4,6+\ExampleInftyOffset);

   \draw[function,smooth,samples=100,domain=5:9,line cap=round] plot(\x, {4/((\x-4)*(\x-4))});
   \draw[function,line cap=round] (9,0.16) -- (10,0.16);

   \draw[function,secondarycolor-light,smooth,samples=100,domain=4:5,line cap=round] plot(\x, {4/((\x-2)*(\x-2))+41/9});
   \draw[function,secondarycolor-dark,smooth,samples=100,domain=5:7,line cap=round] plot(\x, {4/((\x-2)*(\x-2))+41/9});
   \draw[function,secondarycolor-dark,line cap=round] (7,4.71556) -- (10,4.71556);
  \end{pgfinterruptboundingbox}
 \end{scope}

   \begin{pgfonlayer}{foreground}
    \node[discontinuityblank,draw=secondarycolor-light] (d1) at (4,6+\ExampleInftyOffset) {};
    \node[discontinuityblank] (d2) at (5,6+\ExampleInftyOffset) {};
    \node[discontinuityfilled,secondarycolor-light] (d3) at (4,5.55556) {};
    \node[discontinuityfilled,secondarycolor-dark] (d4) at (5,5) {};
    \node[discontinuityfilled] (d5) at (5,4) {};
   \end{pgfonlayer}
\end{tikzpicture}
}

\newcommand{\ConstrainedExampleDrawRealisticSoCFunctionUpperBound}
{
\begin{tikzpicture}[figure,scale=\TradeoffFunctionExampleScale]
  \filldraw[draw=black30,fill=black7] plot[domain=4:5](\x, {4/((\x-3)*(\x-3))+0.55556}) --
  plot[domain=5:7](\x, {32/((\x-1)*(\x-1))-0.44444}) --
  plot[domain=7:9](\x, {4/((\x-4)*(\x-4))}) --
  plot[domain=9:7](\x, {4/((\x-6)*(\x-6))-0.28444}) --
  plot[domain=7:4](\x, {4/((\x-2)*(\x-2))+32/9});
 
 \ExampleDrawCoordinateSystemWithPositiveInfty{4}{10}{0}{6}{3}{0}

 \PlotXAxisName{10}{0}{$\atime$}

 \node[functionterm,text=primarycolor] (f2box) at (6.6,1.8) {\phantom{$\tcfunctionshortcut^+(\atime)+\tcfunctionshortcut^-(\rightbordershortcut^-)=\frac{4}{(\atime-2)^2}$}};
 \begin{scope}
 \clip (f2box.north east) rectangle (f2box.south west);
  \filldraw[draw=black30,fill=black7] plot[domain=4:5](\x, {4/((\x-3)*(\x-3))+0.55556}) --
  plot[domain=5:7](\x, {32/((\x-1)*(\x-1))-0.44444}) --
  plot[domain=7:9](\x, {4/((\x-4)*(\x-4))}) --
  plot[domain=9:7](\x, {4/((\x-6)*(\x-6))-0.28444}) --
  plot[domain=7:4](\x, {4/((\x-2)*(\x-2))+32/9});
 \end{scope}
 
 \node[functionterm,text=secondarycolor] (f1) at (7.5,5.5) {$\tcfunction^+(\atime)=\frac{4}{(\atime-2)^2}+\frac{50}{9}$};
 \node[functionterm,text=primarycolor,fill=none] (f2) at (6.6,1.8) {$\tcfunctionshortcut^+(\atime)+\tcfunctionshortcut^-(\rightbordershortcut^-)=\frac{4}{(\atime-2)^2}$};
 
 \begin{scope}
  \begin{pgfinterruptboundingbox}
   \clip (3,0) rectangle (10,7+\ExampleInftyOffset);
   \draw[function,primarycolor] (3,6+\ExampleInftyOffset) -- (4,6+\ExampleInftyOffset);

   \draw[function,secondarycolor,smooth,samples=100,domain=4:7,line cap=round] plot(\x, {4/((\x-2)*(\x-2))+41/9});
   \draw[function,secondarycolor,line cap=round] (7,4.71556) -- (10,4.71556);

   \draw[function,primarycolor,smooth,samples=100,domain=4:7,line cap=round] plot(\x, {4/((\x-2)*(\x-2))});
   \draw[function,primarycolor,line cap=round] (7,0.16) -- (10,0.16);
  \end{pgfinterruptboundingbox}
 \end{scope}

   \begin{pgfonlayer}{foreground}
    \node[discontinuityblank,draw=primarycolor] (d1) at (4,6+\ExampleInftyOffset) {}; 
    \node[discontinuityfilled,primarycolor] (d3) at (4,1) {};
    \node[discontinuityfilled,secondarycolor] (d4) at (4,50/9) {};
   \end{pgfonlayer}
\end{tikzpicture}
}

\usepackage[
colorlinks=true,
linkcolor=thesisred-dark,
urlcolor=thesisred-dark,
citecolor=thesisblue-dark,
pdfauthor={Moritz Baum, Julian Dibbelt, Dorothea Wagner, and Tobias Zündorf},
pdftitle={Modeling and Engineering Constrained Shortest Path Algorithms for Battery Electric Vehicles}
]{hyperref}

\usepackage[nopostdot,style=super,nonumberlist,toc,nowarn]{glossaries}

\makeglossaries

\newacronym[description={Battery Electric Vehicle}]     {ev}{EV}{battery electric vehicle}
\newacronym[description={State of Charge}]              {soc}{SoC}{state of charge}
\newacronym{vrp}{VRP}{Vehicle Routing Problem}
\newacronym{csp}{CSP}{Constrained Shortest Path}
\newacronym{spprc}{SPPRC}{Shortest Path Problem with Resource Constraints}
\newacronym{evcsp}{EVCSP}{EV Constrained Shortest Path}
\newacronym{evcas}{EVCAS}{EV Continuous Adaptive Speeds Shortest Path}
\newacronym{alt}{ALT}{A*, Landmarks, Triangle Inequality}
\newacronym{ch}{CH}{Contraction Hierarchies}
\newacronym{cch}{CCH}{Customizable Contraction Hierarchies}
\newacronym{crp}{CRP}{Customizable Route Planning}
\newacronym{mld}{MLD}{Multilevel Dijkstra}
\newacronym{phast}{PHAST}{PHAST Hardware-Accelerated Shortest Path Trees}
\newacronym[description={Breadth-First Search}]         {bfs}{BFS}{breadth-first search}
\newacronym[description={Depth-First Search}]           {dfs}{DFS}{depth-first search}
\newacronym[description={Directed Acyclic Graph}]       {dag}{DAG}{directed acyclic graph}
\newacronym[description={Single-Source Shortest Path}]  {sssp}{SSSP}{single-source shortest path}
\newacronym[description={Single-Pair Shortest Path}]    {spsp}{SPSP}{single-pair shortest path}
\newacronym[description={All-Pairs Shortest Path}]      {apsp}{APSP}{all-pairs shortest path}
\newacronym{srtm}{SRTM}{Shuttle Radar Topography Mission}
\newacronym{phem}{PHEM}{Passenger Car and Heavy Duty Emission Model}
\newacronym{hbefa}{HBEFA}{Handbook on Emission Factors for Road Traffic}
\newacronym{osm}{OSM}{OpenStreetMap}
\newacronym[description={First-In-First-Out}]           {fifo}{FIFO}{first-in-first-out}
\newacronym{ed}{ED}{Edge Difference}
\newacronym{cq}{CQ}{Cost of Queries}
\newacronym{sc}{SC}{Shortcut Complexity}
\newacronym{dn}{DN}{Deleted Neighbors}
\newacronym[description={Bicriteria Shortest Path (Algorithm)}]      {bsp}{BSP}{bicriteria shortest path}
\newacronym[description={Tradeoff Function Propagating (Algorithm)}] {tfp}{TFP}{tradeoff function propagating}
\newacronym                                             {champ}{CHAsp}{CH, A*, Adaptive Speeds}
\newacronym[description={Fully Polynomial-Time Approximation Scheme}] {fptas}{FPTAS}{fully polynomial-time approximation scheme}

\newtheorem{definition}{Definition}

\newtheorem{lemma}{Lemma}

\newcommand{\email}[1]{\texttt{#1}}
\title{Modeling and Engineering Constrained Shortest Path Algorithms for Battery Electric Vehicles\thanks{Our work was partially supported by DFG Research Grant WA\,654/16--2 and DFG Research Grant WA\,654/23--1.
This manuscript is the full version of an extended abstract that appeared in the proceedings of the 25th Annual European Symposium on Algorithms~\cite{Bau17a}. It builds upon the PhD thesis of one of the authors~\cite{Bau18b}}}
\author[1]{Moritz Baum}
\author[2]{Julian Dibbelt}
\author[1]{Dorothea Wagner}
\author[1]{Tobias Zündorf}
\affil[1]{Karlsruhe Institute of Technology, Germany\\\email{moritz@ira.uka.de,dorothea.wagner@kit.edu,tobias.zuendorf@kit.edu}}
\affil[2]{Sunnyvale, CA, United States\\\email{algo@dibbelt.de}}
\date{}

\begin{document}

\maketitle

\begin{abstract}
We study the problem of computing constrained shortest paths for battery electric vehicles. Since battery capacities are limited, fastest routes are often infeasible. Instead, users are interested in fast routes on which the energy consumption does not exceed the battery capacity. For that, drivers can deliberately reduce speed to save energy. Hence, route planning should provide both path \emph{and} speed recommendations. To tackle the resulting \NP-hard optimization problem, previous work trades correctness or accuracy of the underlying model for practical running times.
We present a novel framework to compute \emph{optimal} constrained shortest paths (without charging stops) for electric vehicles that uses more realistic physical models, while taking speed adaptation into account. Careful algorithm engineering makes the approach practical even on large, realistic road networks: We compute optimal solutions in less than a second for typical battery capacities, matching the performance of previous inexact methods. For even faster query times, the approach can easily be extended with heuristics that provide high quality solutions within milliseconds.
\end{abstract}

\section{Introduction}\label{sec:introduction}

\Glspl*{ev} have matured, giving the prospect of high powertrain efficiency and independence of fossil fuels, but a major hindrance of their adoption remains the limited battery capacity of most vehicles combined with a lengthy recharge time.
To overcome \emph{range anxiety}, careful route planning that prevents battery depletion during a ride is paramount.
To avoid time-consuming charging stops, it can pay off to save energy by deliberately driving below posted speed limits, especially along high-speed roads. Speed planning becomes even more relevant with the advent of autonomous vehicles, where driving speeds can be planned in advance to ensure that the target is reached.
Another substantial difference to vehicles run by combustion engines is the ability to recuperate energy when braking.
In this work, we study route planning algorithms that incorporate these modeling considerations to capture the characteristics and needs of \glspl*{ev}.
We discuss different ways to model \emph{adaptive speeds} and propose algorithmic solutions to solve the resulting \NP-hard problem of finding a fast route on which the target can be reached without intermediate recharging.
We demonstrate the practicality of our approaches in a comprehensive evaluation on realistic input.

\paragraph{Related Work.}
Classic route planning approaches make use of a graph-based representation of the considered transportation network, where scalar arc costs correspond to, \eg, travel times. A shortest path is then found by Dijkstra's algorithm~\cite{Dij59}.
A wide range of \emph{speedup techniques}~\cite{Bas14} enable provably correct but faster queries in practice. For instance, A*~search~\cite{Har68} uses \emph{vertex potentials} to guide the search towards the target.
\Gls*{ch} introduced by Geisberger et~al.~\cite{Gei12b} employs a preprocessing step to obtain a search graph that allows skipping vast parts of the network at query time. For that, \gls*{ch} iteratively \emph{contracts} vertices according to a heuristic vertex ranking, while adding \emph{shortcut} arcs to maintain distances within the remaining graph. 
Extensions to multicriteria settings exist for both A*~search~\cite{Erb14,Mac12,Man10,San13} and \gls*{ch}~\cite{Fun13,Gei10b,Sto12d}.
Moreover, A*~search and \gls*{ch} can be combined to Core-\acrshort*{alt}~\cite{Bau08b}, where all but the highest-ranked vertices are contracted, which form the \emph{core} graph. On that, a variant of A*~search uses precomputed distances to \emph{landmark} vertices for potentials~\cite{Gol05b}.

Route planning for \glspl*{ev} involves battery capacity constraints and negative arc weights due to recuperation, which is tractable when optimizing energy consumption as a single criterion, \ie, ignoring driving speed entirely~\cite{Bau13a,Eis11,Sac11}. Such energy-optimal routes often exhibit disproportionate detours, as using minor, slow roads can save energy due to less air drag~\cite{Bau13a}.
Variants of the \NP-hard \emph{\gls*{csp}} problem~\cite{Han80a} overcome this by minimizing energy consumption without exceeding a given time limit~\cite{Jur14,Liu16,Sto12c} or finding the fastest route (without adaptive speeds) that does not exceed battery constraints~\cite{Bau18,Wan13}.
These works extend well-known bicriteria search algorithms~\cite{Han80b,Mar84} to tackle \gls*{csp} problem variants for~\glspl*{ev}. Storandt~\cite{Sto12c} proposes a combination with \gls*{ch}, which computes optimal results in the order of milliseconds on subcountry-scale networks.

All works mentioned so far assume that the driving speed per road segment is \emph{fixed} in the sense that it cannot be adjusted by the driver, neglecting attractive solutions that may still use major roads (\eg, motorways) and save energy by deliberately driving below posted speed limits.
Flores et~al.~\cite{Flo15} consider the problem of planning driving speeds of an \gls*{ev} for a given route in the network, \ie, they only plan speeds, but not the route itself.
Lv et~al.~\cite{Lv16} introduce a dynamic programming approach to optimally plan the speed of a solar-powered \gls*{ev}. Designed for simulation purposes, it is too slow for interactive applications.
Fontana~\cite{Fon13} proposes a variant of the \gls*{csp} problem with the additional requirement of determining velocities for all road segments. Taking robustness criteria into account to deal with uncertainty \wrt time and energy consumption, a heuristic approach based on Lagrangian relaxation achieves reasonable running times (few seconds) on small graphs.
By sampling discrete alternative speeds, tradeoffs can also be modeled as parallel arcs~\cite{Bau14a}, but this yields too many nondominated intermediate solutions, growing exponentially even for chains of vertices. Accordingly, only heuristics offer acceptable performance for common vehicle ranges when using parallel arcs.
Instead, Hartmann and Funke~\cite{Har14} model tradeoffs as continuous \emph{functions}, assuming the driver can go at \emph{any} speed up to a given speed limit per arc. However, their solution is still based on discrete sampling of these continuous functions (it also ignores battery constraints). Despite employing a variant of~\gls*{ch} as a preprocessing-based speedup technique, their approach does not scale well to large networks, with query times in the order of minutes.
Alternatively, \emph{two-phase paths} consist of a fastest and an energy-optimal subpath~\cite{Goo14}, where driving speed along arcs is determined by the type of the respective subpath. This allows for reasonably fast queries in the order of seconds (without preprocessing), but results are not optimal in general.
Strehler et~al.~\cite{Mer15} give theoretical insights for a \gls*{csp} problem including variable speeds for~\glspl*{ev}. Most importantly, they develop a \gls*{fptas} for this problem. Unfortunately, the algorithm is slow in practice. They also propose a heuristic search based on discretized speeds, but do not evaluate their approach.

Speed planning and multicriteria optimization are also relevant in related Crew Scheduling and \glspl*{vrp}. A subproblem that is often considered in this context is called~\emph{\gls*{spprc}}~\cite{Irn05}, which exists in many variants and can be seen as a generalization of \gls*{csp} in which a shortest path subject to one or more resource constraints is sought. These constraints can be represented by resource extension functions~\cite{Irn07}, which define the change in resource consumption or state when traversing an arc of the network~\cite{Sch15}.
Often, these state transitions are modeled as linear~\cite{Ioa98,Lib11,Til18} or nonlinear functions~\cite{He18}. Another common approach is to use parallel arcs to represent different tradeoff choices~\cite{Gar16,Let14}.
For an overview of variants of \gls*{spprc} and the \gls*{vrp} with realistic constraints and road attributes, see Ticha et~al.~\cite{Tic18}.
Works dealing with speed planning can be found in Bekta\c{s} et~al.~\cite{Bek16} and Fukasawa et~al.~\cite{Fuk18}. Another related problem is the \emph{Pollution Routing Problem}~\cite{Bek11}, which aims at minimizing emissions of vehicles. \Glspl*{ev} are explicitly considered in several works~\cite{Des16,Goe15,Mon15,Sch14b}. As a general observation, \gls*{vrp} settings are very complex, as they extend the Traveling Salesman Problem and often involve routing of multiple vehicles. Consequently, approaches in the literature (typically making use of mathematical programming and heuristics) do not scale to large, fully modeled road networks.
Finally, we note that speed planning for vehicles is relevant not only in the context of route planning for \glspl*{ev}, but also in several other areas of transportation, such as supply chain management~\cite{Ber16}, aviation~\cite{Xu16}, or ship routing~\cite{Men14,Nor11}.

\paragraph{Contribution and Outline.}
We study a generalization of the \gls*{csp} problem that considers \emph{continuous, adaptive speeds} for~\glspl*{ev}: We allow the vehicle to adjust its speed to reach its target quickly and without recharging (if possible). We derive a realistic, nonlinear consumption model and develop an algorithm to solve the resulting \NP-hard problem \emph{optimally} on challenging, realistic instances.
This complements our previous work~\cite{Bau18}, in which we considered fastest feasible routes with charging stops (but for fixed driving speeds).
In what follows, we first recap the \gls*{csp} problem and discuss adaptive speed models in Section~\ref{sec:model}. We propose functions mapping travel time to energy consumption, which yield the most challenging but precise problem setting.
Afterwards, Section~\ref{sec:operations} discusses realistic, nonlinear consumption models and derives operations to compute and compare speed--consumption tradeoffs on paths of the network.
In Section~\ref{sec:approach}, we introduce our basic exponential-time algorithm that uses these operations. Propagating continuous consumption functions during network exploration, it aims to improve both running time \emph{and} solution quality compared to previous discretized approaches.
In Section~\ref{sec:approach}, we extend known variants of A*~search to our setting, improving the practical performance of our basic algorithm.
We then incorporate a technique based on \gls*{ch} in Section~\ref{sec:ch}, for which a particular challenge is the computation of shortcuts that represent \emph{bivariate functions} to capture the constraints imposed by our consumption model.
Combining A*~search and~\gls*{ch}, we obtain our fastest approach, which is presented in Section~\ref{sec:chasp}.
Our experimental evaluation, given in Section~\ref{sec:experiments}, reveals that we can compute optimal solutions within seconds and below for typical battery capacities, on par or faster than previous \emph{heuristic} algorithms.
Our own heuristic variant provides high-quality solutions and is fast enough for interactive applications.
We conclude our work with final remarks in Section~\ref{sec:conclusion}.

\section{Models and Problem Settings}\label{sec:model}

We begin by introducing a basic problem setting in Section~\ref{sec:model:evcsp}, in which we assume that the speed on every arc is fixed. We discuss a simple way to extend graphs to take adaptive speeds into account by adding multi-arcs and demonstrate its drawbacks. Afterwards, we introduce a more sophisticated model of adaptive speeds and the corresponding problem generalization in Section~\ref{sec:model:evcas}. Finally, Section~\ref{sec:model:example} illustrates how solutions can be obtained in this model on a simplified example.

In each setting discussed below, we are given a directed graph~$\graph=(\vertices,\arcs)$. Arcs in the graph are directed to model, \eg, one-way streets.
Traversing an arc in the graph consumes both time and energy. While the time to traverse an arc is always (strictly) positive, energy consumption can become negative to model recuperation, \eg, when going downhill. Moreover, the vehicle is equipped with a battery that has a given capacity~$\maxbattery\in\posreals$. This imposes constraints on the feasibility of paths, as the \gls*{soc} of the vehicle must not drop below~$0$. Additionally, the \gls*{soc} can become at most $\maxbattery$ after traversing an arc with negative energy consumption.

In what follows, we focus on computing fastest feasible paths for~\glspl*{ev}, \ie, fastest paths under the constraint that they are reachable by the~\gls*{ev} (given its initial~\gls*{soc}). However, we point out that it is not hard to adapt the algorithms we present to closely related problem settings, such as
\begin{itemize}
 \item computing a path with minimum driving time such that the target is reached with at least a certain minimum \gls*{soc};
 \item computing a path with minimum driving time such that the \gls*{soc} does not fall below a certain threshold at \emph{any} point during a journey;
 \item computing an energy-optimal path that does not exceed a certain travel time;
 \item computing the full Pareto set of nondominated solutions at the target (\wrt driving time and energy consumption).
\end{itemize}

\subsection{Constrained Shortest Paths for \glspl*{ev}}\label{sec:model:evcsp}

In a problem resembling traditional \gls*{csp} settings, we are given a graph $\graph=(\vertices,\arcs)$ together with two cost functions $\drivingtimefunction\colon\arcs\to\strictposreals$ and $\consumptionfunction\colon\arcs\to\reals$ representing \emph{driving time} and \emph{energy consumption} on an arc, respectively.
Although energy consumption can become negative, physical constraints prevent that cycles in the graph have negative energy consumption. In addition to that, the \emph{battery constraints} mentioned above apply: If the battery of the \gls*{ev} has the \gls*{soc} $\soc_\vertexa\in[0,\maxbattery]$ at some vertex~$\vertexa\in\vertices$, then traversing an arc $(\vertexa,\vertexb)\in\arcs$ results in the \gls*{soc}
\begin{align*}
 \soc_\vertexb:= \begin{cases}
 -\infty     & \mbox{if } \soc_\vertexa-\consumptionfunction(\vertexa,\vertexb)<0\mbox{,}\\
 \maxbattery & \mbox{if } \soc_\vertexa-\consumptionfunction(\vertexa,\vertexb)>\maxbattery\mbox{,}\\
 \soc_\vertexa-\consumptionfunction(\vertexa,\vertexb) & \mbox{otherwise.} \end{cases}
\end{align*}
Here, an \gls*{soc} of $-\infty$ implies that the arc cannot be traversed because the battery would run empty. A path is feasible if the \gls*{soc} never drops to $-\infty$ when traversing its arcs. We consider the \emph{\gls*{evcsp} Problem}, which is defined as follows.
\begin{definition}[\acrlong*{evcsp}]
\label{def:evcsp}
 Given a source vertex~$\source\in\vertices$, a target vertex~$\target\in\vertices$, and an initial \gls*{soc}~$\soc_\source\in[0,\maxbattery]$, find a path that is feasible and minimizes driving time.
\end{definition}
Note that if we did not allow recuperation and hence, all consumption values became nonnegative, this problem would be equivalent to the well-known \NP-hard \gls*{csp} problem~\cite{Han80a}, which (in the terminology of our setting) asks for a path with minimum driving time such that energy consumption does not exceed the threshold~$\soc_\source$. Consequently, our problem at hand is \NP-hard, too.

\paragraph{The Bicriteria Shortest Path Algorithm.}
To solve the~\gls*{evcsp} problem, we can adapt the well-known exponential-time \emph{\gls*{bsp}} algorithm~\cite{Han80b,Mar84} in a straightforward manner.
The algorithm maintains label sets, in our case containing tuples $(\accumulateddrivingtime,\accumulatedsoc)$ of driving time~$\accumulateddrivingtime\in\posreals$ and~\gls*{soc}~$\accumulatedsoc\in\reals$. A label $(\accumulateddrivingtime,\accumulatedsoc)$ \emph{dominates} another label $(\accumulateddrivingtime',\accumulatedsoc')$ if $\accumulateddrivingtime\le\accumulateddrivingtime'$ and~$\accumulatedsoc\ge\accumulatedsoc'$.
Starting with the label~$(0,\soc_\source)$ at the source~$\source$, the algorithm \emph{settles} in each step the label~$\alabel=(\accumulateddrivingtime,\accumulatedsoc)$ with minimum driving time~$\accumulateddrivingtime\ge0$ among any labels not settled before. For all arcs~$(\vertexa,\vertexb)\in\arcs$ outgoing from the vertex $\vertexa\in\vertices$ this label belongs to, it then generates a new label $\alabel'$ with driving time~$\accumulateddrivingtime+\drivingtimefunction(\vertexa, \vertexb)$. We set the \gls*{soc} of the new label~$\alabel'$ to~$\min\{\maxbattery,\accumulatedsoc-\consumptionfunction(\vertexa,\vertexb)\}$. If this results in nonnegative~\gls*{soc}, the label~$\alabel'$ represents a feasible path and is added to the labels at $\vertexb$ if it is not dominated by any of them. Labels at $\vertexb$ dominated by $\alabel'$ are discarded.
Given that driving times are strictly positive, this algorithm is \emph{label setting}, \ie, settled labels are never dominated afterwards.\footnote{Note that if we were scanning vertices in increasing order of energy consumption instead (which may be negative), potential shifting would be necessary to make the algorithm label setting~\cite{Bau13a,Eis11,Joh77,Sac11}.}
An optimal (constrained) solution is found when a label at the target~$\target$ is settled.

\paragraph{Modeling Adaptive Speeds in the \gls*{evcsp} Problem.}
The \gls*{bsp} algorithm finds the fastest \emph{route} subject to battery constraints.
As argued before, travel time and energy consumption are not only affected by the choice of the route itself, but also by \emph{driving behavior}.
Allowing \emph{multiple} driving speeds (and consumption values) per road segment, one could,~\eg, save energy on the motorway by driving at reasonable speeds below the posted speed limits.
Therefore, we consider \emph{adaptive speeds}, \ie, we allow the \gls*{ev} to adjust its speed within reasonable limits to reach the destination as fast as possible and with sufficient~\gls*{soc}.
In addition to the actual route from the source to the target, we then also have to specify (optimal) driving speeds along that route.

\begin{figure}[t]
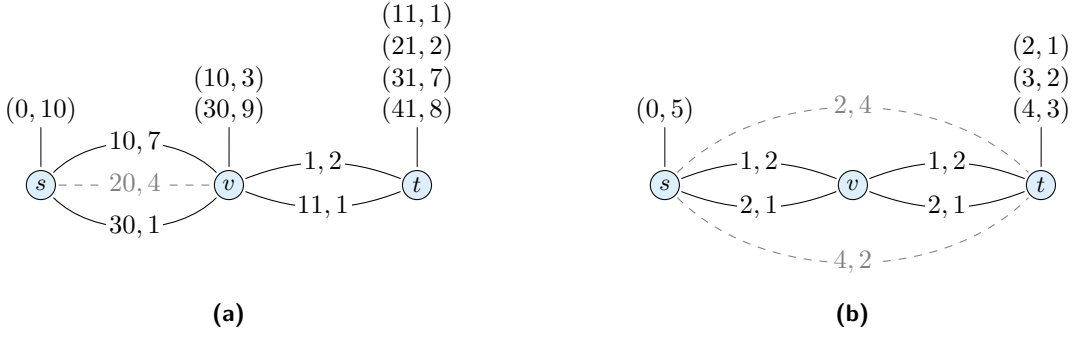

 \centering%
 \begin{subfigure}[b]{.5\textwidth}%
 \centering%
 \ConstrainedExampleDrawParallelEdgesDominance
 \caption{}%
 \label{fig:tradeoff-parallel-arcs:dominance}%
 \end{subfigure}%
 \begin{subfigure}[b]{.5\textwidth}%
 \centering%
 \ConstrainedExampleDrawParallelEdgesContraction
 \caption{}%
 \label{fig:tradeoff-parallel-arcs:contraction}%
 \end{subfigure}%
 \caption{Adaptive speeds modeled as parallel arcs. Arc labels indicate tuples of driving time and energy consumption. We also show the label sets of vertices as computed by the \acrshort*{bsp} algorithm in an $\source$--$\target$~query (ignoring dashed arcs). (a)~The initial \gls*{soc} is~$10$. Adding the dashed arc $(\source,\vertex)$ results in a new label~$(20,6)$ at~$\vertex$ and a label $(21,4)$ at~$\target$, which dominates the label~$(21,2)$ in the current set. (b)~The initial \gls*{soc} is~$5$. Assume that the original graph consists only of the vertices $\source$ and~$\target$, connected by the two dashed arcs. Thus, there are two labels $(2,1)$ and $(4,3)$ at the target vertex when the algorithm terminates. Inserting the vertex~$\vertex$, \eg, for visualization in practice, results in the depicted graph with solid arcs. Even though the distances between $\source$ and $\target$ are maintained, the algorithm now computes the additional label $(3,2)$ on the modified input.}%
\label{fig:tradeoff-parallel-arcs}%
\end{figure}

A straightforward way to model adaptive speeds is to \emph{sample} reasonable speeds for each road segment~\cite{Bau14a,Goo14,Har14,Mer15}. Then, one can add \emph{parallel} arcs in the underlying graph representation, which correspond to alternative driving speeds (inducing certain values of driving time and energy consumption); see Figure~\ref{fig:tradeoff-parallel-arcs}.
The major benefit of this approach is its simplicity: Since we still obtain an instance of~\gls*{evcsp}, we can immediately apply the \gls*{bsp} algorithm to solve the problem.
However, it also comes with several drawbacks. First of all, parallel arcs greatly increase running time, due to a larger number of nondominated solutions. In fact, the number of nondominated (\ie, Pareto-optimal) solutions can be exponential even on a single route; see Figure~\ref{fig:tradeoff-parallel-arcs:dominance} and \cf Letchford et~al.~\cite{Let14}. Consequently, only heuristic algorithms achieve practical running times~\cite{Bau14a,Goo14,Har14}.
By discretizing a continuous range of tradeoffs, parallel arcs that model alternative speeds have other undesirable effects, such as producing many insignificant, yet nondominated solutions. Figure~\ref{fig:tradeoff-parallel-arcs:dominance} shows such an example in which some solutions at the target vertex $\target$ provide rather unattractive tradeoffs, namely, spending ten extra units of time to save only one unit of energy. Adding another sample (indicated by the dashed arc), on the other hand, would result in a new label at $\target$ that dominates one of these less favorable solutions. In other words, the number of samples influences both running time and quality of the results: More samples increase running time, but fewer samples reduce quality.
Similarly, adding or contracting degree-two vertices in the graph, which are commonly included for visualization purposes, affects the solution space even if distances are maintained; see Figure~\ref{fig:tradeoff-parallel-arcs:contraction}. This is clearly not desirable, since such modeling decisions should not have any impact on the optimal solution.
To remedy the above issues, we propose a more sophisticated model, which uses continuous \emph{functions} to model the tradeoffs on arcs.

\subsection{Constrained Shortest Paths with Continuous Adaptive Speeds}\label{sec:model:evcas}

In the \gls*{evcsp} problem, scalar values represent driving time and energy consumption of an arc $\arc\in\arcs$ in the input graph~$\graph=(\vertices,\arcs)$. Instead, we now assume that there is a \emph{tradeoff function} $\tcinner_\arc\colon\strictposreals\to\reals$ based on a physical consumption model, mapping a desired driving time $\atime\in\strictposreals$ along $\arc$ to the resulting energy consumption~$\tcinner_\arc(\atime)$.
In reality, the driving time on a road segment cannot be chosen arbitrarily. Lower bounds are induced by speed limits and the maximum speed of a vehicle. On the other hand, driving slower than a reasonable minimum speed would mean to become an obstacle for other drivers. Additionally, there is a certain point at which driving slower will no longer pay off in terms of energy consumption.
This yields (positive) minimum and maximum driving times~$\leftintervalborder_\arc\in\strictposreals$ and~$\rightintervalborder_\arc\in\strictposreals$, respectively, for the function~$\tcinner_\arc$, with~$\leftintervalborder_\arc\le\rightintervalborder_\arc$.
We incorporate these bounds into a \emph{consumption function}~$\tcfunction_\arc \colon \posreals\to\reals\cup\{\infty\}$ with 
\begin{align}
 \label{eq:tradeoff-functiongeneralform}
 \tcfunction_\arc(\atime) := \begin{cases}
   \infty                                     & \mbox{if } \atime < \leftintervalborder_\arc, \\
   \tcinner_\arc(\rightintervalborder_\arc) & \mbox{if } \atime > \rightintervalborder_\arc, \\
   \tcinner_\arc(\atime)                     & \mbox{otherwise.}
 \end{cases}
\end{align}
Thus, driving times below~$\leftintervalborder_\arc$ are infeasible (modeled as infinite consumption) and driving times above~$\rightintervalborder_\arc$ become unprofitable.
A driving time $\atime\in\posreals$ is also called \emph{admissible} if~$\atime\in[\leftintervalborder_\arc,\rightintervalborder_\arc]$, \ie, it lies in the relevant subdomain of~$\tcfunction_\arc$.
In the special (degenerate) case~$\leftintervalborder_\arc = \rightintervalborder_\arc$, the function $\tcfunction_\arc$ represents a constant pair $(\leftintervalborder_\arc, \tcfunction_\arc(\leftintervalborder_\arc))$ of fixed driving time and energy consumption. We call $\arc$ and $\tcfunction_\arc$ \emph{constant} in this case, as the arc allows no speed adaptation.
Due to physical constraints, the minimum values $\tcfunction_\arc(\rightintervalborder_\arc)$ of the consumption functions $\tcfunction_\arc$ of arcs $\arc\in\arcs$ must not induce cycles with negative energy consumption in the graph (going in a cycle cannot increase the~\gls*{soc}).

As before, we assume that the \gls*{ev} is equipped with a battery that has a certain \emph{capacity}~$\maxbattery\in\posreals$ and that its \gls*{soc} must not drop below $0$ nor exceed~$\maxbattery$.
When incorporating these constraints into our setting, we obtain a bivariate \emph{\gls*{soc}~function} $\socprofile_\arc\colon\posreals\times[0,\maxbattery]\cup\{-\infty\}\to[0,\maxbattery]\cup\{-\infty\}$ for every arc $\arc=(\vertexa,\vertexb)\in\arcs$, mapping the \gls*{soc} at~$\vertexa$ to the resulting \gls*{soc} at~$\vertexb$ when traversing $\arc$ with a specific driving time. The function $\socprofile_\arc$ is given by
\begin{align}
 \label{eq:tradeoff-socfunction}
 \socprofile_\arc(\atime,\soc):=
\begin{cases}
 -\infty     & \mbox{if } \soc-\tcfunction_\arc(\atime)<0, \\
 \maxbattery & \mbox{if } \soc-\tcfunction_\arc(\atime)>\maxbattery, \\
 \soc - \tcfunction_\arc(\atime) & \mbox{otherwise,}
\end{cases}
\end{align}
where an \gls*{soc} of $-\infty$ denotes an empty battery. Hence, $\socprofile_\arc(\atime,\soc)=-\infty$ implies that the arc cannot be traversed at the corresponding speed (as it would cause the battery to run empty).
Further, we have $\socprofile_\arc(\atime,-\infty)=-\infty$ for arbitrary~$\atime\in\posreals$.

Given the \gls*{soc} $\soc_\source\in[0,\maxbattery]$ at a source~$\source\in\vertices$, an $\source$--$\target$~path $[\source=\vertex_1,\dots,\vertex_k=\target]$, and driving times $\atime_i\in\posreals$ for the arcs $(\vertex_i,\vertex_{i+1})$ of the path, with~$i\in\{1,\dots,k-1\}$, we can determine a corresponding \gls*{soc} $\soc_\target\in[0,\maxbattery]\cup\{-\infty\}$ at the target~$\target\in\vertices$ as follows.
Starting at the source $\source=\vertex_1$, the \gls*{soc} at $\vertex_k=\target$ is obtained after iteratively evaluating the \gls*{soc} function $\socprofile_{(\vertex_i,\vertex_{i+1})}$ at $\atime_i$ and the \gls*{soc} at the previous vertex~$\vertex_i$.
Formally, we get
\begin{align}
 \soc_1&:=\soc_\source\text{,}\nonumber\\
 \soc_{i+1}&:=\socprofile_{(\vertex_{i},\vertex_{i+1})}(x_i,\soc_i)&& \mbox{for } i\in\{1,\dots,k-1\}\text{,}\nonumber\\
 \soc_\target&:=\soc_{k}\text{.}\label{eq:soc-eval}
\end{align}
Note that $\soc_\target=-\infty$ holds if the path is infeasible (for the given driving times).

Using this more sophisticated model, Definition~\ref{def:evcas} formally introduces the \emph{\gls*{evcas} Problem}, which we examine in the subsequent sections of this work.
\begin{definition}[\acrlong*{evcas}]
\label{def:evcas}
 Given a source vertex~$\source\in\vertices$, a target vertex~$\target\in\vertices$, and an initial \gls*{soc}~$\soc_\source\in[0,\maxbattery]$, find an $\source$--$\target$~path $\apath=[\source=\vertex_1,\vertex_2\dots,\vertex_k=\target]$ together with driving times~$\atime_i\in\posreals$ for every arc $(\vertex_i,\vertex_{i+1})$ of~$\apath$, with~$i\in\{1,\dots,k-1\}$, such that battery constraints are respected, \ie, $\soc_\target\ge0$ in Equation~\ref{eq:soc-eval}, and the overall driving time~$\atime:=\sum_{i=1}^{k-1} \atime_i$ is minimized.
\end{definition}
An instance of \gls*{evcas} where all functions are degenerate constant tuples is also an input instance to~\gls*{evcsp}, which is \NP-hard. Hence, \gls*{evcas} is \NP-hard as well.

\subsection{Solving \gls*{evcas} on a Simplified Example Instance}\label{sec:model:example}

In this work, we focus on solving the \gls*{evcas} problem.
To gain insights about the structure of optimal solutions, we now derive consumption functions and (bivariate) \gls*{soc} functions for given \emph{paths} instead of arcs.
We illustrate such functions in an example using piecewise linear tradeoff functions.
Note that piecewise linear functions are commonly used to model constraints in variants of the \gls*{spprc}~\cite{Ioa98,Lib11,Til18} and in time-dependent route planning~\cite{Bat13,Bau16b,Del09b,Kon16b}.
In Section~\ref{sec:operations}, we propose a more realistic nonlinear model, which is used in the remainder of this work.

\begin{figure}[t]
 \centering%
 \begin{subfigure}[b]{.28\textwidth}%
 \centering%
 \ConstrainedExampleDrawPositiveTradeoffFunction
 \caption{}%
 \label{fig:tradeoff-tradeoff-simple-model:positive}%
 \end{subfigure}%
 \begin{subfigure}[b]{.28\textwidth}%
 \centering%
 \ConstrainedExampleDrawNegativeTradeoffFunction
 \caption{}%
 \label{fig:tradeoff-tradeoff-simple-model:negative}%
 \end{subfigure}%
 \begin{subfigure}[b]{.42\textwidth}%
 \centering%
 \ConstrainedExampleDrawTradeoffFunctionPath
 \caption{}%
 \label{fig:tradeoff-tradeoff-simple-model:path}%
 \end{subfigure}%
 \caption{Consumption functions in a simplistic model. (a)~The consumption function $\tcfunction_{(\vertexa,\vertexb)}$ of an arc~$(\vertexa,\vertexb)$ with minimum driving time $\leftintervalborder_{(\vertexa,\vertexb)}=1$ and maximum driving time~$\rightintervalborder_{(\vertexa,\vertexb)}=3$. (b)~The consumption function $\tcfunction_{(\vertexb,\vertexc)}$ of an arc $(\vertexb,\vertexc)$ with minimum driving time $\leftintervalborder_{(\vertexa,\vertexb)}=1$ and maximum driving time~$\rightintervalborder_{(\vertexa,\vertexb)}=2$. (c)~The consumption function $\tcfunction_{\apath}$ of the path $\apath=[\vertexa,\vertexb,\vertexc]$. The shaded area indicates possible pairs of driving time and consumption along the path.}%
\label{fig:tradeoff-simple-model}%
\end{figure}

For now, assume that the tradeoff function $\tcinner_\arc$ of every arc~$\arc\in\arcs$ is \emph{decreasing} and \emph{linear}, \ie, $\tcinner_\arc(\atime)=\aparam\atime+\bparam$ for all~$\atime\in\posreals$, where $\aparam\in\negreals$ and $\bparam\in\reals$ are constant coefficients.
The values $\aparam$ and $\bparam$ may differ between arcs, though, to reflect different road types or other relevant factors~\cite{Bru05,Yao13}.
Figure~\ref{fig:tradeoff-tradeoff-simple-model:positive} and Figure~\ref{fig:tradeoff-tradeoff-simple-model:negative} show consumption functions (plugging in limits $\leftintervalborder_{(\vertexa,\vertexb)}$ and $\rightintervalborder_{(\vertexa,\vertexb)}$ on driving time) for two arcs $(\vertexa,\vertexb)\in\arcs$ and~$(\vertexb,\vertexc)\in\arcs$.

To find the minimum driving time from $\vertexa$ to~$\vertexc$, we are interested in the consumption function of the path~$\apath=[\vertexa,\vertexb,\vertexc]$, \ie, a function $\tcfunction_\apath$ that maps driving time $\atime\in\posreals$ spent on the $\vertexa$--$\vertexc$ path to the \emph{minimum} energy consumption $\tcfunction_\apath(\atime)$ on the path.
Formally, to get the value of $\tcfunction_\apath(\atime)$ for some driving time~$\atime\in\posreals$, we have to pick (nonnegative) values $\atime_1\in\posreals$ and~$\atime_2\in\posreals$, such that $\atime=\atime_1+\atime_2$ and $\tcfunction_{(\vertexa,\vertexb)}(\atime_1)+\tcfunction_{(\vertexb,\vertexc)}(\atime_2)$ is minimized.
The shaded area in Figure~\ref{fig:tradeoff-tradeoff-simple-model:path} indicates possible distributions of driving times $\atime_1$ and $\atime_2$ among the two arcs and the resulting energy consumption. The lower envelope of this area yields the desired function~$\tcfunction_\apath$.
Intuitively, we want to spend as much of the available extra time (exceeding the minimum~$2$) as possible on the arc that provides the best tradeoff, \ie, the consumption function with the steeper slope (where spending additional time saves most energy).
As a result, the consumption function of a path is always \emph{convex} on the subdomain where its image is finite; see Figure~\ref{fig:tradeoff-tradeoff-simple-model:path}. Observe that, while tradeoff functions of single arcs $\arc\in\arcs$ are linear on the interval~$[\leftintervalborder_\arc,\rightintervalborder_\arc]$ of admissible driving times, the tradeoff function of a path is \emph{piecewise} linear on the interval induced by its minimum and maximum driving time. The number of linear subfunctions defining the function $\tcfunction_\apath$ is bounded by the number of arcs in the path.

\paragraph{Battery Constraints.}
The situation becomes more involved if we also take battery constraints into account. Then, energy consumption not only depends on the driving time we are willing to spend along a path, but also on the initial \gls*{soc}.
Hence, we obtain a bivariate function~$\socprofile_\apath$, which maps driving time and initial \gls*{soc} at~$\vertexa$ to the \gls*{soc} at~$\vertexc$.
Note that consumption is positive on the first arc $(\vertexa,\vertexb)$ and negative on the second arc~$(\vertexb,\vertexc)$ in Figure~\ref{fig:tradeoff-simple-model}.
As before, the second arc $(\vertexb,\vertexc)$ provides the better tradeoff. However, for low initial \gls*{soc}, we have to ensure that the first arc $(\vertexa,\vertexb)$ can be traversed. Hence, spending some additional time on this arc may be inevitable to obtain a feasible solution. In contrast, high \gls*{soc} values may prevent recuperation along the second arc~$(\vertexb,\vertexc)$, so driving slower no longer pays off at some point.
Figure~\ref{fig:tradeoff-tradeoff-bivariate-functions} sketches the resulting bivariate \gls*{soc} function for specific values of initial \gls*{soc} and driving time.
For a given initial \gls*{soc} $\soc\in[0,\maxbattery]$ at~$\vertexa$, we see how spending more time on the $\vertexa$--$\vertexb$~path can increase the \gls*{soc} at~$\vertexc$ (Figure~\ref{fig:tradeoff-tradeoff-bivariate-functions:time}).
When fixing the driving time $\atime\in\posreals$ (Figure~\ref{fig:tradeoff-tradeoff-bivariate-functions:soc}), the optimal amount of time spent on each arc varies with the initial \gls*{soc} at~$\vertexa$. Consequently, even if we fix the total driving time that is spent on the path, the \gls*{soc} function of a path $\apath$ no longer has the specific form as in the case of scalar energy consumption values~\cite{Bau17b,Eis11}.

The optimal solution to \gls*{evcas} is the minimum driving time of $\tcfunction_{\apath}$ for the corresponding initial \gls*{soc} (which is part of the input). For example, the optimal driving time is $4$ if the initial \gls*{soc} is~$1$, or $2$ if the initial \gls*{soc} is at least~$2$.

\begin{figure}[t]
 \centering%
 \begin{subfigure}[b]{.5\textwidth}%
 \centering%
 \ConstrainedExampleDrawBivariateTimeSoCTradeoffFunction
 \caption{}%
 \label{fig:tradeoff-tradeoff-bivariate-functions:time}%
 \end{subfigure}%
 \begin{subfigure}[b]{.5\textwidth}%
 \centering%
 \ConstrainedExampleDrawBivariateSoCSoCTradeoffFunction
 \caption{}%
 \label{fig:tradeoff-tradeoff-bivariate-functions:soc}%
 \end{subfigure}%
 \caption{The bivariate \gls*{soc} function of the path $\apath$ from Figure~\ref{fig:tradeoff-simple-model}, assuming a battery capacity of~$\maxbattery=4$. (a)~The \gls*{soc} $\socprofile_\apath(\atime,\soc)$ at $\vertexc$, subject to driving time $\atime\in\posreals$ on $\apath$ for different fixed values $\soc\in\{1,2,3,4\}$ of initial~\gls*{soc}. (b)~The \gls*{soc} $\socprofile_\apath(\atime,\soc)$ at $\vertexc$, subject to initial \gls*{soc} $\soc\in[0,\maxbattery]$ for different fixed driving times~$\atime\in\{2,3,4,5\}$. }%
\label{fig:tradeoff-tradeoff-bivariate-functions}%
\end{figure}

\section{Realistic Consumption Functions and Basic Operations}\label{sec:operations}

Realistic tradeoff functions are nonlinear. We introduce such tradeoff functions in Section~\ref{sec:operations:consumption-functions}.
Although these realistic functions require a more technical analysis, many observations made for the simplistic, linear model discussed in Section~\ref{sec:model:example} carry over to the realistic, nonlinear tradeoff functions.

If we want to incorporate the propagation of consumption functions into the \gls*{bsp} algorithm in order to solve~\gls*{evcas}, we need to generalize some of its basic operations. First, \gls*{bsp} computes scalar sums to obtain the driving time and the \gls*{soc} at a vertex. Since we are now dealing with \emph{continuous} functions, we require operations that compute functions for (best) tradeoffs of \emph{paths} instead of arcs. In the previous section, we sketched how such consumption functions can be obtained in a simplistic model based on piecewise linear functions. In Section~\ref{sec:operations:linking}, we describe how this is done in the realistic model. Second, the \gls*{bsp} algorithm identifies dominated labels to keep the number of labels small. We discuss generalized dominance schemes for consumption functions in Section~\ref{sec:operations:dominance}.

\subsection{A Realistic Consumption Model}\label{sec:operations:consumption-functions}

Both considered metrics, driving time and energy consumption, depend on the vehicle's speed. In accordance with realistic physical models established in the literature~\cite{Agr16,Asa16,Bed15,Fio16,Har14,Lar12,Lv16}, we assume that energy consumption on a certain road segment~$\arc\in\arcs$ can be expressed by a function $\speedconsumptionfunction_\arc\colon\strictposreals\to\reals$ given as
\begin{align}
 \label{eq:tradeoff-physicalmodel}
 \speedconsumptionfunction_\arc(\velocity)=\aconst\velocity^2+\bconst\slope+\cconst \text{,}
\end{align}
where $\velocity\in\strictposreals$ is the (constant) vehicle speed, $\slope\in\reals$ is the (constant) slope of the road segment, and~$\aconst\in\posreals$,~$\bconst\in\posreals$, and~$\cconst\in\posreals$ are constant nonnegative coefficients of the consumption model.
The term $\aconst\velocity^2$ is caused by aerodynamic drag, which increases with driving speed in a quadratic fashion.
Note that energy consumption can become negative for downhill segments with $\slope<0$.
The parameters~$\aconst$,~$\bconst$, and~$\cconst$ may vary for different arcs due to, \eg, different road types or other factors affecting energy consumption~\cite{Bru05,Swe11,Yao13}.
Assuming constant speed and slope per arc is not a restriction, since we can add intermediate vertices in the graph to model changing conditions. Furthermore, one can show that deliberately varying the speed along a single road segment (with constant slope and speed limit) never pays off in our model~\citep[Corollary~1]{Har14}.

Since we are interested in functions mapping \emph{driving time} $\atime\in\strictposreals$ to energy consumption~$\tcinner_\arc(\atime)$, we substitute $\velocity=\length/\atime$ in Equation~\eqref{eq:tradeoff-physicalmodel}, where $\length\in\strictposreals$ denotes the length of the road segment. As slope and length of an arc are fixed, we simplify this below by setting $\aparam:=\aconst\length^2$ and~$\cparam:=\bconst\slope+\cconst$. Observe that $\aparam\in\posreals$ is nonnegative, while $\cparam\in\reals$ may be a negative value (for downhill arcs).
We introduce a third constant $\bparam\in\posreals$, which we will need later to shift functions along the time axis.
Altogether, we obtain the tradeoff function~$\tcinner_\arc\colon\strictposreals\to\reals$ mapping driving time to energy consumption, which is defined as
\begin{align}
 \label{eq:tradeoff-functioninnerform}
 \tcinner_\arc(\atime):=\frac{\aparam}{(\atime-\bparam)^2}+\cparam\text{.}
\end{align}
For single arcs, we always obtain $\bparam=0$ and assume driving time $\atime$ to be \emph{strictly} positive. Thus, the term $\atime-\bparam$ in the denominator is strictly positive and $\tcinner_\arc(\atime)$ is a finite real value.
Furthermore, note that $\tcinner_\arc$ is \emph{decreasing} and \emph{convex} on its domain $\strictposreals$ in this case.
Tradeoff functions of \emph{paths} may require values $0<\bparam<\atime$ to reflect additional time spent on previous arcs.
In the simplistic model discussed above, we saw that tradeoff functions of paths may be piecewise linear. Similarly, we allow tradeoff functions in the realistic model to be defined as \emph{piecewise} functions, so they may consist of multiple subfunctions of the form of Equation~\eqref{eq:tradeoff-functioninnerform}.

\begin{figure}[t]
 \centering%
 \ConstrainedExampleDrawRealisticTradeoffFunction
 \caption{A consumption function, defined by a single tradeoff function with parameters $\aparam=3$,~$\bparam=1$, and~$\cparam=1$. The indicated subdomain borders induced by its minimum and maximum driving time, respectively, are $\leftintervalborder=2$ and~$\rightintervalborder=6$.}%
\label{fig:tradeoff-consumption-function}%
\end{figure}

Given a tradeoff function $\tcinner\colon\strictposreals\to\reals$, we plug in the minimum and maximum driving times~$\leftintervalborder\in\strictposreals$ and~$\rightintervalborder\in\strictposreals$ to obtain the corresponding \emph{consumption function}~$\tcfunction\colon\posreals\to\reals\cup\{\infty\}$; see Figure~\ref{fig:tradeoff-consumption-function} for an example.
In general, we require consumption functions to be \emph{continuous} in the interval~$[\leftintervalborder,\infty)$, but not necessarily differentiable.
In particular, we demand that $\bparam<\atime$ holds for all subfunctions of the form of Equation~\eqref{eq:tradeoff-functioninnerform} within their respective subdomain of~$[\leftintervalborder,\rightintervalborder]$. Hence, the denominator term $\atime-\bparam$ is always strictly positive.
Together with the assumption~$\alpha\ge0$, this implies that consumption functions are either \emph{constant} (if~$\alpha=0$, in which case we further assume~$\leftintervalborder=\rightintervalborder$) or \emph{strictly decreasing} on the interval~$[\leftintervalborder,\rightintervalborder]$, \ie, $\tcfunction(\atime_1)>\tcfunction(\atime_2)$ holds for all $\atime_1,\atime_2\in[\leftintervalborder,\rightintervalborder]$ with~$\atime_1<\atime_2$.
At certain points below, we also make use of the \emph{inverse function} $\tcfunction^{-1}\colon[\tcfunction(\rightintervalborder),\tcfunction(\leftintervalborder)]\to\strictposreals$ of a consumption function~$\tcfunction$. Observe that it is well-defined on the specified domain.

\subsection{Linking Consumption Functions}\label{sec:operations:linking}

To compute best possible tradeoffs on paths consisting of multiple arcs, we define a \emph{link operation} $\linkop\colon\consumptionfunctionspace\times\consumptionfunctionspace\to\consumptionfunctionspace$ on the function space $\consumptionfunctionspace$ of consumption functions as specified in Section~\ref{sec:operations:consumption-functions}.
For two paths $\apath_1=[\vertex_1,\dots,\vertex_i]$ and $\apath_2=[\vertex_i,\dots,\vertex_k]$ in the graph, let $\apath_1\circ\apath_2:=[\vertex_1,\dots,\vertex_i,\dots,\vertex_k]$ denote their concatenation.
Given two consumption functions $\tcfunction_1$ and $\tcfunction_2$ representing energy consumption on two paths $\apath_1$ and~$\apath_2$, respectively, linking $\tcfunction_1$ and $\tcfunction_2$ results in a consumption function $\tcfunction:=\linkop(\tcfunction_1,\tcfunction_2)$ that maps driving time spent when traversing the path $\apath:=\apath_1\circ\apath_2$ to the \emph{minimum} possible energy consumption (if we ignore battery constraints for now).
Let $\leftintervalborder_1\in\strictposreals$ and $\rightintervalborder_1\in\strictposreals$ denote the minimum and maximum driving time of $\tcfunction_1$, respectively. Similarly, let $\leftintervalborder_2\in\strictposreals$ and $\rightintervalborder_2\in\strictposreals$ denote the corresponding driving times of~$\tcfunction_2$.
Clearly, we obtain $\tcfunction(\atime)=\infty$ for all $\atime<\leftintervalborder_1+\leftintervalborder_2$ and $\tcfunction(\atime)=\tcfunction_1(\rightintervalborder_1)+\tcfunction_2(\rightintervalborder_2)$ for $\atime>\rightintervalborder_1+\rightintervalborder_2$.
For any remaining value~$\atime\in[\leftintervalborder_1+\leftintervalborder_2,\rightintervalborder_1+\rightintervalborder_2]$, we have to determine times $\atime_1\in[\leftintervalborder_1,\rightintervalborder_1]$ and $\atime_2\in[\leftintervalborder_2, \rightintervalborder_2]$ that sum up to~$\atime_1+\atime_2=\atime$ and minimize overall consumption (as in the simple model discussed in Section~\ref{sec:model}).
We set~$\cvalue:=\atime_1$ below, which yields
\begin{align}
 \label{eq:tradeoff-linkedfunction}
 \tcfunction(\atime) = \min_{\substack{\cvalue\in[\leftintervalborder_1,\rightintervalborder_1]\\\cvalue\in[\atime-\rightintervalborder_2,\atime-\leftintervalborder_2]}}\tcfunction_1(\cvalue)+\tcfunction_2(\atime-\cvalue)
\end{align}
for all $\atime\in[\leftintervalborder_1+\leftintervalborder_2,\rightintervalborder_1+\rightintervalborder_2]$. In other words, to minimize the energy consumption for a given time~$\atime\in[\leftintervalborder_1+\leftintervalborder_2,\rightintervalborder_1+\rightintervalborder_2]$, we have to divide the amount of time that exceeds the minimum possible total driving time $\leftintervalborder_1+\leftintervalborder_2$ among the two paths, such that consumption is minimized.
We would like to spend this additional amount of time $\atime-(\leftintervalborder_1+\leftintervalborder_2)$ on the path corresponding to the function with steeper slope, since it provides the better tradeoff (we save more energy per additional unit of time spent on the corresponding arc). This is illustrated in Figure~\ref{fig:tradeoff-example-linking}.
In what follows, we formally derive the link operation.

\begin{figure}[t]
 \centering
 \begin{subfigure}[b]{.29\textwidth}%
 \centering%
 \ConstrainedExampleDrawRealisticTradeoffFunctionFirstEdge
 \caption{}%
 \label{fig:tradeoff-example-linking:first-arc}%
 \end{subfigure}%
 \begin{subfigure}[b]{.29\textwidth}%
 \centering%
 \ConstrainedExampleDrawRealisticTradeoffFunctionSecondEdge
 \caption{}%
 \label{fig:tradeoff-example-linking:second-arc}%
 \end{subfigure}%
 \begin{subfigure}[b]{.42\textwidth}%
 \centering%
 \ConstrainedExampleDrawRealisticTradeoffFunctionLinkedFunction
 \caption{}%
 \label{fig:tradeoff-example-linking:result}%
 \end{subfigure}%
 \caption{Linking two consumption functions. (a)~The consumption function~$\tcfunction_1$, defined by the parameters $\aparam_1=4$,~$\bparam_1=1$,~$\cparam_1=-1$,~$\leftintervalborder_1=2$, and~$\rightintervalborder_1=4$. (b)~The consumption function~$\tcfunction_2$, defined by the parameters $\aparam_2=0.5$,~$\bparam_2=1$,~$\cparam_2=1$,~$\leftintervalborder_2=2$, and~$\rightintervalborder_2=5$. (c)~The consumption function $\tcfunction=\linkop(\tcfunction_1,\tcfunction_2)$, with $\tcfunction(\atime)=\tcfunction_1(\cfunction(\atime))+\tcfunction_2(\atime-\cfunction(\atime))$ on the interval~$[4,9]$. It is defined by three subfunctions with the indicated subdomains~$[4,5]$,~$[5,6.5]$, and~$[6.5,9]$.}
 \label{fig:tradeoff-example-linking}
\end{figure}

\paragraph{Linking Functions Defined by Single Tradeoff Functions.}
For now, assume that each of the given functions $\tcfunction_1$ and $\tcfunction_2$ is defined by a \emph{single} tradeoff subfunction, rather than multiple ones. For example, this is the case if both paths $\apath_1$ and $\apath_2$ consist of single arcs, \ie, $\apath_1=[\vertexa,\vertexb]$ and~$\apath_2=[\vertexb,\vertexc]$ for some $(\vertexa,\vertexb)\in\arcs$ and~$(\vertexb,\vertexc)\in\arcs$.
The result $\tcfunction:=\linkop(\tcfunction_1,\tcfunction_2)$ of linking $\tcfunction_1$ and $\tcfunction_2$ is a piecewise-defined consumption function, which may consist of multiple subfunctions of the form as in Equation~\eqref{eq:tradeoff-functioninnerform}; see Figure~\ref{fig:tradeoff-example-linking} for an example.
Intuitively, the first subfunction of $\tcfunction$ represents a shifted part of the steeper input function for small values of~$\atime$. It is followed by a combination of both functions and a subfunction that corresponds to the input function that is gentler for large values of~$\atime$. Any of these parts may collapse (for example, the combined part only exists if one can pick admissible driving times such that both functions have identical negative slopes).

We now show how $\tcfunction=\linkop(\tcfunction_1,\tcfunction_2)$ is computed. Apparently, the best choice of~$\cvalue$ in Equation~\eqref{eq:tradeoff-linkedfunction} depends on the value~$\atime\in[\leftintervalborder_1+\leftintervalborder_2,\rightintervalborder_1+\rightintervalborder_2]$. Therefore, we consider the \emph{$\cvalue$-function} $\cfunction\colon[\leftintervalborder_1+\leftintervalborder_2,\rightintervalborder_1+\rightintervalborder_2]\to\posreals$ that maps every admissible value of~$\atime$ to the optimal choice of~$\cvalue$. Then, we immediately get for arbitrary $\atime\in\posreals$ that
\begin{align}
 \label{eq:tradeoff-linkedfunction-alt}
 \tcfunction(\atime) =
 \begin{cases}
  \infty                                                                          & \mbox{if } \atime < \leftintervalborder_1 + \leftintervalborder_2, \\
  \tcfunction_1(\rightintervalborder_1) + \tcfunction_2(\rightintervalborder_2)   & \mbox{if } \atime > \rightintervalborder_1 + \rightintervalborder_2, \\
  \tcfunction_1(\cfunction(\atime)) + \tcfunction_2 (\atime - \cfunction(\atime)) & \mbox{otherwise.}
 \end{cases}
\end{align}
Hence, we essentially need to compute~$\cfunction$ to obtain the desired function~$\tcfunction$.
To this end, consider an arbitrary fixed driving time~$\atime\in[\leftintervalborder_1+\leftintervalborder_2,\rightintervalborder_1+\rightintervalborder_2]$. To identify the optimal value $\cfunction(\atime)$, we examine the derivative~$\tcfunction'_\atime$ of the term $\tcfunction_\atime(\cvalue):=\tcfunction_1(\cvalue)+\tcfunction_2(\atime-\cvalue)$.
We obtain a unique zero $\cvalue^*_\atime$ for this derivative under the assumption that $\bparam_1<\cvalue$ and~$\cvalue<\atime-\bparam_2$; see Appendix~\ref{app:linking:arcs} for details.
The value $\cvalue^*_\atime$ minimizes energy consumption for an unrestricted distribution of driving times that sum up to~$\atime$. However, from Equation~\eqref{eq:tradeoff-linkedfunction} we get the additional constraints $\cfunction(\atime)\ge\max\{\leftintervalborder_1,\atime-\rightintervalborder_2\}$ and~$\cfunction(\atime)\le\min\{\rightintervalborder_1,\atime-\leftintervalborder_2\}$.
Since $\cvalue^*_\atime$ is the unique zero of $\tcfunction_\atime'$ in the open interval~$(\bparam_1,\atime-\bparam_2)$, monotonicity of $\tcfunction_\atime$ in the intervals $(\bparam_1,\cvalue^*_\atime]$ and $[\cvalue^*_\atime,\atime-\bparam_2)$ follows and we get
\begin{align}
 \label{eq:tradeoff-cfunction}
 \cfunction(\atime) =
 \begin{cases}
  \max\{\leftintervalborder_1,\atime-\rightintervalborder_2\} & \mbox{if }\cvalue^*_\atime<\max\{\leftintervalborder_1,\atime-\rightintervalborder_2\}\text{,}\\
  \min\{\rightintervalborder_1,\atime-\leftintervalborder_2\} & \mbox{if }\cvalue^*_\atime>\min\{\rightintervalborder_1,\atime-\leftintervalborder_2\}\text{,}\\
  \cvalue^*_\atime & \mbox{otherwise.}
 \end{cases}
\end{align}
Equations~\eqref{eq:tradeoff-linkedfunction-alt} and~\eqref{eq:tradeoff-cfunction} together are sufficient to specify the desired function~$\tcfunction$.

In Appendix~\ref{app:linking:arcs}, we show that the resulting function $\tcfunction$ is indeed a consumption function, \ie, a function that has the general form as in Equation~\eqref{eq:tradeoff-functiongeneralform}, is continuous and decreasing on the interval $[\leftintervalborder_1+\leftintervalborder_2,\infty)$, and whose tradeoff function is defined by a (small) constant number of subfunctions that have the form as in Equation~\eqref{eq:tradeoff-functioninnerform}. Degenerate cases are possible, too. For example, if $\tcfunction_1$ or $\tcfunction_2$ is constant, the link operation becomes a simple shift on the x-axis and~y-axis.
In conclusion, the link operation requires constant time in the special case where both input functions are defined by a single tradeoff function. Lemma~\ref{lem:tc-simplelink} summarizes our results.

\begin{lemma}
\label{lem:tc-simplelink}
 Given two consumption functions~$\tcfunction_1$ and $\tcfunction_2$, each defined by a single tradeoff subfunction as defined in Equation~\ref{eq:tradeoff-functioninnerform}, the link operation $\linkop(\tcfunction_1,\tcfunction_2)$ requires constant time and its result is a consumption function that is uniquely described by at most three tradeoff subfunctions as defined in Equation~\ref{eq:tradeoff-functioninnerform}.
\end{lemma}

\paragraph{Linking General Consumption Functions.}
We tackle the general case, in which we are given two consumption functions $\tcfunction_1$ and~$\tcfunction_2$, each possibly consisting of \emph{multiple} tradeoff subfunctions with the general form as in Equation~\eqref{eq:tradeoff-functioninnerform}, and want to compute the function~$\tcfunction:=\linkop(\tcfunction_1,\tcfunction_2)$.
Consider a tradeoff subfunction $\tcinner$ of $\tcfunction_1$ and its subdomain~$[\leftintervalborder_\tcinner,\rightintervalborder_\tcinner)$. The subfunction $\tcinner$ itself induces a consumption function $\tcfunction_\tcinner$ with minimum driving time $\leftintervalborder_\tcinner$ and maximum driving time~$\rightintervalborder_\tcinner$; see Equation~\eqref{eq:tradeoff-functiongeneralform}.
By the fact that $\tcfunction_1$ is decreasing and by construction of the consumption function~$\tcfunction_\tcinner$, we get $\tcfunction_\tcinner(\atime)\ge\tcfunction_1(\atime)$ for all $\atime\in\posreals$ and $\tcfunction_\tcinner(\atime)=\tcfunction_1(\atime)$ for all~$\atime\in[\leftintervalborder_\tcinner,\rightintervalborder_\tcinner)$.
Since the same argument can be made for subfunctions of~$\tcfunction_2$, it follows directly from Equation~\eqref{eq:tradeoff-linkedfunction} that we obtain $\tcfunction$ after applying the constant-time link operation from Lemma~\ref{lem:tc-simplelink} to all pairs of consumption functions induced by subfunctions of $\tcfunction_1$ and~$\tcfunction_2$. The lower envelope of all resulting \emph{candidate subfunctions} yields the desired function~$\tcfunction$.
Obviously, $\tcfunction$ is again a piecewise function. Moreover, as the lower envelope of functions that are decreasing on a common domain must be decreasing on this domain as well, the function $\tcfunction$ is decreasing.
Finally, the following Lemma~\ref{lem:tradeoff-function-continuous} claims that $\tcfunction$ is also continuous on its interval of admissible driving time (see Appendix~\ref{app:linking:lemma} for a formal proof).
Therefore, the function space of consumption functions is indeed closed under the link operation. Note that the link operation is also commutative and associative.
\begin{lemma}
 \label{lem:tradeoff-function-continuous}
 Given two consumption functions $\tcfunction_1$ and~$\tcfunction_2$, each defined by an arbitrary number of subfunctions as defined in Equation~\ref{eq:tradeoff-functioninnerform}, the function $\tcfunction:=\linkop(\tcfunction_1,\tcfunction_2)$ is continuous on the interval~$[\leftintervalborder,\infty)$, where $\leftintervalborder\in\strictposreals$ denotes the minimum driving time of $\tcfunction$.
\end{lemma}

The running time of the na\"{i}ve link operation described above is quadratic in the number of subfunctions of $\tcfunction_1$ and~$\tcfunction_2$. In Appendix~\ref{app:linking:linear-time}, we show how the complexity of the link operation for general consumption functions can be reduced to linear time.
In our experiments, the number of subfunctions per consumption function was relatively small on average, so the speedup provided by the more sophisticated linear-time method was limited (less than 10\,\%). Hence, the algorithm presented in Appendix~\ref{app:linking:linear-time} may rather be considered to be of theoretical interest.

\subsection{Dominance Tests for Consumption Functions}\label{sec:operations:dominance}

If we are given multiple consumption functions corresponding to paths between the same pair of vertices, we only want to propagate those which represent \emph{optimal} tradeoffs. Assume we are given two consumption functions $\tcfunction_1$ and $\tcfunction_2$ corresponding to two such paths, each possibly defined by \emph{multiple} subfunctions (\cf Section~\ref{sec:operations:linking}). We say that $\tcfunction_1$ \emph{dominates} $\tcfunction_2$ if $\tcfunction_1(\atime)\le\tcfunction_2(\atime)$ holds for all~$\atime\in\posreals$. Apparently, $\tcfunction_2$ cannot be part of an optimal solution in this case, so it would not need to be propagated by the search algorithm.

To efficiently test whether $\tcfunction_1$ dominates~$\tcfunction_2$, we inspect their respective subfunctions $\tcinner_1^1,\dots,\tcinner_1^k$ and~$\tcinner_2^1,\dots,\tcinner_2^\ell$. For some $i\in\{1,\dots,k\}$ and $j\in\{1,\dots,\ell\}$, consider the pair $\tcinner_1^i$ and $\tcinner_2^{\smash{j}}$ of subfunctions, both having the form of Equation~\eqref{eq:tradeoff-functioninnerform}, and let $[\leftintervalborder_1^i,\rightintervalborder_1^i)$ and~$[\leftintervalborder_2^{\smash{j}},\rightintervalborder_2^{\smash{j}})$ denote their respective subdomains.
Observe that if the subdomains do not overlap, dominance can be checked easily by inspecting function values at the endpoints of these subdomains.
Hence, assume that their intersection is not empty, \ie,~$[\leftintervalborder_1^i,\rightintervalborder_1^i)\cap[\leftintervalborder_2^{\smash{j}},\rightintervalborder_2^{\smash{j}})\neq\emptyset$.
We can test in constant time whether $\tcinner_1^i(\atime)\le\tcinner_2^{\smash{j}}(\atime)$ holds for all~$\atime\in[\max\{\leftintervalborder_1^i,\leftintervalborder_2^{\smash{j}}\},\min\{\rightintervalborder_1^i,\rightintervalborder_2^{\smash{j}}\}]$ as follows.
The subfunction $\tcinner_1^i$ dominates the subfunction $\tcinner_2^{\smash{j}}$ in this interval if and only if $\tcinner_1^i(\atime)\le\tcinner_2^{\smash{j}}(\atime)$ holds for the two values $\atime=\max\{\leftintervalborder_1^i,\leftintervalborder_2^{\smash{j}}\}$ and~$\atime=\min\{\rightintervalborder_1^i,\rightintervalborder_2^{\smash{j}}\}$ at the borders of their subdomain intersection, as well as the unique extreme point
of $\tcinner_1^i-\tcinner_2^{\smash{j}}$ (if it falls within the considered intersection of their subdomains; see also Appendix~\ref{app:linking:dominance}).
Since we only have to compare subfunctions if their subdomains intersect, we can test whether $\tcfunction_1$ dominates $\tcfunction_2$ in a single linear scan (comparing subfunctions in increasing order of driving time).
Hence, the running time of a dominance test is linear in the number of subfunctions of $\tcfunction_1$ and~$\tcfunction_2$.

\begin{figure}[t]
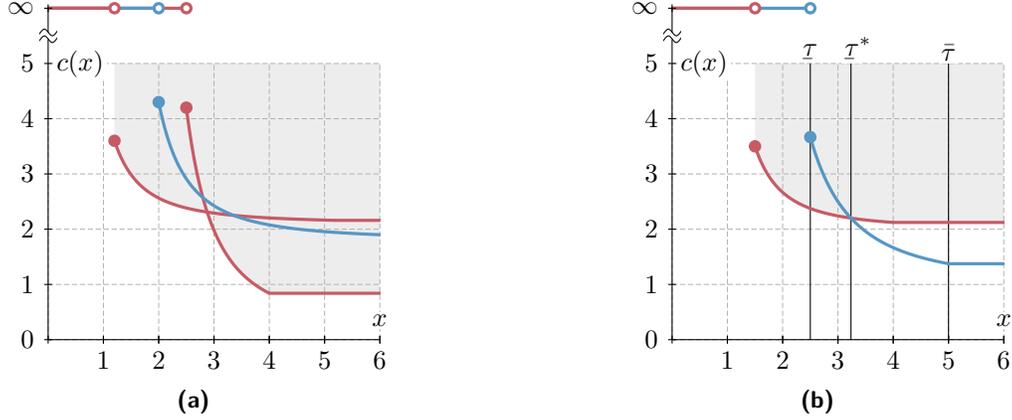

 \centering
 \begin{subfigure}[b]{.5\textwidth}%
 \centering%
 \ConstrainedExampleDrawTradeoffFunctionDominance
 \caption{}%
 \label{fig:tradeoff-example-dominance:set-dominance}%
 \end{subfigure}%
 \begin{subfigure}[b]{.5\textwidth}%
 \centering%
 \ConstrainedExampleDrawTradeoffFunctionPartialDominance
 \caption{}%
 \label{fig:tradeoff-example-dominance:partial-dominance}%
 \end{subfigure}%
 \caption{Dominance of consumption functions. The consumption function $\tcfunction$ (blue) is compared to other functions corresponding to a set of alternative paths (red). Shaded areas indicate function values that are dominated by these other functions. (a) The consumption function $\tcfunction$ is not dominated by any function alone, but by the lower envelope of their two functions. In other words, for each $\atime\in\posreals$, one of the two other functions yields lower consumption. Hence, $\tcfunction$ does not provide an optimal solution for any driving time. (b) The consumption function $\tcfunction$ is partially dominated by the other function, so its minimum driving time can be increased from $\leftintervalborder$ to~$\leftintervalborder^*$.}
 \label{fig:tradeoff-example-dominance}
\end{figure}

\paragraph{Computing Dominated Subfunctions.}
The dominance test described above enables pairwise comparison of consumption functions to identify dominates ones, similar to the pairwise checks between labels in the \gls*{bsp} algorithm.
However, since we are dealing with continuous functions, a consumption function may also be \emph{partially} dominated or dominated only by a \emph{set} $\labelset$ of other consumption functions corresponding to paths to the same vertex; see Figure~\ref{fig:tradeoff-example-dominance}.
In other words, even if $\tcfunction$ is not dominated by any other consumption function alone, it is possible that for some (or even all) admissible driving times $\atime\in[\leftintervalborder,\rightintervalborder]$ of~$\tcfunction$, there is a function $\tcfunction^*\in\labelset$ with~$\tcfunction^*(\atime)\le\tcfunction(\atime)$. We discuss different strategies to efficiently identify dominated subfunctions of~$\tcfunction$.

A common approach in the literature to avoid pairwise dominance checks in related settings is to explicitly compute the \emph{lower envelope} of all functions currently stored at a vertex~\cite{Del09b,Ioa98}. This can also go in hand with a change in the label propagation strategy~\cite{Rai09}.
Unfortunately, managing lower envelopes comprising multiple functions (each consisting of several tradeoff subfunctions) is difficult in our setting, because the quadratic-time link operation described in Section~\ref{sec:operations:linking} becomes too expensive. However, our linear-time operation requires \emph{convex} consumption functions (see Appendix~\ref{app:linking:linear-time}). It is easy to see that the lower envelope of multiple convex consumption function is not convex in general. Although one can generalize the link operation to handle non-convex functions (by essentially dividing it into convex subfunctions), label propagation becomes more expensive in theory and in practice~\cite{And15}.
Another important aspect that makes this approach much more complicated in practice is the fact that computing the lower envelope of multiple functions requires us to identify their \emph{intersection points}, \ie, given two consumption functions $\tcfunction_1\in\labelset$ and~$\tcfunction_2\in\labelset$, we have to find driving times $\atime\in\posreals$ with~$\tcfunction_1(\atime)=\tcfunction_2(\atime)$. While this equation has at most two solutions, it is relatively expensive to solve, as expanding the above equation results in a large number of terms.
In addition to that, computing intersection points is prone to rounding errors and numerical instability, which can have a significant effect on maintainability and the quality of solutions in practice.

Instead, we propose an alternative approach that avoids pairwise dominance checks, but does not require us to explicitly compute and maintain the lower envelope of multiple consumption functions. It tries to identify dominated \emph{parts} of the function $\tcfunction$ in question.
To this end, we compute a value $\leftintervalborder^*\in[\leftintervalborder,\rightintervalborder]$ such that for all $\atime\in[\leftintervalborder,\leftintervalborder^*]$, there is at least one function $\tcfunction^*$ in the set $\labelset$ with~$\tcfunction^*(\atime)\le\tcfunction(\atime)$.
This value $\leftintervalborder^*$ can be computed in a single (coordinated) linear scan over the subfunctions of $\tcfunction$ and \emph{all} subfunctions of the functions in~$\labelset$.
If we obtain $\leftintervalborder^*=\rightintervalborder$ in this scan, $\tcfunction$ is dominated for all~$\atime\in\posreals$, so it cannot be part of an optimal solution.
Otherwise, we set~$\leftintervalborder=\leftintervalborder^*$. Analogously, we then compute a value $\rightintervalborder^*\in[\leftintervalborder^*,\rightintervalborder]$ such that $\tcfunction$ is dominated for all $\atime\in[\rightintervalborder^*,\rightintervalborder]$ and set~$\rightintervalborder=\rightintervalborder^*$. Afterwards, we discard any subfunction of $\tcfunction$ whose subdomain does not intersect~$[\leftintervalborder^*,\rightintervalborder^*]$.

To compute the values of $\leftintervalborder^*$ and~$\rightintervalborder^*$, we still face the problem that we may have to compute the intersection points of (subfunctions of) $\tcfunction$ and $\tcfunction^*\in\labelset$. To avoid the costly and error-prone computation of intersection points, our implementation computes \emph{approximate} values $\leftintervalborder^*$ and~$\rightintervalborder^*$.
It uses the simple test described before to identify subfunctions that dominate a subfunction of $\tcfunction$ in the \emph{whole} intersection of their respective subdomains. Clearly, we do not necessarily find the optimal values for $\leftintervalborder^*$ and $\rightintervalborder^*$ that way. Furthermore, note that there may still exist values $\atime$ in the open interval $(\leftintervalborder^*,\rightintervalborder^*)$ with $\tcfunction^*(\atime)\le\tcfunction(\atime)$ for some function~$\tcfunction^*\in\labelset$, which we could also aim to identify~\cite{He18}. However, as argued above, the fact that we avoid the computation of intersection points makes our dominance test efficient while keeping it relatively easy to implement and maintain.

\section{The Tradeoff Function Propagating Algorithm}\label{sec:approach}

We are now ready to describe our \emph{\gls*{tfp}} algorithm, which utilizes the tools from Section~\ref{sec:operations} to generalize the exponential-time \gls*{bsp} algorithm and solve~\gls*{evcas}.
For a source~$\source\in\vertices$, a target~$\target\in\vertices$, and an initial \gls*{soc} $\soc_\source\in[0,\maxbattery]$, it computes the fastest $\source$--$\target$ path and optimal driving times such that battery constraints are respected.

\begin{algorithm}[tp]
  \caption{Pseudocode of the \acrshort*{tfp} algorithm.}\label{alg:tfp}%
  \tcp{initialize label sets}%
  \ForEach{$\vertex\in\vertices$}{\label{line:tfp:init}%
    $\labelset_{\operatorname{set}}(\vertex)\assign\emptyset$\;%
    $\labelset_{\operatorname{uns}}(\vertex)\assign\emptyset$\;%
  }%
  $\labelset_{\operatorname{uns}}(\source)\assign\{\tcfunction_\source\equiv\maxbattery-\soc_\source\}$\;\label{line:tfp:initsourcelabel}%
  $\queue$.\queueInsert{$\source,\queuekey(\source,\tcfunction_\source)$}\;\label{line:tfp:initsource}%
  \BlankLine
  \tcp{run main loop}%
  \While{$\queue$.\queueIsNotEmpty{}}{%
    \BlankLine
    \tcp{extract next vertex}%
    $\vertexa\assign\queue$.\queueMinElement{}\;\label{line:tfp:extractmin}%
    $\tcfunction_\vertexa\assign\labelset_{\operatorname{uns}}(\vertexa)$.\queueDeleteMin{}\;%
    $\labelset_{\operatorname{set}}(\vertexa)$.\queueInsert{$\tcfunction_\vertexa$}\;\label{line:tfp:settle}%
    \tcp{return if target found}%
    \If{$\vertexa=\target$}{%
      \BlankLine
      \tcp{see Section~\ref{sec:approach:paths} for full solution retrieval}%
      \Return $\queuekey(\vertexa,\tcfunction_\vertexa)$\;\label{line:tfp:stoppingcriterion}%
      }%
    \BlankLine
    \tcp{update priority queue}%
    \eIf{$\labelset_{\operatorname{uns}}(\vertexa)$.\queueIsNotEmpty{}}{\label{line:tfp:updatemin}%
      $\tcfunction\assign\labelset_{\operatorname{uns}}(\vertexa)$.\queueMinElement{}\;%
      $\queue$.\queueUpdate{$\vertexa,\queuekey(\vertexa,\tcfunction)$}\;%
    }{%
      $\queue$.\queueDeleteMin{}\;\label{line:tfp:deletemin}%
    }%
    \BlankLine
    \tcp{scan outgoing arcs}%
    \ForEach{$(\vertexa,\vertexb)\in\arcs$}{\label{line:tfp:scanarcs}%
      $\tcfunction\assign\linkop(\tcfunction_\vertexa,\tcfunction_{(\vertexa,\vertexb)})$\;%
      $\leftintervalborder_\tcfunction\assign\min(\{\infty\}\cup\{\atime\in\posreals\mid\tcfunction(\atime)\le\maxbattery\})$\;\label{line:tfp:cropleftborder}
      $\rightintervalborder_\tcfunction\assign\max(\{\leftintervalborder_\tcfunction\}\cup\{\atime\in\posreals\mid\tcfunction(\atime)\ge0\})$\;\label{line:tfp:croprightborder}
      \If{$\tcfunction\not\equiv\infty$}{%
        $\labelset_{\operatorname{uns}}(\vertexb)$.\queueInsert{$\tcfunction$}\;%
        \If{$\tcfunction=\labelset_{\operatorname{uns}}(\vertexb)$.\queueMinElement{}}{%
         $\queue$.\queueUpdate{$\vertexb,\queuekey(\vertexb,\tcfunction)$}\;\label{line:tfp:inserthead}%
        }%
      }%
    }%
  }%
\end{algorithm}

\subsection{Algorithm Description}\label{sec:approach:description}

Algorithm~\ref{alg:tfp} shows pseudocode of~\gls*{tfp}. The algorithm exploits the fact that the \gls*{soc} at the source is part of the query input. Rather than computing bivariate \gls*{soc} functions (\cf Equation~\eqref{eq:tradeoff-socfunction} in Section~\ref{sec:model:evcas} and the example from Section~\ref{sec:model:example}), this enables us to propagate univariate consumption functions (defined by sequences of tradeoff subfunctions) and ensure on-the-fly that battery constraints are not violated.

\paragraph{Data Structures.}
To keep the number of label comparisons small, we use the same method as in a previous publication~\cite{Bau18}, in which each vertex $\vertex\in\vertices$ maintains a set~$\labelset_\text{set}(\vertex)$ and a heap~$\labelset_\text{uns}(\vertex)$ containing settled (\ie, extracted) and unsettled labels, respectively.
Each label is a (piecewise) consumption function, mapping the driving time on a certain $\source$--$\vertex$~path to energy consumption.
The \emph{key} of a label $\tcfunction$ at the vertex $\vertex$ is defined as
\begin{align}
 \label{eq:queuekey-tfp}
 \queuekey(\vertex,\tcfunction):=\leftintervalborder_\tcfunction,
\end{align}
\ie, its minimum driving time. The corresponding maximum energy consumption $\tcfunction(\leftintervalborder_\tcfunction)$ is used to break ties. (Observe that the key does not depend on the vertex~$\vertex$, however, we will need $\vertex$ as a parameter of the key function when we modify it in Section~\ref{sec:astar}.)

As in Baum~et~al.~\cite{Bau18}, we maintain the invariant that for each~$\vertex\in\vertices$, $\labelset_\text{uns}(\vertex)$ is empty or the unsettled label in~$\labelset_\text{uns}(\vertex)$ with \emph{minimum} key is not dominated by any settled label in~$\labelset_\text{set}(\vertex)$.
New labels are pushed into~$\labelset_\text{uns}(\vertex)$. Whenever the minimum element of~$\labelset_\text{uns}(\vertex)$ changes (because an element is added or extracted), we check whether the new minimum element is dominated by some settled label in~$\labelset_\text{set}(\vertex)$ and discard it in this case.
The algorithm uses a priority queue to determine the next vertex to be scanned. Note that unlike the \gls*{bsp} algorithm, our priority queue contains vertices rather than labels. This does not alter the order in which labels are scanned by the algorithm, but the number of entries in the priority queue is smaller.

In summary, we use the following data structures: The label sets $\labelset_\text{uns}(\cdot)$ and $\labelset_\text{set}(\cdot)$ store the consumption functions of a vertex (keys are calculated on-the-fly from these functions if required, so they do not have to be stored explicitly), whereas the priority queue operates on pairs of vertices and keys.

\paragraph{Label Propagation.}
Label sets and the priority queue are initialized in lines~\ref{line:tfp:init}--\ref{line:tfp:initsource} of Algorithm~\ref{alg:tfp}. Initially, all label sets are empty, except for the constant function $\tcfunction_\source\equiv\maxbattery-\soc_\source$ at the source~$\source$ (the admissible driving time is $0$ and the ``consumption'' at $\source$ equals the difference between the battery capacity and the initial~\gls*{soc}).
The source vertex is also added to the priority queue, which uses the minimum key among any unsettled labels of a vertex as its key; \cf Equation~\eqref{eq:queuekey-tfp}.

In each step of its main loop, the algorithm extracts some vertex~$\vertexa\in\vertices$ with minimum key from the priority queue and settles the corresponding label $\tcfunction_\vertexa$ by taking it from the set $\labelset_\text{uns}(\vertexa)$ and inserting it into~$\labelset_\text{set}(\vertexa)$ (lines~\ref{line:tfp:extractmin}--\ref{line:tfp:settle} of Algorithm~\ref{alg:tfp}). Afterwards, the priority queue and the set $\labelset_\text{uns}(\vertexa)$ are updated accordingly~(lines~\ref{line:tfp:updatemin}--\ref{line:tfp:deletemin}).
Finally, outgoing arcs are scanned (lines~\ref{line:tfp:scanarcs}--\ref{line:tfp:inserthead}).
For every arc~$(\vertexa,\vertexb)\in\arcs$, the function $\tcfunction:=\linkop(\tcfunction_\vertexa,\tcfunction_{(\vertexa,\vertexb)})$ is computed.
Note that~$\tcfunction$ may violate battery constraints, so we set the minimum driving time $\leftintervalborder_\tcfunction$ to the smallest driving time value that yields a consumption of at most~$\maxbattery$ in line~\ref{line:tfp:cropleftborder} of Algorithm~\ref{alg:tfp} (if no such value exists, we set $\leftintervalborder_\tcfunction=\infty$ to indicate that there are no admissible driving times). Similarly, we identify~$\leftintervalborder_\tcfunction$, the largest nonnegative consumption value (a consumption of $0$ corresponds to a fully charge battery) in line~\ref{line:tfp:croprightborder}.
To find these values efficiently in practice, we first check whether $0$ or $\maxbattery$ are in the domain of the inverse $\tcfunction^{-1}$ of~$\tcfunction$. If this is the case, we set the minimum driving time of $\tcfunction$ to $\leftintervalborder_\tcfunction=\tcfunction^{-1}(\maxbattery)$ or the maximum driving time to~$\rightintervalborder_\tcfunction=\tcfunction^{-1}(0)$, respectively.
Observe that this corresponds to setting $\tcfunction(\atime)=\infty$ for all driving times $\atime\in\posreals$ with $\tcfunction(\atime)>\maxbattery$ and $\tcfunction(\atime)=0$ for all $\atime\in\posreals$ with~$\tcfunction(\atime)<0$.
If the resulting function yields finite consumption for some driving time, it is added to $\labelset_\text{uns}(\vertexb)$ and the key of $\vertexb$ in the priority queue is updated, if necessary.

Note that because minimum driving times are strictly positive, the algorithm is \emph{label setting}, \ie, labels extracted from the queue are never dominated by another label later on. Thus, an optimal path is found as soon as a label at~$\target$ is extracted. The minimum driving time of this label is the optimal driving time; see line~\ref{line:tfp:stoppingcriterion} of Algorithm~\ref{alg:tfp}.

\paragraph{Dominance Tests.}
In accordance with Section~\ref{sec:operations:dominance}, we distinguish different ways to identify dominated labels. First, our basic \gls*{tfp} algorithm implements the \emph{lightweight} pairwise dominance tests described at the beginning of Section~\ref{sec:operations:dominance}. For a consumption function $\tcfunction$ in the set $\labelset_\text{uns}(\vertex)$ of some vertex~$\vertex\in\vertices$, we perform a pairwise comparison with each function in $\labelset_\text{set}(\vertex)$ to determine whether $\tcfunction$ is dominated by one of them.
Recall that in doing so, the algorithm may miss that $\tcfunction$ is dominated by a set of other labels (as in Figure~\ref{fig:tradeoff-example-dominance:set-dominance}) or partially dominated (as in Figure~\ref{fig:tradeoff-example-dominance:partial-dominance}).

Second, our \emph{improved} variant computes values $\leftintervalborder^*\in\posreals$ and $\rightintervalborder^*\in\posreals$ as described in Section~\ref{sec:operations:dominance}, such that for each $\atime<\leftintervalborder^*$ and~$\atime>\rightintervalborder^*$, there is at least one function $\tcfunction^*$ contained in the set $\labelset_\text{set}(\vertex)$ with~$\tcfunction^*(\atime)\le\tcfunction(\atime)$.
Note that after this procedure, even if $\tcfunction$ is not discarded (because it is not dominated for all~$\atime\in\posreals$), it may no longer be the function with minimum key in~$\labelset_\text{uns}(\vertex)$, in which case we repeat the dominance check for the new minimum element of~$\labelset_\text{uns}(\vertex)$.
Using the improved dominance check greatly reduces the number of labels in practice, at the cost of (little) additional overhead per dominance test (we now require a coordinated scan over multiple subfunctions).

Recall from Section~\ref{sec:operations:dominance} that we avoid exact computation of $\leftintervalborder^*$ and $\rightintervalborder^*$ for the sake of numerical stability. To distinguish, we call this variant of our improved dominance check \emph{stable}, as opposed to the \emph{exact} variant that requires intersection points to find the actual values of $\leftintervalborder^*$ and $\rightintervalborder^*$. Note that all three variants (lightweight, improved stable, and improved exact) maintain correctness of the algorithm, but the number of labels and subfunctions the algorithm has to propagate may differ.

\subsection{Path Retrieval}\label{sec:approach:paths}

To obtain the actual $\source$--$\target$~path, each label maintains two pointers to its corresponding parent label and vertex.
To retrieve the optimal driving time (and speed) per arc, we explicitly store with each label $\tcfunction$ the function $\cfunctionsecond$, given as $\cfunctionsecond(\atime):=\atime-\cfunction(\atime)$, \wrt the previous link operation that led to the creation of~$\tcfunction$ (see Section~\ref{sec:operations:linking} and Appendix~\ref{app:approach:paths}). It is defined on the domain $[\leftintervalborder,\rightintervalborder]$ induced by the minimum driving time $\leftintervalborder\in\strictposreals$ and the maximum driving time $\rightintervalborder\in\strictposreals$ of $\tcfunction$.

After the search has found the target~$\target$, the following backtracking routine yields the path itself and driving times on all arcs.
We start backtracking from~$\target$. Let~$\atime=\leftintervalborder_\tcfunction$ denote the minimum accumulated driving time of the label~$\tcfunction$ extracted at~$\target$ and let $\vertex\in\vertices$ be the parent vertex of~$\tcfunction$. Then the driving time on the arc $(\vertex,\target)$ is $\cfunctionsecond(\atime)$, where $\cfunctionsecond$ is the function stored with the label~$\tcfunction$. We continue this procedure recursively at its parent label and for the driving time value~$\atime-\cfunctionsecond(\atime)$, until the source vertex $\source$ is reached.

\subsection{A Polynomial-Time Heuristic}\label{sec:approach:heuristic}

Even with improved dominance checks, \gls*{tfp} has exponential running time.
However, the algorithm can easily be extended to a more efficient heuristic search, at the cost of inexact results. We propose a polynomial-time approach that is based on $\varepsilon$-dominance~\cite{Bat11}.
During the search, when performing the dominance test for a label $\tcfunction\in\labelset_\text{uns}(\vertex)$ in the unsettled label set of some vertex~$\vertex\in\vertices$, it is kept in $\labelset_\text{uns}(\vertex)$ only if it yields an improvement over settled labels in $\labelset_\text{set}(\vertex)$ by at least a certain fraction $\varepsilon\maxbattery$, with~$\varepsilon\in(0,1]$, for some driving time.
Thus, when identifying dominated parts of~$\tcfunction$, we test for each driving time $\atime\in\posreals$ whether $\tcfunction^*(\atime)\le\tcfunction(\atime)+\varepsilon\maxbattery$ holds for some settled label~$\tcfunction^*\in\labelset_\text{set}(\vertex)$.
Lemma~\ref{lem:poly-label-size-heuristic} shows that this implies that the number of settled labels per label set can become at most~$\lceil1/\varepsilon\rceil+1$, provided that exact improved dominance checks are applied. Given that the algorithm is label setting, each of these labels is extracted from the priority queue and added to $\labelset_\text{set}(\cdot)$ at most once, so this yields polynomial running time in $|\vertices|$ and~$\lceil1/\varepsilon\rceil$. For a formal proof of Lemma~\ref{lem:poly-label-size-heuristic}, see Appendix~\ref{app:approach:lemma}.

\begin{lemma}
 \label{lem:poly-label-size-heuristic}
 The number of settled labels contained in the set $\labelset_\text{set}(\vertex)$ of each vertex $\vertex\in\vertices$ is at most $\lceil1/\varepsilon\rceil+1$ when running \gls*{tfp} with $\varepsilon$-dominance and exact improved dominance checks.
\end{lemma}

\section{A* Search}\label{sec:astar}

Even the heuristic variant of the \gls*{tfp} algorithm described above remains impractical for realistic ranges and reasonable values of~$\varepsilon$.
Therefore, this section proposes extensions of A*~search to speedup \gls*{tfp} and its heuristic variant, whereas Section~\ref{sec:ch} discusses extensions of~\gls*{ch}.

\paragraph{Potential Shifting.}
A well-known approach to reduce (practical) running times in multicriteria settings is the adaptation of A*~search~\cite{Har68,Man10}.
The key idea of this approach is to \emph{guide} the search towards the target by assigning smaller keys (in the priority queue) to vertices that are closer to the target.
This is achieved by making use of a \emph{potential function}~$\potential\colon\vertices\to\reals$ on the vertices~\cite{Joh77}. The potential function is called~\emph{consistent} (\wrt a cost function $\costfunction\colon\arcs\to\reals$ on the arcs) if $\costfunction(\vertexa, \vertexb)-\potential(\vertexa)+\potential(\vertexb)\geq0$ for all~$(\vertexa,\vertexb)\in\arcs$.
Vertex potentials are added to the keys of labels in the priority queue. Using a consistent potential function, the label-setting property of search algorithms is maintained.

In Section~\ref{sec:astar:single-criterion} and Section~\ref{sec:astar:pwl-functions} below, we propose two variants of A*~search that compute a consistent potential function at query time, prior to running~\gls*{tfp}.
We use these potentials to modify the keys of labels when running Algorithm~\ref{alg:tfp}; \cf Equation~\eqref{eq:queuekey-tfp}.

\subsection{Potentials Based on Single-Criterion Search}\label{sec:astar:single-criterion}

Similar to the approach by Tung and Chew~\cite{Tun92}, our first variant aims at directing the search towards the target by preferring arcs on (unrestricted) fastest paths.
Before running~\gls*{tfp}, we execute a single backward search (\ie, Dijkstra's algorithm running on the \emph{reverse} input graph~$\backwardgraph$, which has all arc directions reversed) from the target~$\target\in\vertices$. This search uses the cost function $\drivingtimelowerboundfunction\colon\arcs\to\posreals$ with $\drivingtimelowerboundfunction(\arc):=\leftintervalborder_\arc$ for all $\arc\in\arcs$, representing minimum driving time on an arc.
This yields, for each vertex~$\vertex\in\vertices$, the minimum unrestricted driving time from $\vertex$ to $\target$, denoted~$\distance_{\drivingtimelowerboundfunction}(\vertex,\target)$, which is a \emph{lower bound} on the constrained driving time from $\vertex$ to~$\target$.
We obtain a potential function $\drivingtimepotential\colon\vertices\to\posreals$ by setting $\drivingtimepotential(\vertex):=\distance_{\drivingtimelowerboundfunction}(\vertex,\target)$. By the triangle inequality, the potential function $\drivingtimepotential$ fulfills the condition $\atime-\drivingtimepotential(\vertexa)+\drivingtimepotential(\vertexb)\ge0$ for all arcs $(\vertexa,\vertexb)\in\arcs$ and all \emph{admissible} driving times~$\atime\in[\leftintervalborder_\arc,\rightintervalborder_\arc]$.
Thus, \gls*{tfp} is label setting when using~$\drivingtimepotential$, which implies that correctness of the algorithm is maintained~\cite{Joh77}.
We replace the key function from Equation~\eqref{eq:queuekey-tfp} by
\begin{align*}
 \queuekey(\vertex,\tcfunction):=\leftintervalborder_\tcfunction+\drivingtimepotential(\vertex)
\end{align*}
for a given consumption function $\tcfunction$ with minimum driving time $\leftintervalborder_\tcfunction$ at some vertex~$\vertex\in\vertices$.

\paragraph{Improvements.}
To compute~$\drivingtimepotential$, the backward search described above visits all vertices in the (reverse) graph. This can be wasteful, especially for small vehicle ranges.
To avoid scans of unreachable vertices, we make use of a second cost function $\consumptionlowerboundfunction\colon\arcs\to\reals$ with $\consumptionlowerboundfunction(\arc):=\tcfunction_\arc(\rightintervalborder_\arc)$ for all~$\arc\in\arcs$, which yields the minimum energy consumption of an arc. We run a backward search on $\backwardgraph$ using this cost function to compute, for all~$\vertex\in\vertices$, lower bounds $\distance_{\consumptionlowerboundfunction}(\vertex,\target)$ on the energy consumption required to reach $\target$ from~$\vertex$. As some arc costs may be negative, we use a \emph{label-correcting} variant of Dijkstra's algorithm, \ie, a variant in which the label of a vertex is not necessarily final when it is extracted from the queue (the vertex may be reinserted and extracted from the queue several times). This algorithm has exponential running time in the worst case~\cite{Joh73}, but it performs well in practice if the number arcs with negative costs is small\footnote{Note that we could also apply the shifting technique of Johnson~\cite{Joh77} to make the search label setting and restore polynomial running time if its performance became an issue.}~\cite{Art10a,Bau13a,Bau18}.
We \emph{prune} the search, \ie, we do not relax any outgoing arcs, whenever the distance label of some scanned vertex exceeds the battery capacity~$\maxbattery$.
Then, we run Dijkstra's algorithm on $\backwardgraph$ to compute~$\drivingtimepotential(\vertex)$ as before, but restrict the search to vertices for which lower bound on energy consumption is below~$\maxbattery$.
Afterwards, we restrict the \gls*{tfp} search to the same set of vertices.
Note that a similar approach can be used to compute vertex potentials for the \gls*{bsp} algorithm when solving~\gls*{evcsp} on a graph with multi-arcs and scalar costs~\cite{Bau14a}.

Additionally, we use lower bounds~$\distance_{\consumptionlowerboundfunction}(\cdot,\cdot)$ for pruning in~\gls*{tfp}:
Before adding a new label $\tcfunction$ to the label set of some vertex~$\vertex\in\vertices$, we first set $\tcfunction(\atime)=\infty$ for all driving times $\atime\in\posreals$ with~$\tcfunction(\atime)+\distance_{\consumptionlowerboundfunction}(\vertex,\target)>\maxbattery$. To do so, we check (in a linear scan over the subfunctions defining the inverse function~$\tcfunction^{-1}$) whether $\maxbattery-\distance_{\consumptionlowerboundfunction}(\vertex,\target)$ is in the domain of $\tcfunction^{-1}$ and set the minimum driving time of $\tcfunction$ to~$\tcfunction^{-1}(\maxbattery-\distance_{\consumptionlowerboundfunction}(\vertex,\target))$ if this is the case. Otherwise, we either obtain $\tcfunction\equiv\infty$ or the function $\tcfunction$ remains unchanged.

\subsection{Potentials Based on Bound Function Propagation}\label{sec:astar:pwl-functions}

The potential function~$\drivingtimepotential$ may be too conservative if the consumption on the fastest path is very high.
In such cases, it pays off to adapt the potential function~$\convexpotential\colon\vertices\times[0,\maxbattery]\cup\{-\infty\}\to\posreals\cup\{\infty\}$ introduced in a previous work~\cite{Bau18}. The potential function $\convexpotential$ incorporates the current \gls*{soc} at a vertex to provide more accurate bounds. For each vertex~$\vertex\in\vertices$, it uses a convex, piecewise linear function~$\convexlowerboundfunction_\vertex\colon\reals\to\posreals\cup\{\infty\}$ that maps the current \gls*{soc} at~$\vertex$ to a lower bound on the remaining driving time from~$\vertex$ to~$\target$. As higher \gls*{soc} allows the \gls*{ev} to drive faster, $\convexlowerboundfunction_\vertex$ is also \emph{decreasing} \wrt \gls*{soc}.
We say that the potential function $\convexpotential$ is \emph{consistent} if the condition
\begin{align*}
 \atime-\convexpotential(\vertexa,\soc)+\convexpotential(\vertexb,\socprofile_{(\vertexa,\vertexb)}(\atime,\soc))\ge0 
\end{align*}
holds for all~$(\vertexa,\vertexb)\in\arcs$,~$\atime\in[\leftintervalborder_{(\vertexa,\vertexb)},\rightintervalborder_{(\vertexa,\vertexb)}]$, and~$\soc\in[0,\maxbattery]$.
In other words, the above term must be nonnegative for any admissible parameter choice for the \gls*{soc} function~$\socprofile_{(\vertexa,\vertexb)}$. We define~$\convexpotential(\vertexa,-\infty):=\infty$. We further require that the potential $\convexpotential(\target,\soc)$ at the target $\target\in\vertices$ equals $0$ for all~$\soc\in[0,\maxbattery]$, which implies that consistent potentials maintain correctness of \gls*{tfp}; see a previous work~\cite{Bau18} for details on consistency of vertex potentials that  incorporate~\gls*{soc}.

\paragraph{Computing Vertex Potentials.}
The functions representing the potentials $\convexpotential(\cdot,\cdot)$ are determined in a label-correcting backward search from~$\target$.
This search operates on the reverse graph $\backwardgraph=(\vertices,\backwardarcs)$ of the input graph and maintains a piecewise linear function $\convexlowerboundfunction_\vertex\colon\reals\to\posreals\cup\{\infty\}$ for each~$\vertex\in\vertices$, represented by a sequence~$\convexfunctionbreakpoints_\vertex=[(\soc_1,\atime_1),\dots,(\soc_k,\atime_k)]$ of breakpoints, such that $\soc_i<\soc_j$ for $i<j$, $\convexlowerboundfunction_\vertex(\soc)=\infty$ for $\soc<\soc_1$, and $\convexlowerboundfunction_\vertex(\soc)=\atime_k$ for $\soc\ge\soc_k$. For arbitrary values $\soc_{i}\le\soc<\soc_{i+1}$, with~$i\in\{1,\dots,k-1\}$, we evaluate $\convexlowerboundfunction_\vertex$ by linear interpolation between two consecutive breakpoints.
We ignore battery constraints in the search, so we obtain lower bounds on \gls*{soc} values, which can also become negative. See Appendix~\ref{app:astar:backward-search} for the pseudocode.

Each vertex stores a \emph{single} label consisting of its (tentative) lower bound function.
The search is initialized with a function $\convexlowerboundfunction_\target$ at the target~$\target\in\vertices$ that evaluates to~$0$ for arbitrary nonnegative~\gls*{soc}.
Labels are extracted in increasing order of their minimum driving time. Given the label extracted at some~$\vertexa\in\vertices$, the search \emph{links} the lower bound in this label with each outgoing arc~$(\vertexa,\vertexb)\in\arcs$.
To this end, it first \emph{converts} the consumption function~$\tcfunction_{(\vertexa,\vertexb)}$ mapping driving time to energy consumption to a piecewise linear function $\convexlowerboundfunction_{(\vertexa,\vertexb)}$ mapping \gls*{soc} to a lower bound on driving time.
Let $\leftintervalborder\in\strictposreals$ and~$\rightintervalborder\in\strictposreals$ denote the minimum and maximum driving time of $\tcfunction_{(\vertexa,\vertexb)}$, respectively.
If $\tcfunction_{(\vertexa,\vertexb)}$ is constant (\ie,~$\leftintervalborder=\rightintervalborder$), we immediately obtain the lower bound defined by~$\convexfunctionbreakpoints_{(\vertexa,\vertexb)}:=[(\tcfunction_{(\vertexa,\vertexb)}(\leftintervalborder),\leftintervalborder)]$.
Otherwise, we consider two geometric lines~$\aline_1$ and $\aline_2$ defined as follows; see also Figure~\ref{fig:tradeoff-bound-simplification}. The first line $\aline_1$ passes through the point $\pointa=(\tcfunction_{(\vertexa,\vertexb)}(\rightintervalborder),\rightintervalborder)$ with slope equal to the (right) derivative of the inverse $\tcfunction_{(\vertexa,\vertexb)}^{-1}$ of $\tcfunction_{(\vertexa,\vertexb)}$ at~$\tcfunction_{(\vertexa,\vertexb)}(\rightintervalborder)$.
The second line $\aline_2$ goes through $\pointb=(\tcfunction_{(\vertexa,\vertexb)}(\leftintervalborder),\leftintervalborder)$ and its slope equals the (left) derivative of $\tcfunction_{(\vertexa,\vertexb)}^{-1}$ at~$\tcfunction_{(\vertexa,\vertexb)}(\leftintervalborder)$.
Since $\leftintervalborder\neq\rightintervalborder$ and $\tcfunction_{(\vertexa,\vertexb)}^{-1}$ is convex on its domain~$[\tcfunction_{(\vertexa,\vertexb)}(\rightintervalborder),\tcfunction_{(\vertexa,\vertexb)}(\leftintervalborder)]$, the lines $\aline_1$ and $\aline_2$ intersect in a unique point~$\pointc\in\reals^2$.
Then, a convex lower bound $\convexlowerboundfunction_{(\vertexa,\vertexb)}$ is defined by the breakpoints~$[\pointa,\pointc,\pointb]$.
To increase accuracy of the lower bound, we may repeat this operation recursively for initial points $\pointa$ or $\pointb$ together with the unique point $\pointc'\in\reals^2$ on the inverse of $\tcfunction_{(\vertexa,\vertexb)}$ that has the same c-coordinate as $\pointc$; see Figure~\ref{fig:tradeoff-bound-simplification}. Recursion stops as soon as the maximum difference between $\tcfunction_{(\vertexa,\vertexb)}^{-1}$ and the function $\convexlowerboundfunction_{(\vertexa,\vertexb)}$ induced by the current sequence of points falls below a predefined threshold~$\varepsilon>0$. Note that this difference is obtained in constant time, since it always occurs in the latest point $\pointc$ that was added to~$\convexlowerboundfunction_{(\vertexa,\vertexb)}$.

\begin{figure}[t]
 \centering%
 \ConstrainedExampleDrawTradeoffFunctionSimplification
 \caption{Constructing the lower bound function $\convexlowerboundfunction_{(\vertexa,\vertexb)}$. It is defined by the endpoints $\pointa$ and $\pointb$ of the inverse of the consumption function $\tcfunction_{(\vertexa,\vertexb)}$ with $\leftintervalborder=2$ and~$\rightintervalborder=6$, and the intersection of the tangents $\aline_1$ and~$\aline_2$. The value $\varepsilon$ indicates the maximum error of the lower bound. Recursion starts at~$\pointc'$.}%
\label{fig:tradeoff-bound-simplification}%
\end{figure}

We link the function~$\convexlowerboundfunction_\vertexa$ at~$\vertexa$ with the resulting lower bound $\convexlowerboundfunction_{(\vertexa,\vertexb)}$ by computing, for any \gls*{soc} value~$\soc\in\reals$, an optimal distribution of the available amount $\soc$ of energy among the arc $(\vertexa,\vertexb)$ and the remaining $\vertexa$--$\target$~path. Formally, the resulting function $\convexlowerboundfunction=\linkop(\convexlowerboundfunction_\vertexa,\convexlowerboundfunction_{(\vertexa,\vertexb)})$ evaluates to
\begin{align*}
\convexlowerboundfunction(\soc):=\min_{\soc^*\in\reals}\convexlowerboundfunction_{(\vertexa,\vertexb)}(\soc^*)+\convexlowerboundfunction_\vertexa(\soc-\soc^*) 
\end{align*}
for all~$\soc\in\reals$.
This function can be computed in a linear scan over the breakpoints of both functions~\cite{Bau18}.
To improve running times in practice, we can discard the label $\convexlowerboundfunction$ at this point if $\convexlowerboundfunction(\soc)>\maxbattery$ for all~$\soc\in\reals$.

After scanning the arc, we \emph{merge} the resulting function~$\convexlowerboundfunction$ with the function~$\convexlowerboundfunction_\vertexb$ in the label of~$\vertexb$, \ie, we compute the pointwise minimum $\convexlowerboundfunction_\vertexb=\mergeop(\convexlowerboundfunction_\vertexb,\convexlowerboundfunction)$ in a linear scan over the breakpoints of $\convexlowerboundfunction_\vertexb$ and~$\convexlowerboundfunction$.
To ensure that the resulting function $\convexlowerboundfunction_\vertexb$ is again convex, we apply the linear-time scan of Graham~\cite{Gra72}: Starting with an empty sequence of breakpoints, it iteratively appends the breakpoints of $\convexlowerboundfunction_\vertexb$ in increasing order of driving time, each time removing previously added breakpoints from the tail of the sequence if that is necessary to maintain the function's convexity.

\paragraph{Running \gls*{tfp} with Vertex Potentials.}
Once the backward search has terminated, we define $\convexpotential(\vertex,\soc):=\convexlowerboundfunction_\vertex(\soc)$ for all $\vertex\in\vertices$ and $\soc\in[0,\maxbattery]$ (see Appendix~\ref{app:astar:lemma} for a formal proof of potential consistency).
Afterwards, we start the actual \gls*{tfp} search, during which we obtain the key of a label $\tcfunction$ at some vertex $\vertex\in\vertices$ by setting it to
\begin{align*}
 \queuekey(\vertex,\tcfunction):=\min_{\atime\in\posreals}\atime+\convexpotential(\vertex,\maxbattery-\tcfunction(\atime)).
\end{align*}
Recall that $\maxbattery-\tcfunction(\atime)$ yields the \gls*{soc} at $\vertex$ for a driving time of~$\atime$ (\cf line~\ref{line:tfp:initsourcelabel} of Algorithm~\ref{alg:tfp} in Section~\ref{sec:approach:description}).
Computing the key requires a linear scan over the subfunctions defining $\tcfunction$ and the lower bound function $\convexlowerboundfunction_\vertex$ at~$\vertex$.
In each step of the scan, we update the minimum by evaluating the term $\atime+\convexlowerboundfunction_\vertex(\maxbattery-\tcfunction(\atime))$ at up to three values~$\atime\in\posreals$, namely, the boundaries of the intersection of the subdomains of the considered subfunctions of $\tcfunction$ and $\convexlowerboundfunction_\vertex$ and at the unique extreme point of their sum (if it falls within the current intersection of subdomains).

\section{Contraction Hierarchies}\label{sec:ch}

We propose an adaptation of the speedup technique \gls*{ch}~\cite{Gei12b} to the \gls*{evcas} problem. The algorithm distinguishes an offline \emph{preprocessing phase} and an online \emph{query phase}~\cite{Bas14}.
Plain \gls*{ch}~\cite{Gei12b} accelerates single-criterion shortest-path queries \wrt a given cost function~$\costfunction\colon\arcs\to\posreals$ by augmenting the graph with \emph{shortcut arcs}. These shortcuts are computed during a preprocessing routine that iteratively \emph{contracts} all vertices of the input graph in a (heuristic) order.
To this end, a \emph{core} graph $\graph'$ is maintained, initially set to $\graph'=\graph$.
When a vertex $\vertexb\in\vertices$ is contracted, it is removed from $\graph'$ together with all its incident arcs, while shortcut arcs are added between its neighbors to preserve distances, if necessary:
Let $\vertexa\in\vertices$ and $\vertexc\in\vertices$ be neighbors of $\vertexb$ in~$\graph'$. To determine whether a shortcut candidate $(\vertexa,\vertexb)$ with cost $\costfunction^*:=\costfunction(\vertexa,\vertexb)+\costfunction(\vertexb,\vertexc)$ is needed, a \emph{witness search} runs Dijkstra's algorithm to compute the distance $\distance_\costfunction(\vertexa,\vertexc)$ in~$\graph'$ (after removing~$\vertexb$, but before adding the shortcut). If $\distance_\costfunction(\vertexa,\vertexc)\le\costfunction^*$ holds, the shortcut is not required.
Eventually, contracting all vertices results in an ``empty'' graph $\graph'=(\emptyset,\emptyset)$.
As final step of the preprocessing stage, an \emph{augmented} graph is constructed, which consists of all vertices and arcs of the original graph~$\graph$, together with all shortcuts that were inserted into the core graph at some point during preprocessing.

Given a source $\source\in\vertices$ and a target~$\target\in\vertices$, the query algorithm runs a bidirectional variant of Dijkstra's algorithm on the augmented graph: a forward search from $\source$ and a backward search from~$\target$. Each search scans only \emph{upward} arcs in the augmented graph \wrt the contraction order, \ie, only those arcs where the tail vertex was contracted before the head vertex.
One can show that this algorithm computes the distance between $\source$ and $\target$ in the input graph~\cite{Gei12b}.

\paragraph{Adapting \gls*{ch} to the \gls*{evcas} problem.}
As in plain \gls*{ch}~\cite{Gei12b}, we contract vertices iteratively during preprocessing and insert shortcut arcs to maintain distances.
However, we contract only a subset of the vertices, leaving an uncontracted core graph---a common approach in complex settings~\cite{Bau15,Dib15,Har14,Sto12c}.
The query algorithm then scans upward arcs and all arcs that are still contained in the core graph.

Unfortunately, the \gls*{soc} at the target vertex $\target\in\vertices$ is not known at query time. This makes backward search difficult, since it would require us to propagate \gls*{soc} functions.
Their bivariate nature makes explicit construction and comparison of \gls*{soc} functions rather involved.
Therefore, we first run (at query time) a \gls*{bfs} from~$\target$ on the backward graph including shortcuts from preprocessing, scanning only \emph{downward arcs} (no core arcs), \ie, arcs for which the head (\wrt direction in~$\graph$) was contracted before the tail, \emph{marking} each as visited. Afterwards, we execute \gls*{tfp} from the source vertex~$\source\in\vertices$, but let it only scan upward arcs, core arcs, and downward arcs marked by the~\gls*{bfs}.
In other words, we use Algorithm~\ref{alg:tfp}, with the sole modification that only the arcs mentioned above are scanned.

In the remainder of this section, we discuss changes that need to be made to the preprocessing routine to adapt it to our setting.
Most importantly, the \gls*{soc} at a vertex is only known at query time. Hence, any shortcut would have to store a \emph{bivariate} \gls*{soc} function (\cf Section~\ref{sec:model:example}). As argued before, this is rather impractical.
In Section~\ref{sec:ch:simple-soc-functions}, we examine special categories of \gls*{soc} functions that enable both a simple representation and efficient dominance tests. Afterwards, we describe how \gls*{ch} can utilize these simple \gls*{soc} functions if we modify several key steps of the preprocessing routine:
\begin{itemize}
 \item We only allow vertex contraction in cases where all resulting shortcuts fall into the aforementioned categories (Section~\ref{sec:ch:contraction}).
 \item We efficiently decide which shortcuts to keep if there are \emph{multiple} candidates connecting the same pair of vertices (Section~\ref{sec:ch:shortcut-comparison}).
 \item We adapt the witness search to deal with bivariate \gls*{soc} functions (Section~\ref{sec:ch:witness-search}).
\end{itemize}
Omitted proofs can be found in Appendix~\ref{app:ch:lemmas}.
In what follows, we say that an \gls*{soc} function $\socprofile_1$ \emph{dominates} another \gls*{soc} function $\socprofile_2$ if $\socprofile_1(\atime,\soc)\ge\socprofile_2(\atime,\soc)$ holds for all $\atime\in\posreals$ and~$\soc\in[0,\maxbattery]$.

\subsection{Simple Representations of Bivariate Functions}\label{sec:ch:simple-soc-functions}

First, we derive the bivariate \gls*{soc} functions of paths where all consumption values are nonnegative.
Consider the consumption functions $\tcfunction_1,\dots,\tcfunction_{k-1}$ of the arcs of a path $\apath=[\vertex_1,\dots,\vertex_k]$ in~$\graph$. Assume that $\tcfunction_i(\atime)\ge0$ holds for all driving times $\atime\in\posreals$ and~$i\in\{1,\dots,k-1\}$.
In this case, battery constraints can render driving at high speed infeasible. On the other hand, recuperation never occurs on~$\apath$. Therefore, the best speed on an arc depends only on the slope of its consumption function, but not on the position of the arc in the path (in contrast to the situation discussed in Section~\ref{sec:model:example} and sketched in Figure~\ref{fig:tradeoff-simple-model}, where the order of arcs clearly matters). Consequently, it is sufficient to first link the consumption functions and apply battery constraints only \emph{once} afterwards.
Lemma~\ref{lem:tradeoff-nonnegative-shortcut} specifies the resulting bivariate \gls*{soc} function, which is represented by a single univariate (consumption) function.

\begin{lemma}
 \label{lem:tradeoff-nonnegative-shortcut}
 Let $\apath = [\vertex_1, \dots, \vertex_k]$ be a path in~$\graph$ and let $\tcfunction_i$ denote the consumption function of the arc~$(\vertex_i,\vertex_{i+1})$ for~$i\in\{1,\dots,k-1\}$. If all consumption functions are nonnegative, \ie, $\tcfunction_i(\atime)\ge0$ holds for all $\atime\in\posreals$ and~$i\in\{1,\dots,k-1\}$, the \gls*{soc} function of~$\apath$ evaluates to
 \begin{align*}
  \socprofile(\atime, \soc)=
  \begin{cases}
   -\infty     & \mbox{if } \soc < \tcfunction(\atime), \\
   \soc - \tcfunction(\atime) & \mbox{otherwise,}
  \end{cases}
 \end{align*}
 where $\tcfunction$ denotes the function obtained after iteratively linking the functions $\tcfunction_1,\dots,\tcfunction_{k-1}$.
\end{lemma}

Note that the functions $\tcfunction_1,\dots,\tcfunction_{k-1}$ can be linked in arbitrary order, since the link operation is both commutative and associative.
Thus, we can easily construct a shortcut for~$\apath$ by iteratively contracting its internal vertices $\vertex_2,\dots,\vertex_{k-1}$ in any given order, each time linking the consumption functions of both incident (shortcut) arcs to compute a new shortcut.

A symmetric argument holds for paths consisting of only nonpositive consumption functions.
Consequently, we can allow contraction of a vertex if \emph{all} (finite) values of consumption functions assigned to \emph{any} of its incident arcs have the same sign, since the resulting bivariate \gls*{soc} functions can be represented efficiently.
On real-world instances, where the majority of consumption values is positive, this approach already allows contraction of large parts of the graph (more than 50\,\% of the vertices in our tests). Nevertheless, the size of the resulting core graph is still too large to achieve significant speedups.

\paragraph{Discharging Paths.}
We discuss simple representations in more general cases, exploiting that most consumption values are positive in realistic instances.
We say that a path~$\apath$ is \emph{discharging} if the \gls*{soc} on $\apath$ never exceeds the (arbitrary) initial~\gls*{soc}, \ie, there is no prefix of~$\apath$ that has negative minimum consumption for any driving time. Definition~\ref{def:discharging-path} formalizes this notion. 
\begin{definition}[Discharging Path]
\label{def:discharging-path}
 A path $\apath=[\vertex_1,\dots,\vertex_k]$ in the graph is discharging if for arbitrary admissible driving times $\atime_i\in\posreals$ with~$i\in\{1,\dots,k-1\}$, we obtain $\sum_{i=1}^j\consumptionfunction_{(\vertex_i,\vertex_{i+1})}(\atime_i)\ge0$ for every~$j\in\{1,\dots,k-1\}$, \ie, the total energy consumption on the path is nonnegative on each prefix $[\vertex_1,\dots,\vertex_{j+1}]$ of~$\apath$, including $\apath$ itself.
\end{definition}
Note that even though the total consumption on a discharging path must be nonnegative, it may still contain subpaths with negative consumption (as long as they are preceded by subpaths with at least as much positive consumption).
Apparently, it is not necessary for us to explicitly check whether the \gls*{soc} exceeds~$\maxbattery$ on a discharging path, since this will never occur. Hence, making sure that battery constraints are not violated becomes easier.
We now show how the \gls*{soc} function of a discharging path is represented by at most two consumption functions.

Clearly, a path consisting of only arcs with nonnegative consumption values (for arbitrary driving times) is discharging. We already showed how it can be represented by a single consumption function.
As a more intricate example, assume we are given a path~$\apath=\apath_1\circ\apath_2$ consisting of two subpaths~$\apath_1$ and $\apath_2$ that can be represented by two consumption functions~$\tcfunction_1$ and~$\tcfunction_2$. Let $\leftintervalborder_1$,~$\rightintervalborder_1$,~$\leftintervalborder_2$, and~$\rightintervalborder_2$ denote the respective minimum and maximum driving times of $\tcfunction_1$ and~$\tcfunction_2$. Moreover, assume that the consumption $\tcfunction_1(\atime)>0$ is \emph{positive} for all $\atime\in\posreals$, while $\tcfunction_2(\atime)\le0$ is \emph{nonpositive} for all~$\atime\in[\leftintervalborder_2,\infty)$.
Finally, assume that~$|\tcfunction_1(\rightintervalborder_1)|\ge|\tcfunction_2(\rightintervalborder_2)|$, \ie, the cost of $\apath_1$ is higher than the gain of~$\apath_2$ for \emph{any} driving time, so $\apath$ is discharging.
We derive the \gls*{soc} function of $\apath$, represented by a \emph{positive part} $\tcfunction^+$ with $\tcfunction^+(\atime):=\tcfunction_1(\atime-\leftintervalborder_2)$ and a \emph{negative part} $\tcfunction^-(\atime)$ with $\tcfunction^-(\atime):=\tcfunction_2(\atime+\leftintervalborder_2)$.
The original functions are shifted along the x-axis to simplify the analysis (note that the minimum driving time of $\tcfunction^-$ is~$0$).
Given some initial \gls*{soc}~$\soc\in[0,\maxbattery]$, the positive part~$\tcfunction^+$, and the negative part~$\tcfunction^-$, we define the \emph{constrained positive part} $\tcfunction^+_\soc\colon\posreals\to\posreals\cup\{\infty\}$ as
\begin{align*}
 \tcfunction^+_\soc(\atime) :=
\begin{cases}
 \infty                 & \mbox{if } \soc < \tcfunction^+(\atime)\mbox{,} \\
 \tcfunction^+(\atime)  & \mbox{otherwise,}
\end{cases}
\end{align*}
which ensures that battery constraints along $\apath_1$ are not violated if the initial \gls*{soc} is~$\soc$; see Figure~\ref{fig:tradeoff-linked-shortcut}. Then, the \gls*{soc} function $\socprofile$ of the path $\apath$ evaluates to $\socprofile(\atime,\soc)=\soc-\linkop(\tcfunction^+_\soc,\tcfunction^-)(\atime)$ for arbitrary $\atime\in\posreals$ and $\soc\in[0,\maxbattery]$. Given some initial~\gls*{soc}, the function $\socprofile$ first checks battery constraints on the positive part $\tcfunction^+$ and links the resulting function $\tcfunction^+_\soc$ with the negative part~$\tcfunction^-$.
Since the underlying path $\apath$ is discharging, we know that no further constraints need to be checked for~$\tcfunction^-$, so the function computed by $\linkop(\tcfunction^+_\soc,\tcfunction^-)$ yields minimum energy consumption subject to driving time for the initial \gls*{soc}~$\soc$.

\begin{figure}[t]
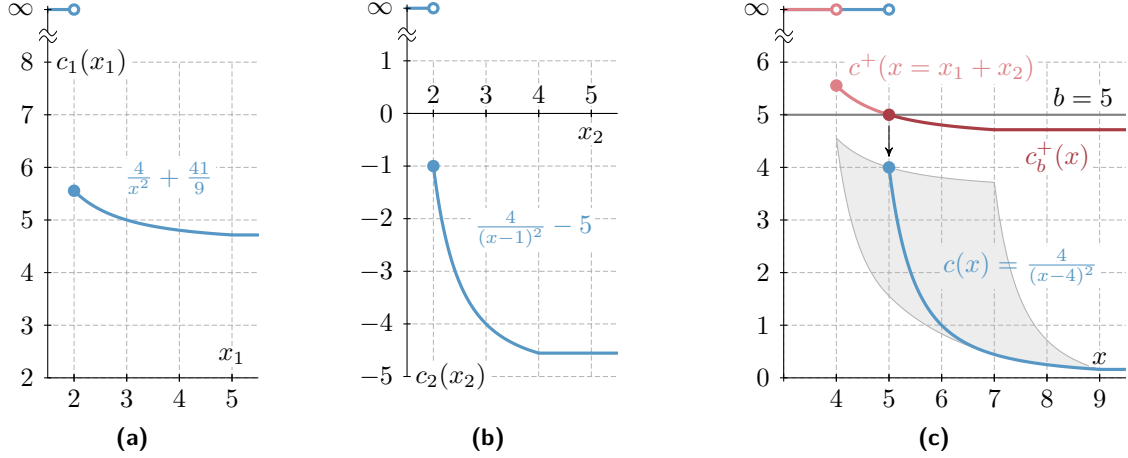

 \centering
 \begin{subfigure}[b]{.285\textwidth}%
 \centering%
 \ConstrainedExampleDrawRealisticSoCFunctionFirstEdge
 \caption{}%
 \label{fig:tradeoff-linked-shortcut:first-arc}%
 \end{subfigure}%
 \begin{subfigure}[b]{.285\textwidth}%
 \centering%
 \ConstrainedExampleDrawRealisticSoCFunctionSecondEdge
 \caption{}%
 \label{fig:tradeoff-linked-shortcut:second-arc}%
 \end{subfigure}%
 \begin{subfigure}[b]{.43\textwidth}%
 \centering%
 \ConstrainedExampleDrawRealisticSoCFunctionLink
 \caption{}%
 \label{fig:tradeoff-linked-shortcut:result}%
 \end{subfigure}%
 \caption{Construction of a consumption function depending on initial~\gls*{soc}. (a) The consumption function~$\tcfunction_1$ of a path $\apath_1$ with positive consumption. (b) The consumption function~$\tcfunction_2$ of a path $\apath_2$ with negative consumption. (c) Deriving the consumption function of the path~$\apath:=\apath_1\circ\apath_2$. Due to battery constraints, the minimum driving time of~$\tcfunction^+$ is $5$ for an initial \gls*{soc}~$\soc=5$. This yields the consumption function $\tcfunction:=\linkop(\tcfunction^+_\soc,\tcfunction^-)$ for the path $\apath$ (the function~$\tcfunction^-$, which corresponds to $\tcfunction_2$ being shifted to the left, is not shown). The shaded area indicates possible images of consumption functions for different values of initial~\gls*{soc}.}
 \label{fig:tradeoff-linked-shortcut}
\end{figure}

\subsection{Vertex Contraction}\label{sec:ch:contraction}

We use our insights about discharging paths to establish a preprocessing routine for~\gls*{ch}.
For the sake of simplicity, we assume that for each arc in the graph, the energy consumption is either nonnegative for all admissible driving times or nonpositive for all admissible driving times. Note that we can always enforce this by splitting an arc $\arc\in\arcs$ with $\tcfunction_\arc(\leftintervalborder_\arc)>0$ and $\tcfunction_\arc(\rightintervalborder_\arc)<0$ into two consecutive arcs corresponding to the positive part and the (shifted) negative part of~$\tcfunction_\arc$, respectively.

\paragraph{Active Vertices.}
Shortcuts that represent either discharging paths or paths consisting solely of arcs whose energy consumption is nonpositive for all admissible driving times are called \emph{discharging} or \emph{nonpositive}, respectively. In both cases, we can use the simple representations we derived above to efficiently store their \gls*{soc} functions.
Consequently, we only allow a vertex $\vertex\in\vertices$ to be contracted if all new shortcuts created as part of its contraction are discharging or nonpositive. We call $\vertex$ \emph{active} in this case.
Note that the number of active vertices grows as contraction proceeds, since this results in longer shortcuts, which are more likely to consist of significant positive parts.
It remains to show how to decide efficiently whether a vertex is active during preprocessing and construct the necessary shortcuts if it is indeed contracted.

Assume that at some point during preprocessing, we want to know if a vertex $\vertexb\in\vertices$ incident to two (original or shortcut) arcs $(\vertexa,\vertexb)$ and~$(\vertexb,\vertexc)$ in the core graph is active.
We have to determine whether a new shortcut $(\vertexa,\vertexc)$ can be constructed from $(\vertexa,\vertexb)$ and~$(\vertexb,\vertexc)$ that is either discharging or nonpositive.
Clearly, a nonpositive shortcut can be constructed if and only if both $(\vertexa,\vertexb)$ and~$(\vertexb,\vertexc)$ are nonpositive.
Otherwise, we want to know whether we can construct a discharging shortcut.
Let $\apath_1$ be the underlying path in the original graph represented by $(\vertexa,\vertexb)$ and let $\apath_2$ be the path represented by $(\vertexb,\vertexc)$.
We have to decide whether $\apath=\apath_1\circ\apath_2$ is a discharging path, \ie, energy consumption is nonnegative on every prefix of $\apath$, regardless of the driving time. Apparently, this can only be the case if $\apath_1$ is a discharging path itself. For $\apath_2$, we distinguish two cases.
First, if $\apath_2$ is discharging as well, to follows immediately that also $\apath$ must be discharging.
Second, if $\apath_2$ is not a discharging path, it must consist solely of arcs for which energy consumptions are \emph{nonpositive} for arbitrary (admissible) driving times, since a shortcut $(\vertexb,\vertexc)$ would not have been created otherwise.
Hence, $\apath_2$ is represented by a single consumption function $\tcfunction_2$ with minimum and maximum driving times $\leftintervalborder_2\in\strictposreals$ and $\rightintervalborder_2\in\strictposreals$, such that $\tcfunction_2(\atime)\le0$ for all $\atime\in[\leftintervalborder_2,\infty)$.
To test whether a discharging shortcut $(\vertexa,\vertexc)$ can be constructed in this case, we consider the positive part $\tcfunction^+_1$ and the negative part $\tcfunction^-_1$ of $\apath_1$ with corresponding maximum driving times $\rightintervalborder^+_1$ and~$\rightintervalborder^-_1$.
The path $\apath$ is discharging if and only if~$\tcfunction^+_1(\rightintervalborder^+_1)+\tcfunction^-_1(\rightintervalborder^-_1)+\tcfunction_2(\rightintervalborder_2)\ge0$, as these driving times minimize consumption on~$\apath$ (or any prefix of $\apath$ that ends at a vertex of~$\apath_2$).

\paragraph{Constructing Shortcuts.}
When we contract an active vertex, we have to insert new shortcuts from pairs of existing incident arcs in the current graph. To this end, we have to compute the corresponding \gls*{soc} function from two given discharging or nonpositive shortcuts.
As before, let $\vertexb\in\vertices$ denote the vertex to be contracted and let $(\vertexa,\vertexb)$ and $(\vertexb,\vertexc)$ denote two incident (shortcut) arcs representing underlying paths $\apath_1$ and $\apath_2$ in the core graph.

If both $(\vertexa,\vertexb)$ and $(\vertexb,\vertexc)$ are nonpositive, we simply have to link their underlying consumption functions.
Otherwise, $(\vertexa,\vertexb)$ is discharging, since $\vertexb$ is an active vertex. Let its \gls*{soc} function be defined by two consumption functions~$\tcfunction^+_1$ and~$\tcfunction^-_1$. We now construct the \gls*{soc} function for the path $\apath:=\apath_1\circ\apath_2$, which must be discharging as well.
If $(\vertexb,\vertexc)$ is nonpositive, let its consumption function be given as~$\tcfunction_2$. We immediately obtain the positive part $\tcfunction^+$ and the negative part $\tcfunction^-$ of the discharging path~$\apath$, with $\tcfunction^+(\atime)=\tcfunction^+_1(\atime-\leftintervalborder_2)$ and $\tcfunction^-(\atime)=\linkop(\tcfunction^-_1,\tcfunction_2)(\atime+\leftintervalborder_2)$ for all~$\atime\in\posreals$.

Finally, assume that $\apath_2$ is a discharging path as well, with the respective consumption functions $\tcfunction^+_2$ and~$\tcfunction^-_2$.
Apparently, if we know the initial \gls*{soc}, we can compute the energy consumption on $\apath$ by computing $\linkop(\linkop(\linkop(\tcfunction_1^+,\tcfunction_1^-),\tcfunction_2^+),\tcfunction_2^-)$ and applying battery constraints \emph{before} each link operation, like in the \gls*{tfp} algorithm (see Section~\ref{sec:approach:description}).
However, we want to represent $\apath$ with only two consumption functions $\tcfunction^+$ and~$\tcfunction^-$.
Recall that the only constraint we have to check for discharging paths is whether the \gls*{soc} drops below~$0$.
Since both $\tcfunction^-_1$ and $\tcfunction^-_2$ are nonpositive for all admissible driving times, the constraint needs only to be checked for $\tcfunction^+_1$ and~$\tcfunction^+_2$ (\ie, before the first and third link operation).
To integrate these checks into a single new positive part~$\tcfunction^+$, we first compute the consumption function $h:=\linkop(\tcfunction^-_1,\tcfunction^+_2)$.
Clearly, the battery can only run empty on $\apath_2$ if this consumption function is \emph{positive} for some admissible driving time (otherwise, we always gain more energy with $\tcfunction_1^-$ than is lost on~$\tcfunction_2^+$).
To distinguish, we split $h$ into a positive part $h^+\colon\posreals\to\posreals\cup\{\infty\}$ with $h^+(\atime):=\max\{h(\atime),0\}$ and a negative part $h^-\colon\posreals\to\negreals\cup\{\infty\}$ with
\begin{align*}
 h^-(\atime) :=
 \begin{cases}
  \infty    & \mbox{if } \atime<\rightintervalborder\text{ and }h(\atime)>0\text{,}\\
  \min\{h(\atime),0\} & \mbox{otherwise,}
 \end{cases}
\end{align*}
where $\rightintervalborder\in\strictposreals$ denotes the maximum driving time of~$h$.
Since $h$ is a valid consumption function, so are both $h^+$ and~$h^-$.
Observe that in case $h$ yields nonnegative (nonpositive) energy consumption for all admissible driving times, the function $h^-$ ($h^+$) always evaluates to $0$ on its subdomain with finite image.
We derive the positive part $\tcfunction^+$ of $\apath$ by setting $\tcfunction^+(\atime):=\linkop(\tcfunction_1^+,h^+)(\atime)$ for all $\atime\in\posreals$ and the negative part $\tcfunction^-$ of $\apath$ by setting~$\tcfunction^-(\atime):=\linkop(h^-,\tcfunction^-_2)(\atime+\leftintervalborder)$ for all $\atime\in\posreals$, where $\leftintervalborder\in\strictposreals$ is the minimum driving time of~$h^-$ (we shift the function to ensure that its minimum driving time is~$0$).
Our way of splitting the function $h$ ensures that battery constraints are only applied to prefixes of $\apath$ with positive energy consumption.
The \gls*{soc} function of $\apath$ is obtained from $\tcfunction^+$ and $\tcfunction^-$ as described in Section~\ref{sec:ch:simple-soc-functions}.

\subsection{Comparing Shortcut Candidates}\label{sec:ch:shortcut-comparison}

In a bicriteria setting, vertex contraction may result in multi-arcs between neighbors of the contracted vertex. In such cases, we only want to keep shortcuts if their \gls*{soc} function is not dominated by \gls*{soc}~functions of parallel arcs.
Hence, after the contraction of a vertex, we want to delete (parts of) \gls*{soc} functions of shortcut candidates that are dominated by existing functions between the same pair of vertices (and vice versa).
To this end, we require efficient dominance checks for \gls*{soc} functions that are either \emph{discharging}, \ie, represent a discharging path, or have nonpositive energy consumption for all admissible driving times.
Assume two given \gls*{soc} functions $\socprofile_1$ and $\socprofile_2$ defined by the respective consumption functions $\tcfunction^+_1$,~$\tcfunction^-_1$,~$\tcfunction^+_2$, and~$\tcfunction^-_2$ (we assume that the positive part evaluates to $0$ for all admissible driving times if consumption is always nonpositive). Further, we assume that our goal is to remove dominated parts of $\socprofile_2$, \ie, we want to identify intervals $\interval\subseteq\posreals$ where $\socprofile_1(\atime,\soc)\ge\socprofile_2(\atime,\soc)$ holds for all driving times $\atime\in\interval$ and all values $\soc\in[0,\maxbattery]$ of initial~\gls*{soc}.
If both \gls*{soc} functions have nonpositive consumption, each is represented by a single consumption function and we can immediately apply the dominance tests described in Section~\ref{sec:operations:dominance}.
For the remaining cases, note that a discharging \gls*{soc} function cannot dominate a function with nonpositive consumption for all possible values of initial \gls*{soc}.
Thus, we consider the case where at least $\socprofile_2$ is discharging. We propose a method that may miss dominated parts of \gls*{soc}~functions, but requires only a linear scan over the consumption functions that define the \gls*{soc} functions $\socprofile_1$ and~$\socprofile_2$. Thereby, correctness is maintained, but unnecessary shortcuts may be inserted.
Let $\varepsilon\ge0$ denote the smallest nonnegative real value such that $\tcfunction^-_1(\atime)\le\tcfunction^-_2(\atime)+\varepsilon$ holds for all $\atime\in\posreals$. This value can be computed in a linear scan over the subfunctions of $\tcfunction^-_1$ and $\tcfunction^-_2$, similar to a dominance test.
Lemma~\ref{lem:tradeoff-profile-dominance} claims that if $\tcfunction^+_1(\atime)+\varepsilon\le\tcfunction^+_2(\atime)$ holds for some~$\atime\in\posreals$, choosing a driving time of $\atime$ for the positive part $\tcfunction^+_2$ always results in a solution that is dominated by~$\socprofile_1$, regardless of the initial~\gls*{soc}.

\begin{lemma}
 \label{lem:tradeoff-profile-dominance}
 Given a nonpositive or discharging \gls*{soc}~function $\socprofile_1$ and a discharging \gls*{soc}~function~$\socprofile_2$, such that their respective consumption functions are~$\tcfunction_1^+$,~$\tcfunction_1^-$,~$\tcfunction_2^+$, and~$\tcfunction_2^-$, let the value $\varepsilon\ge0$ be defined as described above.
 If $\tcfunction^+_1(\atime^+)+\varepsilon\le\tcfunction^+_2(\atime^+)$ holds for some~$\atime^+\in\posreals$, any solution where $\atime^+$ is the (optimal) amount of time spent for $\tcfunction^+_2$ is dominated by $\socprofile_1$, \ie, we obtain either $\tcfunction_2^+(\atime^+)=\infty$ or $\socprofile_1(\atime^++\atime^-,\soc)\ge\socprofile_2(\atime^++\atime^-,\soc)=\soc-(\tcfunction_2^+(\atime^+)+\tcfunction_2^-(\atime^-))$ for all $\atime^-\in\posreals$ and~$\soc\in[0,\maxbattery]$.
\end{lemma}

After creating a new shortcut $(\vertexa,\vertexb)$ with positive part $\tcfunction^+$, we compare it to existing shortcuts between $\vertexa\in\vertices$ and $\vertexb\in\vertices$ as follows. First, we compute the value $\varepsilon$ defined above \wrt each existing shortcut. Then, we determine parts of~$\tcfunction^+$ that are dominated by existing positive parts (after we increase their consumption by~$\varepsilon$) and set $\tcfunction^+(\atime)=\infty$ for such values $\atime\in\posreals$. We do this in a coordinated linear scan over $\tcfunction^+$ and the positive parts of consumption functions of all existing shortcuts, as in our basic dominance tests (see Section~\ref{sec:operations:dominance}).
If~$\tcfunction^+\equiv\infty$ holds afterwards, we remove the shortcut. Analogously, we identify parts of \gls*{soc} functions of existing shortcuts that are dominated by the \gls*{soc} function of the new shortcut candidate.

\begin{figure}[t]
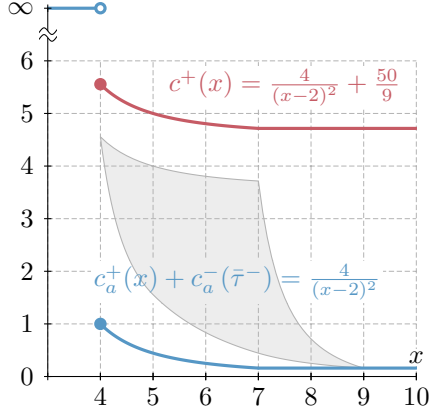

 \centering%
 \ConstrainedExampleDrawRealisticSoCFunctionUpperBound
 \caption{Bound functions computed by the witness search for the \gls*{soc} function depicted in Figure~\ref{fig:tradeoff-linked-shortcut:result}. The shaded area indicates possible values taken by this \gls*{soc} function depending on driving time and~\gls*{soc}. It is upper bounded by the function $\tcfunction^+$ (red) computed by the witness search. The lower bound (blue) is used to identify dominated parts of the \gls*{soc} function if it represents a shortcut candidate.}%
\label{fig:tradeoff-function-upper-bound}%
\end{figure}

\subsection{Witness Search}\label{sec:ch:witness-search}

Consider a discharging shortcut candidate $(\vertexa,\vertexb)$ from $\vertexa\in\vertices$ to $\vertexb\in\vertices$ that is neither dominated by any existing parallel shortcut (from $\vertexa$ to~$\vertexb$) nor dominates an existing parallel shortcut itself.
Before adding $(\vertexa,\vertexb)$ to the graph, we run a witness search (as in plain~\gls*{ch}) to test if the shortcut is necessary to maintain distances in the current graph~$\graph'$ (for some values of driving times and \gls*{soc}).
An exact approach would compute bivariate \gls*{soc} functions representing $\vertexa$--$\vertexb$~paths in $\graph'$ to identify dominated parts of the shortcut candidate.
Like before, this is difficult and potentially expensive due to the lack of efficient operations for construction and comparison of such \gls*{soc} functions. Instead, we only compute univariate \emph{upper bounds} on energy consumption during witness search.

As a key idea, the search drops negative parts from \gls*{soc} functions entirely, so labels in the search consist of only the positive parts. Clearly, these labels are upper bounds on energy consumption.
Our witness search runs the basic \gls*{tfp} algorithm from $\vertexa$ on $\graph'$, but ignores battery constraints and links labels only with positive parts of arcs (we assume that the positive part evaluates to $0$ for all admissible driving times if consumption on the corresponding arc is always nonpositive).
As a result, each label stores an upper bound on overall (finite) consumption that is independent of initial~\gls*{soc} and represented by a single consumption function~$\tcfunction^+$; see Figure~\ref{fig:tradeoff-function-upper-bound}.
Moreover, since the majority of arcs have positive consumption in realistic instances, the upper bounds have relatively small error in most cases.

\paragraph{Witness Comparisons.}
Whenever a label $\tcfunction^+$ at the head vertex $\vertexb$ of the shortcut candidate $\arc:=(\vertexa,\vertexb)$ is extracted during the witness search, we compare the witness $\tcfunction^+$ to the \gls*{soc} function $\socprofile_{\arc}$ of the shortcut candidate as follows. Let $\tcfunctionshortcut^+$ and $\tcfunctionshortcut^-$ denote the consumption functions defining $\socprofile_{\arc}$. Moreover, let $\tcfunctionshortcut^-(\rightbordershortcut^-)$ be the \emph{minimum consumption} of~$\tcfunctionshortcut^-$.
We know that $\socprofile_{\arc}$ cannot be a (unique) optimal choice for a driving time $\atime\in\posreals$ if $\tcfunction^+(\atime)\le\tcfunctionshortcut^+(\atime)+\tcfunctionshortcut^-(\rightbordershortcut^-)$, \ie, the upper bound $\tcfunction^+$ provides better or equal consumption than a \emph{lower bound} on the consumption of~$\socprofile_{\arc}$; see Figure~\ref{fig:tradeoff-function-upper-bound}.
We proceed along the lines of the dominance tests described in Section~\ref{sec:operations:dominance} to identify such values $\atime\in\posreals$ and remove them from the subdomain of admissible driving times of~$\tcfunctionshortcut^+$. If $\tcfunctionshortcut^+\equiv\infty$ holds afterwards, the shortcut is not required.
Otherwise, the witness search stops once the minimum driving time of a scanned label exceeds the maximum driving time of the shortcut candidate.
During the search, we can discard labels if their minimum consumption exceeds the maximum consumption $\tcfunctionshortcut^+(\leftbordershortcut^+)+\tcfunctionshortcut^-(\rightbordershortcut^-)$ of the lower bound of the shortcut candidate.

\paragraph{Simplified Upper Bounds.}
Going even further, we replace upper bounds defined by \emph{multiple} subfunctions that are computed during the witness search with \emph{single} (less accurate) tradeoff functions. Then, the witness search propagates labels of constant size, enabling faster operations and better data locality, since we can preallocate memory for the labels (making cache misses become less likely). Moreover, we use lightweight dominance tests that employ pairwise label comparison (see Section~\ref{sec:operations:dominance}), which are faster in this simplified setting. In Appendix~\ref{app:ch:upper-bounds}, we describe in detail how good bounds of constant descriptive complexity are computed.

\section{CHAsp}\label{sec:chasp}

It is straightforward to combine our variants of A*~search and \gls*{ch} to obtain our fastest algorithm, which we refer to as \emph{\acrshort*{champ}\glsunset{champ} (\acrlong*{champ})}.
Queries are answered by first computing a potential function on the \gls*{ch} search graph and afterwards running \gls*{tfp} (or its heuristic variant) on it, utilizing the corresponding key function.
Note that \gls*{champ} does not alter the output of the underlying basic algorithm, so correctness is maintained when using plain~\gls*{tfp}. Running time remains exponential in the worst case, but we observe significant speedups in practice (see Section~\ref{sec:experiments}). Moreover, \gls*{champ} can be combined with our polynomial-time heuristic in just the same way. In summary, our algorithm consists of several alternative building blocks and both its preprocessing and query stage comprise different steps, which we recap in the overview given below. For implementation details, see also Appendix~\ref{app:implementation}.

\paragraph{Preprocessing.}
In the preprocessing stage, we are given the input graph $\graph=(\vertices,\arcs)$ and the consumption functions of all arcs, on which we perform the following steps.
\begin{enumerate}
 \item Determine a heuristic vertex order by assigning a priority to each vertex (see Geisberger et~al.~\cite{Gei12b} and Appendix~\ref{app:implementation}) and identify active vertices (see Section~\ref{sec:ch:contraction}).
 \item Initialize the core graph by setting it to $\graph'=\graph$. Contract active vertices in the given order, until no active vertices are left or the average degree of active vertices reaches a certain threshold. When contracting a vertex~$\vertex\in\vertices$,
 \begin{itemize}
  \item construct shortcut candidates for each pair of neighbors of $\vertex$ and compare each with any existing arcs in the core graph between the same pair of neighbors (see Section~\ref{sec:ch:shortcut-comparison});
  \item run witness searches between pairs of neighbors with any remaining shortcut candidates to identify unnecessary candidates (see Section~\ref{sec:ch:witness-search});
  \item remove $\vertex$ from $\graph'$, add remaining shortcut candidates to~$\graph'$, update the priority and activity status of all neighbors of~$\vertex$.
 \end{itemize}
 \item Output the contraction order and the graph $\graph$ enriched with all shortcuts that at some point were added to or are still present in $\graph'$ (arcs still present are called \emph{core arcs}).
\end{enumerate}

\paragraph{Queries.}
After completion of the preprocessing stage, queries consist of a source vertex~$\source\in\vertices$, a target vertex~$\target\in\vertices$, and an initial~\gls*{soc}~$\soc_\source\in[0,\maxbattery]$. A query is answered by executing the steps listed below.
\begin{enumerate}
 \item Run a \gls*{bfs} from both $\source$ and~$\target$, scanning only upward and downward arcs (\wrt the contraction order), respectively, and no core arcs. Build a search graph $\graph_{\source,\target}$ consisting of all arcs visited by either \gls*{bfs} and all core arcs.
 \item Compute either the vertex potential $\drivingtimepotential$ or~$\convexpotential$ by running the corresponding backward search from $\target$ on~$\graph_{\source,\target}$ (see Section~\ref{sec:astar}).
 \item Run \gls*{tfp} (see Section~\ref{sec:approach:description}) or its heuristic variant (see Section~\ref{sec:approach:heuristic}) from $\source$ on~$\graph_{\source,\target}$, but use the respective key function from Section~\ref{sec:astar}.
\end{enumerate}

\section{Experiments}\label{sec:experiments}

We implemented all approaches in~C++, using~g++ 7.3.1 (flag -O3) as compiler. Experiments were conducted on a single core of a 4-core Intel Xeon E5-1630v3 clocked at 3.7\,GHz, with 128\,GiB of DDR4-2133 RAM, 10\,MiB of L3 cache, and 256\,KiB of L2 cache.

\begin{table}[t]
\caption{Considered input instances. We report the number of vertices, arcs, arcs with negative consumption~(as a fraction of total number of arcs), and nonconstant arcs~(as a fraction of total number of arcs).}
\label{tbl:graphs-adaptive-speeds}
\centering
\small
\setlength{\tabcolsep}{1.5ex}
\begin{tabular}{lrrr@{}rr@{}rr}
\toprule
Instance & \#\,Vertices & \#\,Arcs & \multicolumn{2}{r}{\#\,Arcs with $\tcfunction<0$} & \multicolumn{2}{r}{\#\,Nonconstant Arcs} \\ 
\midrule
\instanceGerAux   &  4\,692\,091 & 10\,805\,429 &    846\,505 &  (7.83\,\%) &  2\,730\,390 & (25.27\,\%) \\
\instanceGerNoAux &  4\,692\,091 & 10\,805\,429 & 1\,317\,969 & (12.20\,\%) &  2\,730\,390 & (25.27\,\%) \\
\instanceEurAux   & 22\,198\,628 & 51\,088\,095 & 4\,887\,770 &  (9.57\,\%) & 19\,212\,330 & (37.61\,\%) \\
\instanceEurNoAux & 22\,198\,628 & 51\,088\,095 & 6\,607\,508 & (12.93\,\%) & 19\,212\,330 & (37.61\,\%) \\
\bottomrule
\end{tabular}
\end{table}

We consider two graphs representing the road networks of Europe with 22\,198\,628 vertices and 51\,088\,095 arcs, and Germany with 4\,692\,091 vertices and 10\,805\,429 arcs, kindly provided by~PTV~AG.\footnote{\url{ptvgroup.com}}
Applying elevation data from the Shuttle Radar Topography Mission,~v4.1,\footnote{\url{srtm.csi.cgiar.org}} we derived realistic energy consumption functions from two detailed micro-scale emission models~\cite{Hau09}. The first is a vehicle model that is calibrated to a Peugeot~iOn.
The second is an artificial model~\cite{Tie12} that, in contrast to the first, takes power demand of auxiliary consumers (\eg, air conditioning) into account.
These data sources are proprietary, but enable evaluation on detailed and realistic input data.
We extracted functions following Equation~\eqref{eq:tradeoff-physicalmodel} from given samples of speed and energy consumption via regression.
Combining reasonable minimum speeds for different road types (\eg, 80\,km/h on motorways and 30\,km/h in residential areas) with the posted speed limits (if higher), we get intervals of admissible speeds per road segment. As a result, 25\,\% and 38\,\% of the arcs allow adapting the speed for the network of Germany and Europe, respectively. The respective average speed interval lengths among those arcs are 17\,km/h and 27\,km/h for the two instances. The notable differences between the two instances can be explained by a larger percentage of arcs corresponding to local or residential roads in the Germany instance, which only allow for little speed adaptation, but also by differences in road categories and speeds of the different countries in the input data in general.
We denote our instances by Germany-Aux~(\instanceGerAux), Germany-Peugeot~(\instanceGerNoAux), Europe-Aux~(\instanceEurAux), and Europe~Peugeot~(\instanceEurNoAux). They have negative consumption (at least for maximum driving times) on 7.8\,\%~(\instanceGerAux) to~12.9\,\%~(\instanceEurNoAux) of their arcs. Table~\ref{tbl:graphs-adaptive-speeds} gives an overview of the different input instances.

Battery capacities are measured in~kWh, where 1\,kWh corresponds to a range of roughly 5--10\,km (depending on speed and terrain).
Most experiments use realistic ranges of~16--64\,kWh, but we resort to queries of shorter range when necessary to evaluate slower baseline approaches.

Unless mentioned otherwise, our study evaluates random \emph{in-range} queries, \ie, we pick a source vertex~$\source\in\vertices$ uniformly at random. Among all vertices in range from~$\source$ with an initial \gls*{soc}~$\soc_\source=\maxbattery$, we pick the target $\target\in\vertices$ uniformly at random.
Since unreachable targets can be detected by backward search phases of A*~search or by any algorithm for computing energy-optimal routes~\cite{Bau13a,Eis11,Sac11}, this results in more difficult and interesting queries for us.

\begin{table}[t]
\caption{Benefits of our approach (\instanceGerNoAux, 2\,kWh). For each basic technique (\gls*{bsp}, \gls*{tfp} with pairwise dominance checks, \gls*{tfp}-d with improved dominance checks), we report average figures of 100 queries on the thousandth part of the number of settled labels~(\#\,Labels), relaxed arcs~(\#\,Rel.), and pairwise label comparisons during dominance checks~(\#\,Cmp). We further show the average time for a single link operation or dominance check (in nanoseconds), as well as the median~(Med.), average~(Avg.) with standard deviation~(Std.), and maximum (Max.) running time of each algorithm (in milliseconds).}
\label{tbl:adaptive-speeds-model}
\small
\centering
\setlength{\tabcolsep}{0.98ex}
\begin{tabular}{lrrrrrcrrrr}
\toprule
&\multicolumn{5}{c}{Query Statistics} && \multicolumn{4}{c}{Query Times}\\
\cmidrule(lr){2-6}\cmidrule(lr){8-11}
              &   \#\,Labels &          \#\,Rel. &          \#\,Cmp. &      Lnk. &      Dom. &&    Med. &     Avg. &     Std. &        Max. \\
Algo.         &          [k] &               [k] &               [k] &      [ns] &      [ns] &&    [ms] &     [ms] &     [ms] &        [ms] \\
\midrule
\gls*{bsp}    &      35\,586 &           66\,217 &      76\,457\,603 &        25 &    2\,418 && 46\,223 & 174\,810 & 347\,083 & 2\,165\,515 \\
\gls*{tfp}    &           58 &                81 &            1\,035 &       221 &    1\,161 &&      18 &      130 &      314 &      2\,314 \\
\gls*{tfp}-d &           33 &                46 &               302 &       181 &       725 &&      11 &       53 &      116 &         645 \\
\bottomrule
\end{tabular}
\end{table}

\subsection{Model Validation}\label{sec:experiments:model}

We argued that an approach based fully on consumption \emph{functions} unlocks both better tractability and improved solution quality compared to discrete speeds.
Therefore, we also consider instances with multi-arcs to model speed adaptation---as was best practice in previous approaches~\cite{Bau14a,Goo14,Har14}. We generate multi-arcs in a rather conservative way, by sampling consumption functions at velocity steps of~20\,km/h. Optimal paths (\wrt to the simple model) are then computed by the \gls*{bsp} algorithm described in Section~\ref{sec:model:evcsp}.
Indeed, we observe a significant speedup by simply switching to our more realistic model, as Table~\ref{tbl:adaptive-speeds-model} shows.
Even for a small range~(2\,kWh), we see that \gls*{tfp} is three orders of magnitudes faster than \gls*{bsp}, since it evaluates speed-consumption tradeoffs more fine-granularly and with less overhead: Instead of huge discrete Pareto sets, each route is represented by a single continuous function represented by few parameters. As a result, the number of labels settled by \gls*{tfp} (each representing a distinct \emph{route}) is up to three orders of magnitudes less compared to the labels scanned by~\gls*{bsp} (each representing a distinct combination of speed samples on some route). Pairwise label comparisons even drop by more than five orders of magnitude.

\begin{table}[t]
\caption{Running \gls*{bsp} with varying speed steps (\instanceGerNoAux, 2\,kWh). Using A*-$\drivingtimepotential$ to render query times feasible, we ran 100 queries of \gls*{bsp} for different step sizes, resulting in the indicated number of arcs. We also consider instances with two multi-arcs per tradeoff function (step size~$\infty$), the original instance without speed adaptation~(cnst.), and continuous functions (step size~$0$).
We report average label scans and comparisons, running time and standard deviation, as well as average and maximum quality loss of solutions (due to discretized speeds).}
\label{tbl:speed-steps}
\small
\centering
\setlength{\tabcolsep}{0.98ex}
\begin{tabular}{rcrrrrrcrr}
\toprule
Step&&&\multicolumn{4}{c}{Query} &&\multicolumn{2}{c}{Quality Loss}\\
\cmidrule(lr){4-7}\cmidrule(lr){9-10}
[km/h]     &&     \#\,Arcs & \#\,Labels  &          \#\,Cmp. & Avg.\,[ms] &       Std. &&       Avg. &    Max. \\
\midrule
 1         && 60\,511\,294 & 1\,457\,402 & 10\,448\,724\,562 &  93\,407.2 & 382\,901.2 &&   1.000001 & 1.000051\\
 2         && 36\,861\,919 &    899\,737 &  5\,907\,007\,043 &  24\,573.5 & 108\,487.2 &&   1.000020 & 1.000215\\
 5         && 22\,672\,294 &    398\,377 &  1\,622\,249\,986 &   2\,848.7 &  11\,097.0 &&   1.000136 & 1.000872\\
 10        && 17\,942\,419 &    202\,525 &     402\,370\,894 &      595.5 &   2\,419.5 &&   1.000377 & 1.002332\\
 20        && 14\,778\,112 &    117\,205 &     135\,477\,211 &      203.2 &      688.4 &&   1.001481 & 1.022579\\
 $\infty$  && 13\,535\,819 &     35\,006 &      11\,070\,106 &       26.7 &       64.0 &&   1.006657 & 1.091806\\
 cnst.     && 10\,805\,429 &         181 &               204 &        4.4 &        3.8 &&   1.022592 & 1.225135\\
 0         && 10\,805\,429 &         342 &               922 &        6.2 &        4.2 &&   1.000000 & 1.000000\\
\bottomrule
\end{tabular}
\end{table}

Regarding query statistics, note that the number of arc relaxations coincides with the number of link operations. \gls*{bsp} spends almost an order of magnitude less time per operation (essentially, only two pairs of real values have to be added up).
In case of \gls*{bsp}, the number of arc relaxations further coincides with the number of dominance tests, whereas it is an upper bound for \gls*{tfp} (we postpone dominance tests by managing sets of unsettled labels; recall Section~\ref{sec:approach:description}).
Although label comparisons are much simpler for \gls*{bsp} (comparing two pairs of real-valued numbers), the average running time of a single dominance test is higher compared to \gls*{tfp}. This is due to a larger number of labels per set, resulting in significantly more pairwise comparisons for~\gls*{bsp}.
For \gls*{tfp}-d (using improved dominance checks), each dominance test compares a new label to a set of $k$ settled labels in a single coordinated sweep, which we count as $k$ comparisons in Table~\ref{tbl:adaptive-speeds-model}. Even though the coordinated sweep comes at some overhead, the average time per dominance check actually decreases, as more dominated (sub)functions are detected and the number of required comparisons is reduced by more than a factor of~3. Linking becomes somewhat faster for the same reason (labels store fewer subfunctions).
Overall, dominance tests take more than 90\,\% of the running time of \gls*{bsp}. For \gls*{tfp}, linking and dominance tests require roughly 15\,\% and 70\,\% of the overall running time, respectively.
Although atomic operations (linking and comparing labels) are more expensive for~\gls*{tfp}, the drastic reduction in the number of vertex scans explains the speedup. This is interesting, as sampling was expressly considered to manage tractability~\cite{Bau14a,Goo14,Har14,Mer15}.

We observe that the average query time of all approaches is mostly determined by outliers, which is not surprising given that their asymptotic complexity is exponential. As a result, most queries are actually answered much faster than the average, with the median running time being about a factor of~5 lower on average.

\paragraph{Speed Resolution.}
In Table~\ref{tbl:speed-steps}, we evaluate the effect of different step sizes on the performance of \gls*{bsp} and solution quality (see Appendix~\ref{app:experiments} for corresponding plots of query times and result quality).
Since including more multi-arcs can slow down the algorithm significantly, we use A*~search with the potential function~$\drivingtimepotential$; see Baum et~al.~\cite{Bau14a} and Section~\ref{sec:astar:single-criterion}.
Each row of the table shows the result of running \gls*{bsp} (with potentials) on the graph resulting from sampling each tradeoff function at the indicated rate (1--20\,km/h). We also provide results when each function corresponds to exactly two samples at its minimum and maximum driving time, respectively (step size~$\infty$), on the unaltered original graph~(cnst.), and when running \gls*{tfp} on continuous functions (step size~0).
In instances with multi-arcs, we always include the arc that minimizes energy consumption to maintain reachability of the target (by contrast, only 77\,\% of the targets could be reached on the unaltered original instance).

Note that using A*~search comes at the cost of side effects: In cases where the target is (almost) reachable on the unconstrained fastest route, potentials ensure that \gls*{bsp} and \gls*{tfp} (almost) only scan vertices on the shortest route. Therefore, the median query time on \emph{every} instance considered in Table~\ref{tbl:speed-steps} is below 10\,ms (not reported in the table). Furthermore, the average time of faster approaches with small search spaces, such as~\gls*{tfp}, is dominated by the backward search (more than 90\,\% of its total running time).

On average, \gls*{tfp} is significantly faster than \gls*{bsp}, even if only two multi-arcs are added per tradeoff function, and several orders of magnitudes faster for sensible speed steps. At the same time, it finds paths that are up to 9.1\,\% quicker (within battery constraints). While quality loss quickly decreases with smaller speed steps (sampling 10\,km/h steps suffices to keep it below~1\,\%), its running time makes \gls*{bsp} impractical: The number of arcs in the graph less than doubles for step sizes up to~10\,km/h, but the number of labels can grow exponentially even along a single route (see Section~\ref{sec:model:evcsp}).

\begin{table}[t]
\caption{Varying speed interval lengths~(\instanceGerNoAux, 16\,kWh). Each row shows a modified instance, shrinking speed intervals (0\,\%, 25\,\%, or 50\,\%) or increasing the minimum adaptive speed (60\,km/h, 80\,km/h, or 100\,km/h). For 100 in-range queries of \gls*{tfp} with the potential function~$\drivingtimepotential$, we report the average number of settled labels and label comparisons, average running time, and standard deviation. Further, we show the number of optimal results found (using our original instance as ground truth), as well as average and maximum quality loss.}
\label{tbl:speed-corridors}
\small
\centering
\setlength{\tabcolsep}{0.98ex}
\begin{tabular}{rcrrrrcrrr}
\toprule
&&\multicolumn{4}{c}{Query} &&\multicolumn{3}{c}{Quality Loss}\\
\cmidrule(lr){3-6}\cmidrule(lr){8-10}
        Corridor &&  \#\,Labels &      \#\,Cmp. & Avg.\,[ms] &       Std. &&    Opt. &     Avg.&   Max.\\
\midrule
           0\,\% && 1\,503\,555 & 269\,194\,933 &     8\,226 &    29\,298 &&  50\,\% & 1.06733 & 1.48192\\
          25\,\% &&    656\,629 &  51\,987\,431 &     6\,737 &    26\,019 &&  56\,\% & 1.03245 & 1.22557\\
          50\,\% &&    541\,874 &  37\,728\,672 &     5\,342 &    22\,417 &&  62\,\% & 1.01693 & 1.14067\\
         100\,\% &&    364\,662 &  17\,808\,829 &     2\,680 &    12\,748 && 100\,\% & 1.00000 & 1.00000\\ 
\addlinespace
 $\ge$\,60\,km/h &&    553\,521 &  52\,183\,681 &     5\,351 &    25\,812 &&  97\,\% & 1.00002 & 1.00050\\
 $\ge$\,80\,km/h &&    875\,462 & 140\,608\,647 &    10\,772 &    44\,813 &&  79\,\% & 1.01180 & 1.10670\\
$\ge$\,100\,km/h && 1\,323\,155 & 242\,448\,187 &     9\,689 &    37\,055 &&  60\,\% & 1.05274 & 1.48192\\
\bottomrule
\end{tabular}
\end{table}

\paragraph{Varying Speed Interval Lengths.}
Minimum speeds of tradeoff functions in our instances are chosen such that driving slower would become unprofitable or cause the \gls*{ev} to impede traffic flow. However, as adapting the driving speed pays off the most on faster roads, one could consider increasing minimum adaptive speeds.
Below, we evaluate such modifications, by either increasing the minimum speed on every nonconstant arc to shrink its speed interval to a certain fraction of its original length (0\,\%, 25\,\%, or 50\,\%) or by only allowing adaptive speeds up to a certain lower threshold (60\,km/h, 80\,km/h, or 100\,km/h).

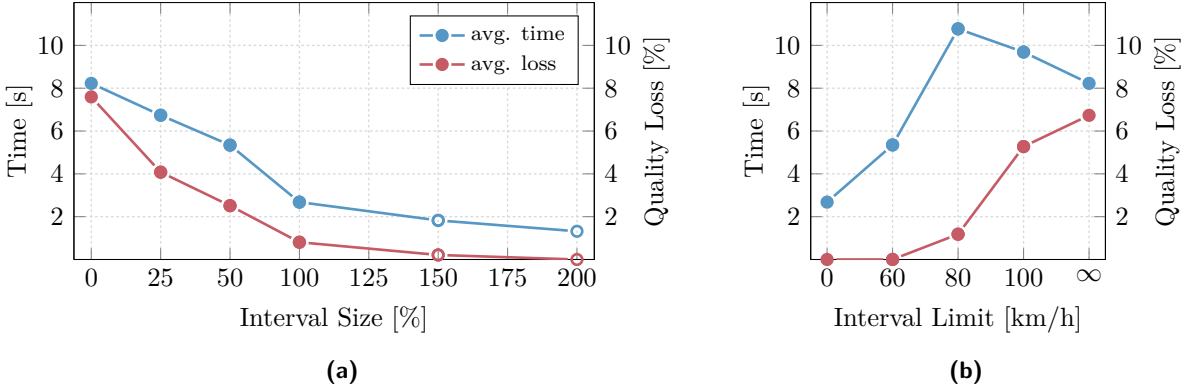
\begin{figure}[t]
 \centering
 \begin{subfigure}[b]{.57\textwidth}%
 \centering%
 \tikzstyle{markSign} = [mark=*]
\tikzstyle{shortenLines} = [shorten <= 3.5pt,shorten >= 3.5pt]

\begin{tikzpicture}[figure]
\pgfplotsset{
    grid style = {dash pattern = on 1pt off 1pt, black15,line width = 0.5pt  }
 }
\pgfplotsset{ every non boxed x axis/.append style={x axis line style=-},
     every non boxed y axis/.append style={y axis line style=-},
     legend image post style={line width=1.5pt}}

\colorlet{plotColor1}{thesisblue} 
\colorlet{plotColor3}{thesisred} 
\colorlet{plotColor4}{thesisyellow} 
     
\begin{axis}[
   height=5cm,
   width=8.5cm,
   xmin=0.75,
   xmax=8.25,
   ymin=0,
   ymax=12,
   ytick pos=left,
   ytick={2,4,6,8,10},
   xlabel={Interval Size [\%]},
   ylabel={Time [s]},
   ylabel style={yshift=-0.2cm},
   xtick={1, 2, 3, 4, 5, 6, 7, 8},
   xticklabels={0, 25, 50, 100, 125, 150, 175, 200},
   grid=major,
   legend entries={avg. time, avg. loss},
   legend cell align=left,
   legend style={at={(0.97,0.95)},
   anchor=north east,
   font=\scriptsize}
]

\addlegendimage{legend line with exact mark,plotColor1}
\addlegendimage{legend line with exact mark,plotColor3}
\addlegendimage{legend line with exact mark,plotColor4}

\addplot [color=plotColor1,markSign,shortenLines,line width=1pt] table {
   1 8.225752407
   2 6.737202741
};
\addplot [color=plotColor1,markSign,shortenLines,line width=1pt] table {
   2 6.737202741
   3 5.341606879
};
\addplot [color=plotColor1,markSign,shortenLines,line width=1pt] table {
   3 5.341606879
   4 2.679503057
};
\addplot [color=plotColor1,heuristicMarkSign,shortenLines,line width=1pt] table {
   4 2.679503057
   6 1.825607091
};
\addplot [color=plotColor1,heuristicMarkSign,shortenLines,line width=1pt] table {
   6 1.825607091
   8 1.318467629
};

\end{axis}

\begin{axis}[
   height=5cm,
   width=8.5cm,
   ytick pos=right,
   ymin=0,
   ymax=12,
   ytick={2,4,6,8,10},
   xmin=0.75,
   xmax=8.25,
   axis x line=none,
   ylabel={Quality Loss [\%]},
   ylabel style={yshift=0.1cm}]

\addplot [color=plotColor3,markSign,shortenLines,line width=1pt] table {
   1 7.5940
   2 4.0781
};
\addplot [color=plotColor3,markSign,shortenLines,line width=1pt] table {
   2 4.0781
   3 2.5131
};
\addplot [color=plotColor3,markSign,shortenLines,line width=1pt] table {
   3 2.5131
   4 0.8065
};
\addplot [color=plotColor3,heuristicMarkSign,shortenLines,line width=1pt] table {
   4 0.8065
   6 0.2115
};
\addplot [color=plotColor3,heuristicMarkSign,shortenLines,line width=1pt] table {
   6 0.2115
   8 0.0000
};

\end{axis}
\end{tikzpicture}
 \caption{}%
 \label{fig:corridor-plot:rel}%
 \end{subfigure}%
 \begin{subfigure}[b]{.43\textwidth}%
 \centering%
 \tikzstyle{markSign} = [mark=*]
\tikzstyle{shortenLines} = [shorten <= 3.5pt,shorten >= 3.5pt]

\begin{tikzpicture}[figure]
\pgfplotsset{
    grid style = {dash pattern = on 1pt off 1pt, black15,line width = 0.5pt  }
 }
\pgfplotsset{ every non boxed x axis/.append style={x axis line style=-},
     every non boxed y axis/.append style={y axis line style=-},
     legend image post style={line width=1.5pt}}

\colorlet{plotColor1}{thesisblue} 
\colorlet{plotColor3}{thesisred} 
\colorlet{plotColor4}{thesisyellow} 
     
\begin{axis}[
   height=5cm,
   width=5.5cm,
   xmin=0.75,
   xmax=5.25,
   ymin=0,
   ymax=12,
   ytick pos=left,
   ytick={2,4,6,8,10},
   xlabel={Interval Limit [km/h]},
   ylabel={Time [s]},
   ylabel style={yshift=-0.2cm},
   xtick={1, 2, 3, 4, 5},
   xticklabels={0, 60, 80, 100, $\infty$},
   grid=major,
]

\addlegendimage{legend line with exact mark,plotColor1}
\addlegendimage{legend line with exact mark,plotColor3}
\addlegendimage{legend line with exact mark,plotColor4}

\addplot [color=plotColor1,markSign,shortenLines,line width=1pt] table {
   1 2.679503057
   2 5.350632822
};
\addplot [color=plotColor1,markSign,shortenLines,line width=1pt] table {
   2 5.350632822
   3 10.772301841
};
\addplot [color=plotColor1,markSign,shortenLines,line width=1pt] table {
   3 10.772301841
   4 9.689220622
};
\addplot [color=plotColor1,markSign,shortenLines,line width=1pt] table {
   4 9.689220622
   5 8.225752407
};

\end{axis}

\begin{axis}[
   height=5cm,
   width=5.5cm,
   ytick pos=right,
   ymin=0,
   ymax=12,
   ytick={2,4,6,8,10},
   xmin=0.75,
   xmax=5.25,
   axis x line=none,
   ylabel={Quality Loss [\%]},
   ylabel style={yshift=0.1cm}]

\addplot [color=plotColor3,markSign,shortenLines,line width=1pt] table {
   1 0.0000
   2 0.0016
};
\addplot [color=plotColor3,markSign,shortenLines,line width=1pt] table {
   2 0.0016
   3 1.1800
};
\addplot [color=plotColor3,markSign,shortenLines,line width=1pt] table {
   3 1.1800
   4 5.2744
};
\addplot [color=plotColor3,markSign,shortenLines,line width=1pt] table {
   4 5.2744
   5 6.7332
};

\end{axis}
\end{tikzpicture}
 \caption{}%
 \label{fig:corridor-plot:abs}%
 \end{subfigure}%
 \caption{Average relative increase of query time and quality loss for the settings considered in Table~\ref{tbl:speed-corridors}. (a)~Results when shrinking all speed intervals to a given fraction of the original size (100\,\%) or increasing the speed interval size (150\,\%, 200\,\%). (b)~Results for intervals with a given minimum speed ($0$ corresponds to the original instance, $\infty$ to an instance in which all arcs have constant driving time).}
 \label{fig:corridor-plot}
\end{figure}

Table~\ref{tbl:speed-corridors} shows the results of 100 queries for each setting, chosen such that the target is in range for even the most restrictive setting (0\,\%, \ie, constant speed on all arcs). As in the previous experiment, we use A*~search with the potential function $\drivingtimepotential$ to enable a longer range of 16\,kWh for this experiment. As before, this affects the distribution of query times (the median is 300--400\,ms in all cases), but we still observe differences in average performance.
Interestingly, shrinking the interval actually \emph{increases} query times in most cases: Shorter intervals imply that more functions are required to cover the whole Pareto front, so more labels have to be propagated. At the same time, solution quality consistently deteriorates, as there are fewer speed options.
For the first category of instances (shrinking intervals of admissible driving times), the number of tradeoff functions remains unchanged (about 25\,\% of all arcs; see Table~\ref{tbl:graphs-adaptive-speeds}). The second category, however, adjusting the minimum adaptive speed means that arcs for which the maximum speed does not exceed this minimum become constant. Since the majority of arcs in a road graph represent slower roads, the number of (nonconstant) tradeoff functions decreases significantly to 9\,\% (3.5\,\%, 0.5\,\%) of all arcs for a minimum adaptive speed of~60\,km/h (80\,km/h, 100\,km/h).
This also explains why query times decrease again if we only allow speed adaption beyond~100\,km/h: Most labels are tuples of constant values and dominance tests become simpler. A similar effect explains why compared to intervals shrunk to 25\,\% of their original size, the number of comparisons when using only constant labels (first row of the table) increases by a factor of more than~5, but the query time just is slightly higher.

Shrinking the admissible intervals of all tradeoff functions consistently deteriorates solution quality. Compared to our basic setting~(100\,\%), results are up to 50\,\% worse. Interestingly, we observe almost no quality loss if we employ a lower adaptive speed limit of~60\,km/h. Query times still increase, indicating that adaptive speeds on slower roads are not needed for result quality, but to remove unimportant tradeoffs from label sets.
Figure~\ref{fig:corridor-plot} plots running times and quality for the different settings shown in Table~\ref{tbl:speed-corridors}. For comparison, we also include extended speed intervals (compared to our basic instance). Larger intervals actually improve performance further, but running times and quality improvements converge (and we argued that speed recommendations below reasonable thresholds are not meaningful in practice).

\subsection{\gls*{champ}}\label{sec:experiments:chasp}

In what follows, we focus on the evaluation of different variants of our fastest approach,~\gls*{champ}. A comparison of \gls*{champ} and basic algorithms is presented in Section~\ref{sec:experiments:scalability}.

\begin{table}[t]
\caption{Impact of core size on performance (\acrshort*{champ},~\instanceGerNoAux,~16\,kWh). Vertex contraction stopped once the average degree of active vertices in the core reached a given threshold (\O\,Deg.). We report the resulting core size (\#\,Vertices) and remaining active (\ie, contractable) vertices (\#\,Active), preprocessing time, and average query times of 1\,000 queries using \gls*{champ} with potential functions $\drivingtimepotential$ and~$\convexpotential$, respectively.}
\label{tbl:optimal-core-cize-champ}
\setlength{\tabcolsep}{1.5ex}
\centering
\small
\begin{tabular}{rrrrrr}
 \toprule
 \multicolumn{3}{c}{Core size} & Prepr. & \multicolumn{2}{c}{Query [ms]} \\
 \cmidrule(lr){1-3}\cmidrule(lr){4-4}\cmidrule(lr){5-6}
 \O\,Deg. &                    \#\,Vertices &            \#\,Active & [h:m:s] & \gls*{champ}-$\drivingtimepotential$ & \gls*{champ}-$\convexpotential$ \\ \midrule
  0       &                             --- &                   --- &     --- &  3\,253.6 &  4\,828.8 \\
  8       & 720\,514            (15.36\,\%) &   326\,527 (6.96\,\%) &    4:59 &     728.6 &     798.3 \\
  16      & 400\,174 \hphantom{1}(8.53\,\%) &   114\,510 (2.44\,\%) &   12:55 &     485.9 &     485.0 \\
  24      & 333\,819 \hphantom{1}(7.11\,\%) &    75\,108 (1.60\,\%) &   21:22 &     445.7 &     442.6 \\
  32      & 305\,301 \hphantom{1}(6.51\,\%) &    58\,796 (1.25\,\%) &   30:28 &     441.1 &     434.0 \\
  48      & 279\,943 \hphantom{1}(5.97\,\%) &    44\,301 (0.94\,\%) &   47:40 &     461.7 &     451.0 \\
  64      & 268\,436 \hphantom{1}(5.72\,\%) &    37\,455 (0.80\,\%) & 1:05:42 &     492.2 &     473.1 \\
 128      & 251\,410 \hphantom{1}(5.36\,\%) &    26\,920 (0.57\,\%) & 2:22:36 &     632.3 &     586.1 \\
 256      & 242\,817 \hphantom{1}(5.18\,\%) &    21\,040 (0.45\,\%) & 5:39:15 &     901.8 &     802.3 \\
\bottomrule
\end{tabular}
\end{table}

Table~\ref{tbl:optimal-core-cize-champ} shows details on \gls*{ch} preprocessing effort and its impact on query performance subject to different core sizes on~\instanceGerNoAux, assuming a battery capacity of~16\,kWh. Vertex contraction was stopped as soon as the average degree of active (\ie, contractable) vertices in the core reached the indicated threshold~(\O\,Deg.).
We report resulting core sizes and preprocessing time, as well as query times of our fastest exact algorithms.
We see that contraction becomes much slower beyond a core degree of~32, due tp the small number of remaining active vertices: Only 58\,796 out of 305\,301 remaining vertices in the core are active when the average degree reaches~32.
This also explains why the speedup compared to the baseline (a threshold of $0$ for the average degree of active core vertices yields plain \gls*{tfp} combined with A*~search) is much smaller than in single-criterion~\gls*{ch} with scalar arc costs~\cite{Gei12b}. Similar deteriorations in speedup were observed in other complex problems, such as time-dependent profile computation~\cite{Bat13}, time-dependent aircraft flight planning~\cite{Bla16}, and multicriteria routing~\cite{Fun13}. Nevertheless, \gls*{ch} still yields an improvement by up to an order of magnitude in our case.

In all experiments below, we picked an average core degree of 32 as stopping criterion of \gls*{ch} preprocessing, as it yields best results not only in terms of average query times, but also all other considered time measures (median, standard deviation, maximum; not reported in the table).
The resulting core size depends on different parameters, including vehicle range and error thresholds (of heuristic variants).
Relative core sizes thus vary between 2.8\,\% for~\instanceGerAux and 8.5\,\% for~\instanceEurNoAux, which is explained by the difference in the amount of arcs with negative consumption. Recall that this has a significant impact on the number of active vertices and the contraction order (see Section~\ref{sec:ch}).

\begin{table}[t]
\caption{Preprocessing and query performance of~\gls*{champ}. For each instance, we provide \gls*{ch} preprocessing times and query times of exact \gls*{champ}, using the potential functions $\drivingtimepotential$ and~$\convexpotential$, respectively.
Query figures are average values of 1\,000 in-range queries (999 queries on \instanceEurNoAux for a range of 64\,kWh and the potential~$\drivingtimepotential$, indicated by an asterisk). We show timings of the A* backward search, the number of settled labels (\#\,Labels) and label comparisons (\#\,Cmp.) during the forward search, and total query timings.}
\label{tbl:performance}
\setlength{\tabcolsep}{0.98ex}
\small
\centering
\begin{tabular}{p{10pt}rlrrrrrrrr}
\toprule
& & & Prepr.\, & \multicolumn{2}{c}{Backward} & \multicolumn{2}{c}{Forward} & \multicolumn{3}{c}{Query\,[ms]}  \\
\cmidrule(lr){4-4}\cmidrule(lr){5-6}\cmidrule(lr){7-8}\cmidrule(lr){9-11}
Ins.& $\maxbattery$ & Pot. &        [h:m:s] & Avg.\,[ms] &     Std. &  \#\,Labels &         \#\,Cmp. &       Med. &       Avg. &          Std. \\
 \midrule
 \multirow{4}{*}{\vspace{-28pt}\hspace{10pt}\vertical{\instanceGerAux}\hspace{-5pt}}
  & 16 &  $\drivingtimepotential$ &   29:21 &        2.8 &      1.6 &         153 &           3\,789 &        3.1 &        4.0 &           5.0 \\
  & 16 &       $\convexpotential$ &   29:21 &       15.8 &     11.7 &          62 &              448 &       13.8 &       16.5 &          12.0 \\
  & 64 &  $\drivingtimepotential$ &   29:22 &       31.2 &     14.1 &      9\,424 &      1\,004\,292 &       38.9 &      172.2 &         441.7 \\
  & 64 &       $\convexpotential$ &   29:22 &      685.3 &    276.3 &         249 &           4\,040 &      660.3 &      686.9 &         276.4 \\
 \addlinespace\multirow{4}{*}{\vspace{-28pt}\hspace{10pt}\vertical{\instanceEurAux}\hspace{-5pt}}
  & 16 &  $\drivingtimepotential$ & 3:01:41 &        2.8 &      2.8 &         125 &           2\,175 &        2.7 &        3.8 &           4.6 \\
  & 16 &       $\convexpotential$ & 3:01:41 &       14.6 &     26.2 &          74 &           1\,006 &       11.3 &       15.3 &          26.3 \\
  & 64 &  $\drivingtimepotential$ & 3:01:41 &       33.2 &     22.1 &     17\,426 &      3\,914\,318 &       39.8 &      455.2 &      3\,935.7 \\
  & 64 &       $\convexpotential$ & 3:01:41 &      605.7 &    369.5 &         608 &          31\,143 &      595.0 &      615.8 &         384.4 \\
 \addlinespace\multirow{4}{*}{\vspace{-28pt}\hspace{10pt}\vertical{\instanceGerNoAux}\hspace{-5pt}}
  & 16 &  $\drivingtimepotential$ &   30:12 &       27.4 &     15.0 &     32\,774 &      6\,352\,489 &       49.1 &      442.9 &      1\,832.5 \\
  & 16 &       $\convexpotential$ &   30:12 &      341.5 &    193.7 &      6\,008 &         491\,173 &      380.5 &      420.5 &         309.7 \\
  & 64 &  $\drivingtimepotential$ &   30:14 &      189.6 &     24.6 & 4\,047\,690 & 6\,438\,508\,333 &   1\,898.1 & 335\,399.8 & 1\,304\,029.5 \\
  & 64 &       $\convexpotential$ &   30:14 &  11\,885.0 & 2\,464.7 &     25\,998 &      3\,585\,423 &  11\,519.2 &  12\,428.3 &      3\,015.3 \\
 \addlinespace\multirow{4}{*}{\vspace{-28pt}\hspace{10pt}\vertical{\instanceEurNoAux}\hspace{-5pt}}
  & 16 &  $\drivingtimepotential$ & 2:59:32 &       23.7 &     19.2 &     23\,305 &      5\,024\,403 &       42.2 &      337.5 &      1\,450.8 \\
  & 16 &       $\convexpotential$ & 2:59:32 &      235.7 &    200.5 &      6\,630 &         800\,431 &      233.0 &      329.8 &         516.3 \\
  & 64 & $\drivingtimepotential$* & 2:59:48 &      267.1 &    122.0 & 3\,100\,076 & 7\,618\,974\,437 &   5\,180.6 & 343\,288.6 & 1\,418\,006.5 \\
  & 64 &       $\convexpotential$ & 2:59:48 &  15\,113.1 & 9\,198.6 &     47\,257 &     16\,209\,990 &  15\,349.6 &  16\,684.1 &     11\,279.6 \\
\bottomrule
\end{tabular}
\end{table}

\paragraph{Exact \gls*{champ}.}
Table~\ref{tbl:performance} shows the performance of \gls*{champ} on all four considered instances. We report vertex scans and dominance tests of the forward search only, excluding the A* backward search (query times include both the forward and the backward search, though). Timings of only the A* backward search of each approach are shown for comparison.

Preprocessing time mostly depends on the graph size, taking about 6 times longer on Europe, but not on the vehicle model: The inclusion of auxiliary consumers reduces the number of negative arcs and hence, contraction can proceed faster, but it also takes longer to reach the threshold degree in the core graph. As a result, the core graph size (not reported in the table) is reduced by almost a factor of 2 after about the same computation time.

Regarding queries, we observe significant differences in performance, depending on the input instance and vehicle range.
For the artificial model, we achieve quite practical times in the order of milliseconds for a realistic range of 16\,kWh and less than a second on average for the long range of~64\,kWh.
For the Peugeot model, on the other hand, average running times exceed tens of seconds and even minutes, with outliers well beyond an hour of running time, depending on the vehicle range.
This gap in running time between both consumption models is explained by the difference in the number of arcs with negative cost. One could argue that the instances \instanceGerNoAux and \instanceEurNoAux are rather excessive in this regard, by not accounting for any auxiliary consumers at all.
As a result, these instances are significantly more difficult to solve for our algorithms.

The search space is consistently smaller when using the potential function~$\convexpotential$, but the backward search is more expensive. In fact, it becomes the major bottleneck for a battery capacity of~16\,kWh on the easier instances. Consequently, average query times are slowed down by about a factor of~4 in this case.
For the harder instances (no auxiliary consumers), the potential function $\convexpotential$ provides better results due to its better scalability and more stable query times. This becomes evident in particular for the longer range of 64\,kWh: On \instanceEurNoAux, one query did not terminate within a preset time limit of ten hours when using the potential~$\drivingtimepotential$. Hence, all reported statistics only take the remaining 999 queries into account (indicated by the asterisk in Table~\ref{tbl:performance}).

Average running times are affected by outliers in most cases, with median query times being up to two orders of magnitude lower.
Comparing the two road networks Germany and Europe, we achieve slightly better running times on Europe for the medium range (explained by differences in the length of admissible speed intervals), but Germany is easier to solve for long ranges (searches are more likely to reach the border of the graph).
In summary, we solve \gls*{evcas} \emph{optimally} in less than a second on average for typical ranges, even on hard instances. For long ranges, our algorithm computes the optimal solution within seconds on the most difficult instance when using the potential function~$\convexpotential$, despite its exponential worst-case running time.

\begin{table}[t]
\caption{Performance of the heuristic variant of \gls*{champ}-$\drivingtimepotential$, for different choices of the parameter~$\varepsilon$ on \instanceGerNoAux and~\instanceEurNoAux. We show figures on query performance for the same 1\,000 random queries as in Table~\ref{tbl:performance}. Additionally, we report the percentage of feasible (F.) and optimal (O.) results, as well as the average and maximum relative error of all queries in which a feasible solution was found.}
\label{tbl:heuristics}
\small
\centering
\setlength{\tabcolsep}{0.98ex}
\begin{tabular}{p{10pt}rcrrrrrrrrr}
\toprule
& & & Pr. &\multicolumn{4}{c}{Query} & \multicolumn{4}{c}{Result Quality}   \\
\cmidrule(lr){4-4}\cmidrule(lr){5-8}\cmidrule(lr){9-12}
Ins.& $\maxbattery$ & $\varepsilon$ & [h:m:s] & \#\,Labels & \#\,Cmp. & Avg.\,[ms] & Std. & F.\,[\%]  & O.\,[\%] & Avg.& Max.\\
 \midrule
 \multirow{8}{*}{\vspace{-35pt}\hspace{10pt}\vertical{\instanceGerNoAux}\hspace{-5pt}}
 & 16 & 0.00 &   30:12 &     32\,774 &      6\,352\,489 &      442.9 &      1\,832.5 & 100.0 & 100.0 & 1.0000 & 1.0000 \\
 & 16 & 0.01 &   29:21 &     19\,922 &      1\,949\,458 &      218.8 &         678.4 & 100.0 &  89.4 & 1.0001 & 1.0047 \\
 & 16 & 0.10 &   25:04 &      6\,891 &         208\,058 &       71.4 &         142.3 &  98.9 &  62.8 & 1.0013 & 1.0502 \\
 & 16 & 1.00 &   17:28 &      1\,743 &          11\,150 &       29.4 &          29.9 &  95.1 &  47.6 & 1.0144 & 1.2294 \\
 \addlinespace
 & 64 & 0.00 &   30:14 & 4\,047\,690 & 6\,438\,508\,333 & 335\,399.8 & 1\,304\,029.5 & 100.0 & 100.0 & 1.0000 & 1.0000 \\
 & 64 & 0.01 &   27:23 &    602\,227 &     70\,352\,417 &  10\,119.0 &     18\,049.5 &  99.9 &  58.7 & 1.0005 & 1.0180 \\
 & 64 & 0.10 &   19:56 &    163\,434 &      4\,798\,549 &   2\,029.8 &      3\,131.9 &  98.2 &  44.2 & 1.0058 & 1.1217 \\
 & 64 & 1.00 &   15:28 &     50\,488 &         459\,291 &      645.0 &         818.9 &  93.6 &  37.1 & 1.0345 & 1.1785 \\
 \addlinespace
 \multirow{8}{*}{\vspace{-35pt}\hspace{10pt}\vertical{\instanceEurNoAux}\hspace{-5pt}}
 & 16 & 0.00 & 2:59:32 &     23\,305 &      5\,024\,403 &      337.5 &      1\,450.8 & 100.0 & 100.0 & 1.0000 & 1.0000 \\
 & 16 & 0.01 & 2:57:02 &     12\,804 &      1\,132\,685 &      143.2 &         461.4 & 100.0 &  82.8 & 1.0001 & 1.0145 \\
 & 16 & 0.10 & 2:41:39 &      5\,046 &         126\,663 &       57.6 &         102.0 &  99.5 &  57.5 & 1.0020 & 1.0418 \\
 & 16 & 1.00 & 2:11:29 &      1\,429 &           7\,642 &       26.9 &          33.8 &  92.7 &  45.8 & 1.0203 & 1.3960 \\
 \addlinespace
 & 64 & 0.00 & 2:59:48 & 3\,100\,076 & 7\,618\,974\,437 & 343\,288.6 & 1\,418\,006.5 & 100.0 & 100.0 & 1.0000 & 1.0000 \\
 & 64 & 0.01 & 2:50:24 &    388\,548 &     41\,533\,224 &   7\,391.8 &     13\,019.9 &  99.8 &  48.6 & 1.0010 & 1.0249 \\
 & 64 & 0.10 & 2:22:52 &    110\,189 &      3\,184\,124 &   1\,734.7 &      2\,684.2 &  95.2 &  36.9 & 1.0108 & 1.1771 \\
 & 64 & 1.00 & 2:01:57 &     24\,652 &         179\,941 &      447.3 &         567.6 &  85.8 &  34.5 & 1.0394 & 1.3255 \\
\bottomrule
\end{tabular}
\end{table}

\paragraph{Heuristic \gls*{champ}.}
Table~\ref{tbl:heuristics} evaluates our heuristic approach for different choices of the parameter~$\varepsilon$ (in \% of total battery capacity; see Section~\ref{sec:approach:heuristic}). During preprocessing of~\gls*{champ}, new shortcuts are included only if they \emph{significantly} improve over the existing ones. Thus, preprocessing becomes faster and core sizes (not reported in the table) decrease down to around 70--80\,\% of their original size.

We also observe a drop in query times:
For a range of~16\,kWh, we achieve a considerable speedup of an order of magnitude. Regarding quality of results, the choice of~$\varepsilon$ clearly matters. For~$\varepsilon=0.01$, the decrease in quality is negligible, but speedup (about a factor of~2) is moderate. For $\varepsilon=0.1$, on the other hand, the optimal solution is still found in many cases. The average error is roughly~0.2\,\%, while the overall maximum is 5\,\%, which is acceptable in practice.
Finally, for~$\varepsilon=1.0$, both the average and maximum error increase significantly. Given that speedup is limited compared to the case $\varepsilon=0.1$, we conclude that the latter provides the best tradeoff in terms of quality and query performance.
Providing high-quality solutions, it enables query times of well below 100\,ms, which is fast enough even for interactive applications.
Moreover, note that in cases where no path is found (about 1\,\% of all queries for~$\varepsilon=0.1$), a simple fallback could return the energy-optimal path, which can be computed quickly~\cite{Bau13a,Eis11,Sac11}.

Considering the longer range of~64\,kWh, the choice of $\varepsilon$ as a percentage of the battery capacity affects (relative) preprocessing and query times, which decrease to a smaller fraction of the exact variant (dominance checks become coarser, allowing more shortcuts and labels to be pruned). We achieve a query speedup by more than two orders of magnitude and observe milder outliers. For~$\varepsilon\ge0.1$, query times of a couple of seconds and less are quite practical. However, quality suffers from the coarser dominance relaxation and we observe average errors of around 1\,\%, with outliers exceeding 10\,\% in quality loss. Depending on the application, a value $\varepsilon\in[0.01,0.1]$ provides the most reasonable tradeoff between running time and quality.

For the easier instances \instanceGerAux and~\instanceEurAux (not reported in the table), we generally observe smaller errors, but also a smaller speedup. This is due to the fact that the A* backward search quickly becomes the bottleneck when combined with a heuristic variant of~\gls*{tfp}. Similarly, using the more sophisticated potential function $\convexpotential$ does not pay off (even for the harder instances).
Finally, recall that our actual implementation does not utilize exact dominance checks for the sake of avoiding expensive and numerically instable intersection tests. Although this means that the bound established in Lemma~\ref{lem:poly-label-size-heuristic} does not hold in theory, it was not exceeded once in our experiments.

\begin{figure}[t]
 \centering
 \input{fig-trsc/evtc-queries-16-ranks}
 \caption{Running times of \gls*{champ} subject to Dijkstra rank. We show results for both potential functions $\drivingtimepotential$ and~$\convexpotential$. For each rank, we consider 1\,000 random queries on \instanceEurNoAux, assuming a battery capacity of~16\,kWh.}
 \label{fig:rankplot}
\end{figure}

\subsection{Evaluating Scalability}\label{sec:experiments:scalability}

To assess the scalability of \gls*{champ}, we consider two different input parameters that have a big impact on the size of the search space: the distance between the source and the target of a query, and the range of the vehicle.
The first experiment presented below evaluates our fastest exact algorithms following the methodology of \emph{Dijkstra ranks}, defined for a query as the number of vertex scans when running Dijkstra's algorithm with costs representing unconstrained driving time~\cite{Bas14,San05}.
Thus, higher ranks correspond to harder queries (of longer distance).
In the second experiment, we compare the different approaches discussed in this work and examine how their query times depend on the vehicle's range.

\paragraph{Query Times Subject to Dijkstra Rank.}
Figure~\ref{fig:rankplot} shows results for our fastest exact approaches on~\instanceEurNoAux, assuming a battery capacity of~16\,kWh. We ran 1\,000 random queries per Dijkstra rank (note that these are not necessarily in-range queries). Queries of higher rank correspond to a longer distance between source and target.
It turns out that median running times of both \gls*{champ}-$\drivingtimepotential$ and \gls*{champ}-$\convexpotential$ are quite robust towards varying Dijkstra ranks. We obtain the most expensive queries at ranks from $2^{17}$ to~$2^{19}$.
(Note that random in-range queries are likely to be among these most difficult ranks.)
For higher ranks, the target is often unreachable. In most cases, this is detected by the backward searches for potential computation, as lower bounds on consumption exceed the battery capacity (see Section~\ref{sec:astar}). We could achieve further speedup for high ranks by running any technique that quickly computes energy-optimal routes to detect unreachable targets~\cite{Bau13a,Eis11,Sac11}.
For lower ranks (\ie, more local queries), the target is often reachable on an unconstrained shortest path, so goal direction of the potential functions works very well. In such cases, the backward phase of A*~search becomes the major bottleneck of the query, hence the lightweight potential $\drivingtimepotential$ yields better query times.
Note that query times vary even for these very low or high ranks, as they largely depend on the portion of the core graph that is inspected by the backward search.

Although the median running time of \gls*{champ}-$\convexpotential$ is consistently higher than the median of~\gls*{champ}-$\drivingtimepotential$, the former is also more robust in that it produces fewer outliers. Its more sophisticated potential function pays off especially for harder queries. Nevertheless, as worst-case running time is exponential, we observe a few outliers that exceed the median by several orders of magnitude.

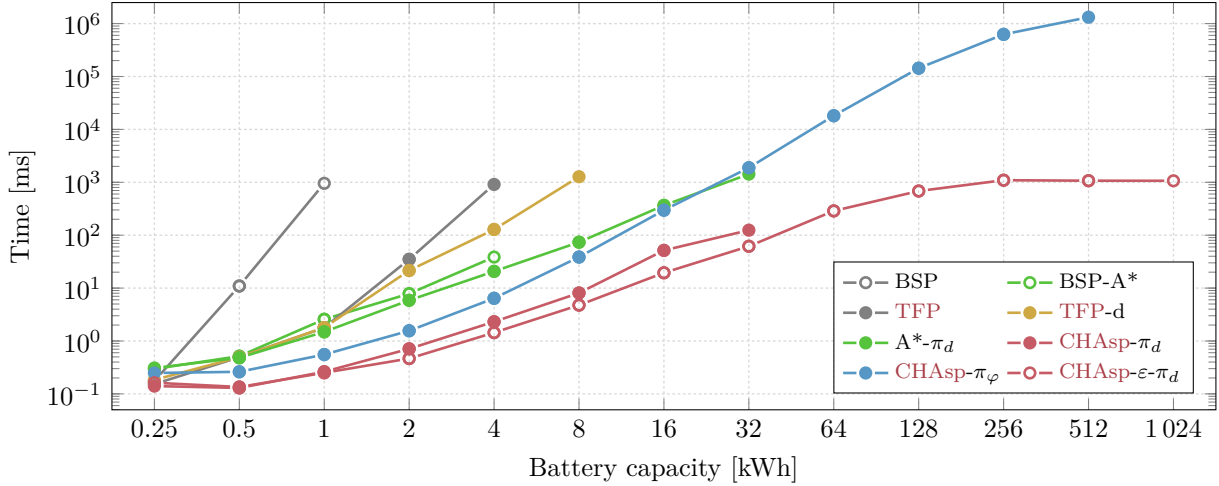
\begin{figure}[t]
  \centering
  \begin{tikzpicture}[figure]
\pgfplotsset{
   grid style = {dash pattern = on 1pt off 1pt, black15,line width = 0.5pt  }
}

\colorlet{plotColor1}{black50} 
\colorlet{plotColor2}{thesisgreen} 
\colorlet{plotColor3}{black50} 
\colorlet{plotColor4}{thesisyellow} 
\colorlet{plotColor5}{thesisgreen} 
\colorlet{plotColor7}{thesisblue} 
\colorlet{plotColor8}{thesisred} 
\colorlet{plotColor9}{thesisred} 

\begin{axis}[
   height=7.0cm,
   width=0.98\textwidth,
   xmin=0.5,
   xmax=13.5,
   ymin=0.05,
   ymax=2500000,
   ymode=log,
   xlabel={Battery capacity [kWh]},
   ylabel={Time [ms]},
   y label style={at={(axis description cs:-0.06,0.5)}},
   /pgf/number format/.cd,
   1000 sep={\,},
   xtick={0, 1, 2, 3, 4, 5, 6, 7, 8, 9, 10, 11, 12, 13},
   xticklabel=\pgfmathparse{2^(\tick-3)}${\pgfmathprintnumber{\pgfmathresult}}$,
   grid=major,
   legend entries={BSP, BSP-A*, \gls{tfp}, \gls{tfp}-d, A*-$\drivingtimepotential$, \gls{champ}-$\drivingtimepotential$, \gls{champ}-$\convexpotential$, \gls{champ}-$\varepsilon$-$\drivingtimepotential$},
   legend cell align=left,
   legend columns=2,
   legend style={at={(0.98,0.04)},
   anchor=south east,
   font=\scriptsize}
]

 \addlegendimage{legend line with heuristic mark,plotColor1}
 \addlegendimage{legend line with heuristic mark,plotColor2}
 \addlegendimage{legend line with exact mark,plotColor3}
 \addlegendimage{legend line with exact mark,plotColor4}
 \addlegendimage{legend line with exact mark,plotColor5}
 \addlegendimage{legend line with exact mark,plotColor8}
 \addlegendimage{legend line with exact mark,plotColor7}
 \addlegendimage{legend line with heuristic mark,plotColor9}

\addplot[color=plotColor1,heuristicMarkSign,shortenLines,line width=1pt] table {
    1 0.175537
    2 10.948500
};
\addplot[color=plotColor1,heuristicMarkSign,shortenLines,line width=1pt] table {
    2 10.948500
    3 955.845500
};

\addplot[color=plotColor2,heuristicMarkSign,shortenLines,line width=1pt] table {
    1 0.296997
    2 0.516113
};
\addplot[color=plotColor2,heuristicMarkSign,shortenLines,line width=1pt] table {
    2 0.516113
    3 2.570560
};
\addplot[color=plotColor2,heuristicMarkSign,shortenLines,line width=1pt] table {
    3 2.570560
    4 7.845455
};
\addplot[color=plotColor2,heuristicMarkSign,shortenLines,line width=1pt] table {
    4 7.845455
    5 38.601400
};

\addplot[color=plotColor3,exactMarkSign,shortenLines,line width=1pt] table {
    1 0.153198
    2 0.492920
};
\addplot[color=plotColor3,exactMarkSign,shortenLines,line width=1pt] table {
    2 0.492920
    3 1.788085
};
\addplot[color=plotColor3,exactMarkSign,shortenLines,line width=1pt] table {
    3 1.788085
    4 34.981950
};
\addplot[color=plotColor3,exactMarkSign,shortenLines,line width=1pt] table {
    4 34.981950
    5 910.824500
};

\addplot[color=plotColor4,exactMarkSign,shortenLines,line width=1pt] table {
    1 0.183960
    2 0.494019
};
\addplot[color=plotColor4,exactMarkSign,shortenLines,line width=1pt] table {
    2 0.494019
    3 1.791990
};
\addplot[color=plotColor4,exactMarkSign,shortenLines,line width=1pt] table {
    3 1.791990
    4 21.505500
};
\addplot[color=plotColor4,exactMarkSign,shortenLines,line width=1pt] table {
    4 21.505500
    5 128.214000
};
\addplot[color=plotColor4,exactMarkSign,shortenLines,line width=1pt] table {
    5 128.214000
    6 1273.605000
};

\addplot[color=plotColor5,exactMarkSign,shortenLines,line width=1pt] table {
    1 0.308472
    2 0.481445
};
\addplot[color=plotColor5,exactMarkSign,shortenLines,line width=1pt] table {
    2 0.481445
    3 1.488400
};
\addplot[color=plotColor5,exactMarkSign,shortenLines,line width=1pt] table {
    3 1.488400
    4 5.874390
};
\addplot[color=plotColor5,exactMarkSign,shortenLines,line width=1pt] table {
    4 5.874390
    5 20.697000
};
\addplot[color=plotColor5,exactMarkSign,shortenLines,line width=1pt] table {
    5 20.697000
    6 73.774050
};
\addplot[color=plotColor5,exactMarkSign,shortenLines,line width=1pt] table {
    6 73.774050
    7 365.816500
};
\addplot[color=plotColor5,exactMarkSign,shortenLines,line width=1pt] table {
    7 365.816500
    8 1430.080000
};

\addplot[color=plotColor7,exactMarkSign,shortenLines,line width=1pt] table {
    1 0.248414
    2 0.261597
};
\addplot[color=plotColor7,exactMarkSign,shortenLines,line width=1pt] table {
    2 0.261597
    3 0.554077
};
\addplot[color=plotColor7,exactMarkSign,shortenLines,line width=1pt] table {
    3 0.554077
    4 1.559935
};
\addplot[color=plotColor7,exactMarkSign,shortenLines,line width=1pt] table {
    4 1.559935
    5 6.407105
};
\addplot[color=plotColor7,exactMarkSign,shortenLines,line width=1pt] table {
    5 6.407105
    6 38.606100
};
\addplot[color=plotColor7,exactMarkSign,shortenLines,line width=1pt] table {
    6 38.606100
    7 296.771000
};
\addplot[color=plotColor7,exactMarkSign,shortenLines,line width=1pt] table {
    7 296.771000
    8 1884.400000
};
\addplot[color=plotColor7,exactMarkSign,shortenLines,line width=1pt] table {
    8 1884.400000
    9 18135.100000
};
\addplot[color=plotColor7,exactMarkSign,shortenLines,line width=1pt] table {
    9 18135.100000
    10 143551.500000
};
\addplot[color=plotColor7,exactMarkSign,shortenLines,line width=1pt] table {
    10 143551.500000
    11 625446.000000
};
\addplot[color=plotColor7,exactMarkSign,shortenLines,line width=1pt] table {
    11 625446.000000
    12 1320655.000000
};

\addplot[color=plotColor8,exactMarkSign,shortenLines,line width=1pt] table {
    1 0.140503
    2 0.129517
};
\addplot[color=plotColor8,exactMarkSign,shortenLines,line width=1pt] table {
    2 0.129517
    3 0.262451
};
\addplot[color=plotColor8,exactMarkSign,shortenLines,line width=1pt] table {
    3 0.262451
    4 0.710938
};
\addplot[color=plotColor8,exactMarkSign,shortenLines,line width=1pt] table {
    4 0.710938
    5 2.311035
};
\addplot[color=plotColor8,exactMarkSign,shortenLines,line width=1pt] table {
    5 2.311035
    6 8.082520
};
\addplot[color=plotColor8,exactMarkSign,shortenLines,line width=1pt] table {
    6 8.082520
    7 51.606450
};
\addplot[color=plotColor8,exactMarkSign,shortenLines,line width=1pt] table {
    7 51.606450
    8 124.436000
};

\addplot[color=plotColor9,heuristicMarkSign,shortenLines,line width=1pt] table {
    1 0.162475
    2 0.135010
};
\addplot[color=plotColor9,heuristicMarkSign,shortenLines,line width=1pt] table {
    2 0.135010
    3 0.250489
};
\addplot[color=plotColor9,heuristicMarkSign,shortenLines,line width=1pt] table {
    3 0.250489
    4 0.466064
};
\addplot[color=plotColor9,heuristicMarkSign,shortenLines,line width=1pt] table {
    4 0.466064
    5 1.439700
};
\addplot[color=plotColor9,heuristicMarkSign,shortenLines,line width=1pt] table {
    5 1.439700
    6 4.763920
};
\addplot[color=plotColor9,heuristicMarkSign,shortenLines,line width=1pt] table {
    6 4.763920
    7 19.469450
};
\addplot[color=plotColor9,heuristicMarkSign,shortenLines,line width=1pt] table {
    7 19.469450
    8 61.712450
};
\addplot[color=plotColor9,heuristicMarkSign,shortenLines,line width=1pt] table {
    8 61.712450
    9 287.989000
};
\addplot[color=plotColor9,heuristicMarkSign,shortenLines,line width=1pt] table {
    9 287.989000
    10 683.819500
};
\addplot[color=plotColor9,heuristicMarkSign,shortenLines,line width=1pt] table {
    10 683.819500
    11 1094.895000
};
\addplot[color=plotColor9,heuristicMarkSign,shortenLines,line width=1pt] table {
    11 1094.895000
    12 1071.955000
};
\addplot[color=plotColor9,heuristicMarkSign,shortenLines,line width=1pt] table {
    12 1071.955000
    13 1067.890000
};

\end{axis}
\end{tikzpicture}
  \caption{Median running times for different battery capacities. For each considered capacity (ranging from 0.25\,kWh to 1\,024\,kWh), the plot shows the median running time of 100 random in-range queries, provided that no query exceeded an hour of computation time.
  We evaluate the \gls*{bsp} algorithm using multi-arcs corresponding to speed samples, A*~search for \gls*{bsp} with the potential function $\drivingtimepotential$ (\gls*{bsp}-A*), the \gls*{tfp} algorithm using pairwise dominance tests, \gls*{tfp} with improved dominance tests (\gls*{tfp}-d), and our proposed speedup techniques (A*-$\drivingtimepotential$,~\gls*{champ}-$\drivingtimepotential$, and~\gls*{champ}-$\convexpotential$). We also show our heuristic approach  with parameter choice~$\varepsilon=0.1$, denoted \gls*{champ}-$\varepsilon$-$\drivingtimepotential$.}
  \label{fig:champ-limitplot-med}
\end{figure}

\paragraph{Query Times Subject to Vehicle Range.}
While query times are relatively stable for different Dijkstra ranks, the range of an \gls*{ev} has a major influence on performance.
We evaluate our approaches for different battery limits in Figure~\ref{fig:champ-limitplot-med}. For every considered battery capacity, we report the median running time of 100 random in-range queries up to the range at which at least one query did not terminate within one hour.
For small capacities, this enables our basic approaches, which are shown in the plot as well.
We also evaluate the \gls*{bsp} algorithm, using multi-arcs to model speed adaptation (20\,km/h step width).
As before, we observe a considerable speedup by switching to our more realistic model, which is based on consumption functions.
As discussed in Section~\ref{sec:experiments:model}, a vast reduction in the number of label scans more than compensates for the more expensive basic operations when using~\gls*{tfp}. It achieves a speedup by up to two orders of magnitude over~\gls*{bsp}.
Plugging in the improved dominance checks described in Section~\ref{sec:approach} pays off as well, as it yields further speedup for battery capacities beyond~1\,kWh.
Adding A*~search, we achieve reasonable median running times of about a second for capacities of up to~32\,kWh, without any preprocessing.
Our technique \gls*{champ}-$\drivingtimepotential$ adds preprocessing to provide further speedup by another order of magnitude.
Matching our previous observations, median running times of \gls*{champ}-$\convexpotential$ are slower for all ranges up to~32\,kWh. However, this algorithm is more robust against outliers and is the only exact method that terminates within an hour for all queries at 64\,kWh and up. As a result, we are able to compute \emph{provably optimal} results for (hypothetical) ranges of up to 512\,kWh (around 3\,000\,km) in less than an hour.
Finally, our heuristic variant scales well with vehicle range. Query times actually bottom out for large battery capacities, as the vehicle range gets close to the graph diameter (for~1\,024\,kWh, the whole graph is always reachable).

\begin{figure}[t]
  \centering
  \begin{tikzpicture}[figure]
\pgfplotsset{
   grid style = {dash pattern = on 1pt off 1pt, black15,line width = 0.5pt  }
}

\colorlet{plotColor1}{black50} 
\colorlet{plotColor2}{thesisgreen} 
\colorlet{plotColor3}{black50} 
\colorlet{plotColor4}{thesisyellow} 
\colorlet{plotColor5}{thesisgreen} 
\colorlet{plotColor7}{thesisblue} 
\colorlet{plotColor8}{thesisred} 
\colorlet{plotColor9}{thesisred} 

\begin{axis}[
   height=7.0cm,
   width=0.98\textwidth,
   xmin=0.5,
   xmax=13.5,
   ymin=0.05,
   ymax=3600000,
   ymode=log,
   xlabel={Battery capacity [kWh]},
   ylabel={Time [ms]},
   y label style={at={(axis description cs:-0.06,0.5)}},
   /pgf/number format/.cd,
   1000 sep={\,},
   xtick={0, 1, 2, 3, 4, 5, 6, 7, 8, 9, 10, 11, 12, 13},
   xticklabel=\pgfmathparse{2^(\tick-3)}${\pgfmathprintnumber{\pgfmathresult}}$,
   grid=major,
   legend entries={BSP, BSP-A*, \gls{tfp}, \gls{tfp}-d, A*-$\drivingtimepotential$, \gls{champ}-$\drivingtimepotential$, \gls{champ}-$\convexpotential$, \gls{champ}-$\varepsilon$-$\drivingtimepotential$},
   legend cell align=left,
   legend columns=2,
   legend style={at={(0.98,0.04)},
   anchor=south east,
   font=\scriptsize}
]

 \addlegendimage{legend line with heuristic mark,plotColor1}
 \addlegendimage{legend line with heuristic mark,plotColor2}
 \addlegendimage{legend line with exact mark,plotColor3}
 \addlegendimage{legend line with exact mark,plotColor4}
 \addlegendimage{legend line with exact mark,plotColor5}
 \addlegendimage{legend line with exact mark,plotColor8}
 \addlegendimage{legend line with exact mark,plotColor7}
 \addlegendimage{legend line with heuristic mark,plotColor9}

\addplot[color=plotColor1,heuristicMarkSign,shortenLines,line width=1pt] table {
    1 51.586900
    2 6486.300000
};
\addplot[color=plotColor1,heuristicMarkSign,shortenLines,line width=1pt] table {
    2 6486.300000
    3 355139.000000
};

\addplot[color=plotColor2,heuristicMarkSign,shortenLines,line width=1pt] table {
    1 1.979000
    2 14.922900
};
\addplot[color=plotColor2,heuristicMarkSign,shortenLines,line width=1pt] table {
    2 14.922900
    3 7509.380000
};
\addplot[color=plotColor2,heuristicMarkSign,shortenLines,line width=1pt] table {
    3 7509.380000
    4 116941.000000
};
\addplot[color=plotColor2,heuristicMarkSign,shortenLines,line width=1pt] table {
    4 116941.000000
    5 1259560.000000
};

\addplot[color=plotColor3,exactMarkSign,shortenLines,line width=1pt] table {
    1 2.513920
    2 12.625000
};
\addplot[color=plotColor3,exactMarkSign,shortenLines,line width=1pt] table {
    2 12.625000
    3 218.718000
};
\addplot[color=plotColor3,exactMarkSign,shortenLines,line width=1pt] table {
    3 218.718000
    4 13912.900000
};
\addplot[color=plotColor3,exactMarkSign,shortenLines,line width=1pt] table {
    4 13912.900000
    5 1092120.000000
};

\addplot[color=plotColor4,exactMarkSign,shortenLines,line width=1pt] table {
    1 2.223880
    2 13.193800
};
\addplot[color=plotColor4,exactMarkSign,shortenLines,line width=1pt] table {
    2 13.193800
    3 92.388900
};
\addplot[color=plotColor4,exactMarkSign,shortenLines,line width=1pt] table {
    3 92.388900
    4 3706.600000
};
\addplot[color=plotColor4,exactMarkSign,shortenLines,line width=1pt] table {
    4 3706.600000
    5 22712.600000
};
\addplot[color=plotColor4,exactMarkSign,shortenLines,line width=1pt] table {
    5 22712.600000
    6 89202.800000
};

\addplot[color=plotColor5,exactMarkSign,shortenLines,line width=1pt] table {
    1 1.970210
    2 6.040040
};
\addplot[color=plotColor5,exactMarkSign,shortenLines,line width=1pt] table {
    2 6.040040
    3 8.177000
};
\addplot[color=plotColor5,exactMarkSign,shortenLines,line width=1pt] table {
    3 8.177000
    4 28.147900
};
\addplot[color=plotColor5,exactMarkSign,shortenLines,line width=1pt] table {
    4 28.147900
    5 366.910000
};
\addplot[color=plotColor5,exactMarkSign,shortenLines,line width=1pt] table {
    5 366.910000
    6 4253.610000
};
\addplot[color=plotColor5,exactMarkSign,shortenLines,line width=1pt] table {
    6 4253.610000
    7 122847.000000
};
\addplot[color=plotColor5,exactMarkSign,shortenLines,line width=1pt] table {
    7 122847.000000
    8 2289750.000000
};

\addplot[color=plotColor7,exactMarkSign,shortenLines,line width=1pt] table {
    1 58.647000
    2 307.992000
};
\addplot[color=plotColor7,exactMarkSign,shortenLines,line width=1pt] table {
    2 307.992000
    3 161.314000
};
\addplot[color=plotColor7,exactMarkSign,shortenLines,line width=1pt] table {
    3 161.314000
    4 299.535000
};
\addplot[color=plotColor7,exactMarkSign,shortenLines,line width=1pt] table {
    4 299.535000
    5 164.124000
};
\addplot[color=plotColor7,exactMarkSign,shortenLines,line width=1pt] table {
    5 164.124000
    6 677.670000
};
\addplot[color=plotColor7,exactMarkSign,shortenLines,line width=1pt] table {
    6 677.670000
    7 4630.370000
};
\addplot[color=plotColor7,exactMarkSign,shortenLines,line width=1pt] table {
    7 4630.370000
    8 20391.900000
};
\addplot[color=plotColor7,exactMarkSign,shortenLines,line width=1pt] table {
    8 20391.900000
    9 45790.800000
};
\addplot[color=plotColor7,exactMarkSign,shortenLines,line width=1pt] table {
    9 45790.800000
    10 742787.000000
};
\addplot[color=plotColor7,exactMarkSign,shortenLines,line width=1pt] table {
    10 742787.000000
    11 2381610.000000
};
\addplot[color=plotColor7,exactMarkSign,shortenLines,line width=1pt] table {
    11 2381610.000000
    12 2768500.000000
};

\addplot[color=plotColor8,exactMarkSign,shortenLines,line width=1pt] table {
    1 0.592041
    2 1.727780
};
\addplot[color=plotColor8,exactMarkSign,shortenLines,line width=1pt] table {
    2 1.727780
    3 2.299070
};
\addplot[color=plotColor8,exactMarkSign,shortenLines,line width=1pt] table {
    3 2.299070
    4 7.789060
};
\addplot[color=plotColor8,exactMarkSign,shortenLines,line width=1pt] table {
    4 7.789060
    5 57.438000
};
\addplot[color=plotColor8,exactMarkSign,shortenLines,line width=1pt] table {
    5 57.438000
    6 432.170000
};
\addplot[color=plotColor8,exactMarkSign,shortenLines,line width=1pt] table {
    6 432.170000
    7 9671.920000
};
\addplot[color=plotColor8,exactMarkSign,shortenLines,line width=1pt] table {
    7 9671.920000
    8 430648.000000
};

\addplot[color=plotColor9,heuristicMarkSign,shortenLines,line width=1pt] table {
    1 0.550049
    2 1.630130
};
\addplot[color=plotColor9,heuristicMarkSign,shortenLines,line width=1pt] table {
    2 1.630130
    3 2.102050
};
\addplot[color=plotColor9,heuristicMarkSign,shortenLines,line width=1pt] table {
    3 2.102050
    4 5.276120
};
\addplot[color=plotColor9,heuristicMarkSign,shortenLines,line width=1pt] table {
    4 5.276120
    5 46.815900
};
\addplot[color=plotColor9,heuristicMarkSign,shortenLines,line width=1pt] table {
    5 46.815900
    6 183.429000
};
\addplot[color=plotColor9,heuristicMarkSign,shortenLines,line width=1pt] table {
    6 183.429000
    7 418.344000
};
\addplot[color=plotColor9,heuristicMarkSign,shortenLines,line width=1pt] table {
    7 418.344000
    8 3650.860000
};
\addplot[color=plotColor9,heuristicMarkSign,shortenLines,line width=1pt] table {
    8 3650.860000
    9 8429.270000
};
\addplot[color=plotColor9,heuristicMarkSign,shortenLines,line width=1pt] table {
    9 8429.270000
    10 113631.000000
};
\addplot[color=plotColor9,heuristicMarkSign,shortenLines,line width=1pt] table {
    10 113631.000000
    11 353508.000000
};
\addplot[color=plotColor9,heuristicMarkSign,shortenLines,line width=1pt] table {
    11 353508.000000
    12 1039300.000000
};
\addplot[color=plotColor9,heuristicMarkSign,shortenLines,line width=1pt] table {
    12 1039300.000000
    13 1097.670000
};

\end{axis}
\end{tikzpicture}
  \caption{Maximum running times for different battery capacities. For the same sets of queries as in Figure~\ref{fig:champ-limitplot-med}, this plot shows the corresponding maximum running times of 100 random in-range queries (unless this maximum exceeded an hour of computation time).}
  \label{fig:champ-limitplot-max}
\end{figure}
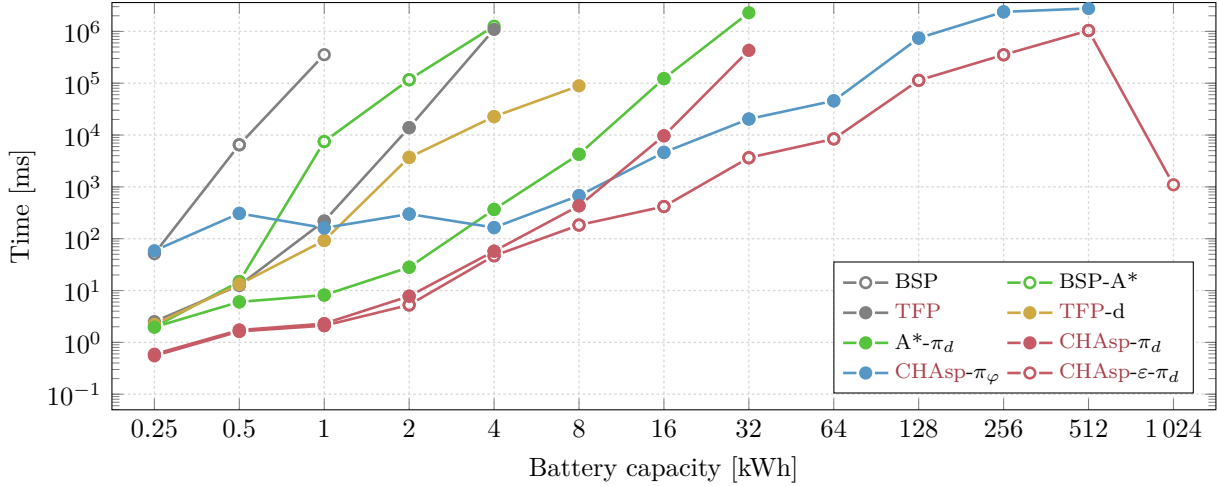

Figure~\ref{fig:champ-limitplot-max} shows maximum running times for the same sets of queries (see Appendix~\ref{app:experiments} for a plot showing average running times). Naturally, these outliers are subject to noise, so curves are less smooth. Our most robust technique,~\gls*{champ}-$\convexpotential$, starts to pays off at a range of 16\,kWh (at lower ranges, its backward search is too expensive). Interestingly, the maximum running time of the heuristic variant of \gls*{champ} drops to about a second for the longest considered range: Since the target is always reachable with either no or only little speed adaptation, goal direction makes the forward search find the target almost immediately. In case of exact~\gls*{champ}-$\convexpotential$, the expensive backward search prevents the algorithm from always terminating within one hour. In practice, it would therefore pay off to first check whether the target is reachable on the fastest path (without speed adaptation) before resorting to either \gls*{champ}-$\drivingtimepotential$ or~\gls*{champ}-$\convexpotential$. The better choice between the two potentials depends on the vehicle range and by how much the target is missed on the unconstrained fastest route.

\section{Conclusion}\label{sec:conclusion}

We introduced a novel framework for computing constrained shortest paths for \glspl*{ev} in practice, using realistic consumption models. Our \gls*{tfp} algorithm respects battery constraints and accounts for adaptive speeds in a mathematically sounder way that unlocks \emph{both} better query performance \emph{and} improved solution quality when compared to previous approaches using discretized, sampled speeds.
Nontrivial speedup techniques based on A*~search and \gls*{ch} make the algorithm practical. It computes \emph{optimal} solutions in less than a second on sensible instances, making it the first practical \emph{exact} approach---with running times similar to previous inexact methods~\cite{Bau14a,Goo14,Har14}.
Our own heuristic enables even faster queries while retaining high-quality solutions.

The result of our computations is not only the suggested route from source to target but also optimal driving speeds along that route.
In practice, these can be passed to the driver as recommendations or directly to a cruise control unit. With the advent of autonomous vehicles, the output of our algorithms can also be utilized for speed planning of self-driving~\glspl*{ev}, either directly or after refinement~\cite{Flo15}.
The \gls*{champ} algorithm can further be used as a subroutine in more complex scenarios, such as fleet management or load-balanced routing. For the latter, computing several Pareto optimal routes could provide a set of alternative candidates to choose from~\cite{Bau14a}. Another next step would be the the integration of planned charging stops~\cite{Bau18,Sto12c}; see also Niklaus~\cite{Nik17}.

To enable faster heuristic variants, it would be useful to precompute potentials for A*~search, as in \acrshort*{alt}~\cite{Gol05a,Gol05b}, or integrate other known approximate or heuristic approaches~\cite{Bau14a,Bat11,Wan16}.
We are also interested in the integration of variable speed limits imposed by, \eg, time-dependent travel times based on historic knowledge of traffic patterns~\cite{Bat13,Bau16b,Del09b,Fos14}.
From a more theoretical perspective, efficient representation and comparison of (general) bivariate \gls*{soc} functions is an open issue.
One could also extend \gls*{evcas} by asking for an optimal solution for \emph{every} initial~\gls*{soc} \wrt given source and target vertices, similar to profile queries in time-dependent or energy-optimal routing~\cite{Bat13,Bau13a}.

\bibliographystyle{plain}
\bibliography{references-trsc}

\clearpage
\appendix

\section{Omitted Details and Proofs from Section~\ref{sec:operations}}\label{app:linking}

In what follows, we provide technical details for linking consumption functions in the special case considered in Section~\ref{sec:operations:linking}, in which both functions represent arcs (Appendix~\ref{app:linking:arcs}), give a formal proof of Lemma~\ref{lem:tradeoff-function-continuous} (Appendix~\ref{app:linking:lemma}), and show how general consumption functions can be linked in linear time (Appendix~\ref{app:linking:linear-time}). Finally, Appendix~\ref{app:linking:dominance} mentions technical details omitted from Section~\ref{sec:operations:dominance}.

\subsection{Linking Functions Defined by Single Tradeoff Functions}\label{app:linking:arcs}

We show how $\tcfunction=\linkop(\tcfunction_1,\tcfunction_2)$ can be computed in constant time in the case where each of the two given consumption functions $\tcfunction_1$ and $\tcfunction_2$ is defined by a \emph{single} tradeoff subfunction, rather than multiple ones (\cf Section~\ref{sec:operations:linking}).
Recall Equation~\eqref{eq:tradeoff-linkedfunction} from Section~\ref{sec:operations:linking} restated below:
\begin{align}
 \tag{\ref{eq:tradeoff-linkedfunction}}
 \tcfunction(\atime) = \min_{\substack{\cvalue\in[\leftintervalborder_1,\rightintervalborder_1]\\\cvalue\in[\atime-\rightintervalborder_2,\atime-\leftintervalborder_2]}}\tcfunction_1(\cvalue)+\tcfunction_2(\atime-\cvalue)\text{.}
\end{align}
As argued in Section~\ref{sec:operations:linking}, the best choice of~$\cvalue$ depends on the value~$\atime\in[\leftintervalborder_1+\leftintervalborder_2,\rightintervalborder_1+\rightintervalborder_2]$.
Therefore, we consider the \emph{$\cvalue$-function} $\cfunction\colon[\leftintervalborder_1+\leftintervalborder_2,\rightintervalborder_1+\rightintervalborder_2]\to\posreals$ that maps every admissible value of~$\atime$ to the optimal choice of~$\cvalue$. Then, we immediately get for arbitrary $\atime\in\posreals$ that
\begin{align}
\tag{\ref{eq:tradeoff-linkedfunction-alt}}
 \tcfunction(\atime) =
 \begin{cases}
  \infty                                                                          & \mbox{if } \atime < \leftintervalborder_1 + \leftintervalborder_2, \\
  \tcfunction_1(\rightintervalborder_1) + \tcfunction_2(\rightintervalborder_2)   & \mbox{if } \atime > \rightintervalborder_1 + \rightintervalborder_2, \\
  \tcfunction_1(\cfunction(\atime)) + \tcfunction_2 (\atime - \cfunction(\atime)) & \mbox{otherwise.}
 \end{cases}
\end{align}
Hence, we essentially need to compute~$\cfunction$ to obtain the desired function~$\tcfunction$.
To this end, consider an arbitrary fixed driving time~$\atime\in[\leftintervalborder_1+\leftintervalborder_2,\rightintervalborder_1+\rightintervalborder_2]$. To identify the optimal value $\cfunction(\atime)$, we examine the derivative~$\tcfunction'_\atime$ of the term $\tcfunction_\atime(\cvalue):=\tcfunction_1(\cvalue)+\tcfunction_2(\atime-\cvalue)$. It evaluates to
\begin{align*}
 \tcfunction'_\atime(\cvalue) = \frac{2\aparam_1}{(\bparam_1 - \cvalue)^3} + \frac{2\aparam_2}{(\atime - \cvalue - \bparam_2)^3} \text{,}
\end{align*}
where $\aparam_1$,~$\bparam_1$,~$\aparam_2$, and $\bparam_2$ are the corresponding coefficients in the tradeoff functions of $\tcfunction_1$ and~$\tcfunction_2$. We assume that $\aparam_1>0$ and $\aparam_2>0$ are positive (the other cases $\aparam_1=0$ or $\aparam_2=0$ are trivial: at least one of the consumption functions is constant in this case and hence, has a unique admissible driving time).
Then, we obtain a unique zero $\cvalue^*_\atime$ for this derivative under the assumption that $\bparam_1<\cvalue$ and~$\cvalue<\atime-\bparam_2$. This holds true for valid choices of~$\cvalue$, because $\cvalue\ge\leftintervalborder_1>\bparam_1$ and $\cvalue\le\atime-\leftintervalborder_2<\atime-\bparam_2$ always holds; see Equation~\eqref{eq:tradeoff-linkedfunction}. Therefore, solving the equation $\tcfunction'_\atime(\cvalue)=0$ for $\cvalue$ yields
\begin{align*}
 \cvalue^*_\atime = \frac{\atime - \bparam_2 + \bparam_1\sqrt[3]{\frac{\aparam_2}{\aparam_1}}}{1 + \sqrt[3]{\frac{\aparam_2}{\aparam_1}}} \text{.}
\end{align*}
The value $\cvalue^*_\atime$ minimizes energy consumption for an unrestricted distribution of driving times that sum up to~$\atime$. From Equation~\eqref{eq:tradeoff-linkedfunction} we get the additional constraints $\cfunction(\atime)\ge\max\{\leftintervalborder_1,\atime-\rightintervalborder_2\}$ and~$\cfunction(\atime)\le\min\{\rightintervalborder_1,\atime-\leftintervalborder_2\}$.
Since $\cvalue^*_\atime$ is the unique zero of $\tcfunction_\atime'$ in the open interval~$(\bparam_1,\atime-\bparam_2)$, monotonicity of $\tcfunction_\atime$ in the intervals $(\bparam_1,\cvalue^*_\atime]$ and $[\cvalue^*_\atime,\atime-\bparam_2)$ follows. Thus, we get
\begin{align}
 \tag{\ref{eq:tradeoff-cfunction}}
 \cfunction(\atime) =
 \begin{cases}
  \max\{\leftintervalborder_1,\atime-\rightintervalborder_2\} & \mbox{if }\cvalue^*_\atime<\max\{\leftintervalborder_1,\atime-\rightintervalborder_2\}\text{,}\\
  \min\{\rightintervalborder_1,\atime-\leftintervalborder_2\} & \mbox{if }\cvalue^*_\atime>\min\{\rightintervalborder_1,\atime-\leftintervalborder_2\}\text{,}\\
  \cvalue^*_\atime & \mbox{otherwise.}
 \end{cases}
\end{align}
We already argued in Section~\ref{sec:operations:linking} that Equation~\eqref{eq:tradeoff-linkedfunction-alt} and Equation~\eqref{eq:tradeoff-cfunction} together are sufficient to specify the desired function~$\tcfunction$. Since we want to explicitly represent $\tcfunction$ using tradeoff functions, we now derive the actual subfunctions that define~$\tcfunction$, depending on the value~$\cvalue^*_\atime$.

First, solving the conditions $\cvalue^*_\atime<\max\{\leftintervalborder_1,\atime-\rightintervalborder_2\}$ and $\cvalue^*_\atime>\min\{\rightintervalborder_2,\atime-\leftintervalborder_2\}$ in Equation~\eqref{eq:tradeoff-cfunction} for $\atime$ yields four equivalences in total, namely
\begin{align}
 \label{eq:tradeoff-cfunction-first}
 &\cvalue^*_\atime < \leftintervalborder_1                 && \Leftrightarrow && \atime < \firstleftboundaryterm := \firstleftboundarytermfull\text{,}\\
 \label{eq:tradeoff-cfunction-second}
 &\cvalue^*_\atime < \atime - \rightintervalborder_2 && \Leftrightarrow && \atime > \secondrightboundaryterm := \secondrightboundarytermfull\text{,}\\
  \label{eq:tradeoff-cfunction-third}
 &\cvalue^*_\atime > \rightintervalborder_1                && \Leftrightarrow && \atime > \firstrightboundaryterm := \firstrightboundarytermfull\text{,}\\
 \label{eq:tradeoff-cfunction-fourth}
 &\cvalue^*_\atime > \atime - \leftintervalborder_2  && \Leftrightarrow && \atime < \secondleftboundaryterm := \secondleftboundarytermfull\text{.}
\end{align}
Note that we obtain similar statements when solving for equality, \eg, we have $\cvalue^*_\atime=\leftintervalborder_1$ if and only if $\atime=\firstleftboundaryterm$. Consequently, we also get $\cvalue^*_\atime>\leftintervalborder_1$ if and only if~$\atime>\firstleftboundaryterm$ (and analogous results for Equations~\eqref{eq:tradeoff-cfunction-second}--\eqref{eq:tradeoff-cfunction-fourth}).
To obtain the actual function~$\tcfunction$ and its subfunctions, we use the following Lemma~\ref{lem:tc-functionslopes}.

\setcounter{lemma}{0}
\renewcommand{\thelemma}{\Alph{section}-\arabic{lemma}}

\begin{lemma}
\label{lem:tc-functionslopes}
 Let~$\tcfunction_1$ and $\tcfunction_2$ be two consumption functions, such that each is defined by a single tradeoff function $\tcinner_1$ and $\tcinner_2$, respectively.
 Moreover, let~$\leftintervalborder_1$,~$\rightintervalborder_1$,~$\leftintervalborder_2$, and~$\rightintervalborder_2$ denote their respective minimum and maximum driving times. 
 Then the following statements hold for their derivatives $\tcfunctionderivative_1$ and $\tcfunctionderivative_2$ (even if we replace all occurrences of the relation ``$\le$'' by ``$=$'').
 \begin{enumerate}
  \item $\tcfunctionrightderivative_1(\leftintervalborder_1) \le \tcfunctionrightderivative_2(\leftintervalborder_2) \Leftrightarrow \firstleftboundaryterm \le \leftintervalborder_1 + \leftintervalborder_2 \le \secondleftboundaryterm$,
  \item $\tcfunctionrightderivative_1(\rightintervalborder_1) \le \tcfunctionrightderivative_2(\rightintervalborder_2) \Leftrightarrow \firstrightboundaryterm \le \rightintervalborder_1 + \rightintervalborder_2 \le \secondrightboundaryterm$,
  \item $\tcfunctionleftderivative_1(\rightintervalborder_1) \le \tcfunctionrightderivative_2(\leftintervalborder_2) \Leftrightarrow \firstrightboundaryterm \le \secondleftboundaryterm$,
  \item $\tcfunctionrightderivative_1(\leftintervalborder_1) \le \tcfunctionleftderivative_2(\rightintervalborder_2) \Leftrightarrow \firstleftboundaryterm \le \secondrightboundaryterm$.
 \end{enumerate}
\end{lemma}

\begin{proof}
  All equivalences follow after simple rearrangements. As an example, we show the first part of the first statement, namely, $\tcfunctionrightderivative_1(\leftintervalborder_1) \le \tcfunctionrightderivative_2(\leftintervalborder_2)$ if and only if $\firstleftboundaryterm\le\leftintervalborder_1+\leftintervalborder_2$. For $i\in\{1,2\}$, let $\aparam_i$ and $\bparam_i$ denote the coefficients of the tradeoff function $\tcinner_i$ as in Equation~\eqref{eq:tradeoff-functioninnerform}. We exploit that~$\aparam_1>0$,~$\aparam_2>0$,~$\leftintervalborder_1>\bparam_1$, and~$\leftintervalborder_2>\bparam_2$ must hold to get
  \begin{align*}
                   && \tcfunctionrightderivative_1(\leftintervalborder_1)                      &\le \tcfunctionrightderivative_2(\leftintervalborder_2)\\
   \Leftrightarrow && -\frac{2 \aparam_1}{(\leftintervalborder_1 - \bparam_1)^3} &\le -\frac{2 \aparam_2}{(\leftintervalborder_2 - \bparam_2)^3}\\
   \Leftrightarrow && (\leftintervalborder_2 - \bparam_2)^3                      &\ge \frac{2\aparam_2}{2\aparam_1}(\leftintervalborder_1 - \bparam_1)^3\\
   \Leftrightarrow && \leftintervalborder_2                                      &\ge \bparam_2 + \sqrt[3]{\frac{\aparam_2}{\aparam_1}}(\leftintervalborder_1 - \bparam_1)\\
   \Leftrightarrow && \leftintervalborder_1 + \leftintervalborder_2              &\ge \leftintervalborder_1 + \bparam_2 + \sqrt[3]{\frac{\aparam_2}{\aparam_1}}(\leftintervalborder_1 - \bparam_1) = \firstleftboundaryterm\text{.}
  \end{align*}
  All other statements follow from similar rearrangements.
\end{proof}

Together with Equations~\eqref{eq:tradeoff-cfunction-first}--\eqref{eq:tradeoff-cfunction-fourth}, Lemma~\ref{lem:tc-functionslopes} enables us to construct the desired function~$\tcfunction=\linkop(\tcfunction_1,\tcfunction_2)$, depending on the slopes (\ie, the derivatives) of $\tcfunction_1$ and $\tcfunction_2$ at their respective subdomain borders.
Exploiting that~$\leftintervalborder_1\le\rightintervalborder_1$ and $\leftintervalborder_2\le\rightintervalborder_2$ hold by definition, we obtain the function $\tcfunction$ after the following case distinction.
As in Lemma~\ref{lem:tc-functionslopes}, let $\tcinner_1$ denote the (unique) tradeoff function defining $\tcfunction_1$ and similarly, let $\tcinner_2$ be the tradeoff function defining~$\tcfunction_2$. We consider the slopes of these tradeoff functions at certain subdomain borders of $\tcfunction_1$ and~$\tcfunction_2$.
Without loss of generality, assume that $\tcinner'_1(\leftintervalborder_1)\le\tcinner'_2(\leftintervalborder_2)$ holds (the other case is symmetric). This leaves us with only three possible cases, which are presented below.

\begin{enumerate}
 \item $\tcfunctionrightderivative_1(\leftintervalborder_1) \le \tcfunctionrightderivative_2(\leftintervalborder_2) \le \tcfunctionleftderivative_1(\rightintervalborder_1) \le \tcfunctionleftderivative_2(\rightintervalborder_2)$:
 Consider the relevant subdomain borders $\leftintervalborder_1+\leftintervalborder_2$ and $\rightintervalborder_1+\rightintervalborder_2$ corresponding to the minimum and maximum driving time of the function $\tcfunction$ obtained after linking $\tcfunction_1$ and~$\tcfunction_2$.
 By the first, second, and third statement of Lemma~\ref{lem:tc-functionslopes}, we know that 
 \begin{align*}
  \firstleftboundaryterm\le\leftintervalborder_1+\leftintervalborder_2\le\secondleftboundaryterm\le\firstrightboundaryterm\le\rightintervalborder_1+\rightintervalborder_2\le\secondrightboundaryterm\text{.}  
 \end{align*}
 For arbitrary values $\atime\in[\leftintervalborder_1+\leftintervalborder_2,\secondleftboundaryterm)$, we use the fact $\atime\ge\firstleftboundaryterm$ and Equation~\eqref{eq:tradeoff-cfunction-first} to infer the inequality~$\cvalue^*_\atime\ge\leftintervalborder_1$. Similarly, the fact $\atime<\secondrightboundaryterm$ and Equation~\eqref{eq:tradeoff-cfunction-second} yield~$\cvalue^*_\atime>\atime-\rightintervalborder_2$.
 Hence, we have~$\cvalue^*_\atime\ge\max\{\leftintervalborder_1,\atime-\rightintervalborder_2\}$ and the first case of Equation~\eqref{eq:tradeoff-cfunction} does not apply.
 On the other hand, we know that both $\atime<\firstrightboundaryterm$ and $\atime<\secondleftboundaryterm$ hold, so by Equation~\eqref{eq:tradeoff-cfunction-third} and Equation~\eqref{eq:tradeoff-cfunction-fourth} we get~$\atime-\leftintervalborder_2<\cvalue^*_\atime<\rightintervalborder_1$. This means that we are in the second case of Equation~\eqref{eq:tradeoff-cfunction} and therefore, $\cfunction(\atime)=\min\{\rightintervalborder_1,\atime-\leftintervalborder_2\}=\atime-\leftintervalborder_2$ is the optimal choice for any $\atime\in[\leftintervalborder_1+\leftintervalborder_2,\secondleftboundaryterm)$.
 Making similar observations for the cases $\atime\in[\secondleftboundaryterm,\firstrightboundaryterm)$ and~$\atime\in[\firstrightboundaryterm,\rightintervalborder_1+\rightintervalborder_2)$, we obtain corresponding values $\cfunction(\atime)=\cvalue^*_\atime$ and $\cfunction(\atime)=\rightintervalborder_1$, respectively. This yields the desired function~$\tcfunction=\linkop(\tcfunction_1,\tcfunction_2)$, which is given by
 \begin{align*}
   \tcfunction(\atime) =
     \begin{cases}
       \infty                                                                                      & \mbox{if } \atime < \leftintervalborder_1 + \leftintervalborder_2,\\
       \tcfunction_1(\atime - \leftintervalborder_2) + \tcfunction_2 (\leftintervalborder_2)       & \mbox{if } \leftintervalborder_1 + \leftintervalborder_2 \le \atime < \secondleftboundaryterm,\\
       \tcfunction_1(\cvalue^*_\atime) + \tcfunction_2 (\atime - \cvalue^*_\atime)                 & \mbox{if } \secondleftboundaryterm \le \atime < \firstrightboundaryterm,\\
       \tcfunction_1(\rightintervalborder_1) + \tcfunction_2 (\atime - \rightintervalborder_1)     & \mbox{if } \firstrightboundaryterm \le \atime < \rightintervalborder_1 + \rightintervalborder_2,\\
       \tcfunction_1(\rightintervalborder_1) + \tcfunction_2(\rightintervalborder_2)               & \mbox{otherwise.}
     \end{cases}
 \end{align*}
 \item $\tcfunctionrightderivative_1(\leftintervalborder_1) \le \tcfunctionleftderivative_1(\rightintervalborder_1) \le \tcfunctionrightderivative_2(\leftintervalborder_2) \le \tcfunctionleftderivative_2(\rightintervalborder_2)$: In this case, we know by the four statements in Lemma~\ref{lem:tc-functionslopes} that the order $\firstleftboundaryterm\le\firstrightboundaryterm\le\leftintervalborder_1+\leftintervalborder_2\le\rightintervalborder_1+\rightintervalborder_2\le\secondleftboundaryterm\le\secondrightboundaryterm$ must hold.
 Equations~\eqref{eq:tradeoff-cfunction-first}--\eqref{eq:tradeoff-cfunction-fourth} yield $\cvalue^*_\atime\ge\max\{\leftintervalborder_1,\atime-\rightintervalborder_2\}$ and $\cvalue^*_\atime>\min\{\rightintervalborder_1,\atime-\leftintervalborder_2\}$ in the whole subdomain $[\leftintervalborder_1+\leftintervalborder_2,\rightintervalborder_1+\rightintervalborder_2)$.
 Consequently, we obtain the optimal value $\cfunction(\atime)=\min\{\rightintervalborder_1,\atime-\leftintervalborder_2\}$ and the function $\tcfunction=\linkop(\tcfunction_1,\tcfunction_2)$ is defined as
 \begin{align*}
   \tcfunction(\atime) =
     \begin{cases}
       \infty                                                                                      & \mbox{if } \atime < \leftintervalborder_1 + \leftintervalborder_2,\\
       \tcfunction_1(\rightintervalborder_1) + \tcfunction_2(\atime - \rightintervalborder_1)      & \mbox{if } \leftintervalborder_1 + \leftintervalborder_2 \le \atime < \rightintervalborder_1 + \leftintervalborder_2,\\
       \tcfunction_1(\atime - \leftintervalborder_2) + \tcfunction_2 (\leftintervalborder_2)       & \mbox{if } \rightintervalborder_1 + \leftintervalborder_2 \le \atime < \rightintervalborder_1 + \rightintervalborder_2,\\
       \tcfunction_1(\rightintervalborder_1) + \tcfunction_2(\rightintervalborder_2)               & \mbox{otherwise.}
     \end{cases}
 \end{align*}
 \item $\tcfunctionrightderivative_1(\leftintervalborder_1) \le \tcfunctionrightderivative_2(\leftintervalborder_2) \le \tcfunctionleftderivative_2(\rightintervalborder_2) \le \tcfunctionleftderivative_1(\rightintervalborder_1)$: Along the lines of the first case, Lemma~\ref{lem:tc-functionslopes} and Equations~\eqref{eq:tradeoff-cfunction-first}--\eqref{eq:tradeoff-cfunction-fourth} yield that the function $\tcfunction=\linkop(\tcfunction_1,\tcfunction_2)$ evaluates to
  \begin{align*}
   \tcfunction(\atime) =
     \begin{cases}
       \infty                                                                                      & \mbox{if } \atime < \leftintervalborder_1 + \leftintervalborder_2,\\
       \tcfunction_1(\atime - \leftintervalborder_2) + \tcfunction_2 (\leftintervalborder_2)       & \mbox{if } \leftintervalborder_1 + \leftintervalborder_2 \le \atime < \secondleftboundaryterm,\\
       \tcfunction_1(\cvalue^*_\atime) + \tcfunction_2 (\atime - \cvalue^*_\atime)                 & \mbox{if } \secondleftboundaryterm \le \atime < \secondrightboundaryterm,\\
       \tcfunction_1(\atime - \rightintervalborder_2) + \tcfunction_2 (\rightintervalborder_2)     & \mbox{if } \secondrightboundaryterm \le \atime < \rightintervalborder_1 + \rightintervalborder_2,\\
       \tcfunction_1(\rightintervalborder_1) + \tcfunction_2(\rightintervalborder_2)               & \mbox{otherwise.}
     \end{cases}
 \end{align*}
\end{enumerate}

Due to convexity of both $\tcinner_1$ and $\tcinner_2$, no other cases remain.
Hence, the function $\tcfunction$ constructed above is defined by at most five subfunctions (two of which are constant). In each expression, we can expand the functions $\tcfunction_1$ and~$\tcfunction_2$ to obtain a term that has the general form of a tradeoff function as in Equation~\eqref{eq:tradeoff-functioninnerform}. In particular, the denominator in the tradeoff functions of both $\tcfunction_1$ and $\tcfunction_2$ is (strictly) positive in all cases, \ie, we always have~$\atime>\bparam$.
Moreover, it is easy to verify that $\tcfunction$ is continuous and decreasing on the interval~$[\leftintervalborder_1+\leftintervalborder_2,\infty)$, by inspecting the corresponding limits at the endpoints of each subdomain of~$\tcfunction$.
In conclusion, the link operation requires constant time in the special case where both input functions are defined by a single tradeoff function.
Lemma~\ref{lem:tc-simplelink} from Section~\ref{sec:operations:linking} summarizes these insights. It is restated below.

\setcounter{lemma}{0}
\renewcommand{\thelemma}{\arabic{lemma}}

\begin{lemma}
 Given two consumption functions~$\tcfunction_1$ and $\tcfunction_2$, each defined by a single tradeoff subfunction, the link operation $\linkop(\tcfunction_1,\tcfunction_2)$ requires constant time and its result is a consumption function that is uniquely described by at most three tradeoff subfunctions.
\end{lemma}

\subsection{Proof of Lemma~\ref{lem:tradeoff-function-continuous}}\label{app:linking:lemma}

Lemma~\ref{lem:tradeoff-function-continuous} from Section~\ref{sec:operations:linking} claims that linking two consumption functions results in a function that is continuous on its interval of admissible driving time. The proof is given below.
\begin{lemma}
 Given two consumption functions $\tcfunction_1$ and~$\tcfunction_2$, the function $\tcfunction:=\linkop(\tcfunction_1,\tcfunction_2)$ is continuous on the interval~$[\leftintervalborder,\infty)$, where $\leftintervalborder\in\strictposreals$ denotes the minimum driving time of $\tcfunction$.
\end{lemma}

\begin{proof}
Assume for contradiction that $\tcfunction$ has a discontinuity at some~$\atime>\leftintervalborder$. By construction of~$\tcfunction$, a discontinuity in the interval $[\leftintervalborder,\infty)$ corresponds to a discontinuity of some candidate function $\tcfunction^*$ induced by a tradeoff subfunction of $\tcfunction_1$ or~$\tcfunction_2$ (see Section~\ref{sec:operations:linking}). Without loss of generality, let $\tcfunction^*$ be a subfunction of~$\tcfunction_1$. We know that $\tcfunction^*$ has exactly one discontinuity at its minimum driving time~$\leftintervalborder^*\in\strictposreals$. Thus, we have~$\tcfunction(\atime)=\tcfunction_1(\leftintervalborder^*)+\tcfunction_2(\atime-\leftintervalborder^*)$. We distinguish two cases.

\case{Case 1} If $\leftintervalborder^*$ is not the minimum driving time of~$\tcfunction_1$, we know that $\tcfunction_1$ is continuous in the neighborhood of~$\leftintervalborder^*$. Since $\tcfunction$ is decreasing and has a discontinuity at~$\atime$, there exists an $\varepsilon>0$ such that~$\tcfunction_1(\leftintervalborder^*-\varepsilon)+\tcfunction_2(\atime-\leftintervalborder^*)<\tcfunction(\atime-\varepsilon)$, contradicting the fact that $\tcfunction$ minimizes this term.

\case{Case 2} If $\leftintervalborder^*$ is in fact the minimum driving time of~$\tcfunction_1$, we know that the corresponding driving time $\atime-\leftintervalborder^*$ of $\tcfunction_2$ must exceed its minimum driving time~$\leftintervalborder_2$, since~$\atime>\leftintervalborder=\leftintervalborder^*+\leftintervalborder_2$. Hence, $\tcfunction_2$ is continuous in the neighborhood of~$\atime-\leftintervalborder^*$ and along the lines of the first case we obtain a contradiction, because now $\tcfunction_1(\leftintervalborder^*)+\tcfunction_2(\atime-\leftintervalborder^*-\varepsilon)<\tcfunction(\atime-\varepsilon)$ must hold for some~$\varepsilon>0$.
\end{proof}

\subsection{Linking General Consumption Functions in Linear Time}\label{app:linking:linear-time}

In Section~\ref{sec:operations:linking}, we described a na\"{i}ve link operation, which has quadratic running time in the number of subfunctions of $\tcfunction_1$ and~$\tcfunction_2$. In what follows, we show how the complexity of the link operation for general consumption functions can be reduced to linear time.
We say that a consumption function $\tcfunction$ with minimum driving time $\leftintervalborder\in\strictposreals$ is \emph{convex} if it is convex on the interval $[\leftintervalborder,\infty)$. Note that this holds true for consumption functions of single arcs; see Section~\ref{sec:model}.
Our more sophisticated link operation exploits the fact that consumption functions are always convex, which we now prove formally.
Consider the $\cvalue$-function~$\cfunction$ of~$\tcfunction=\linkop(\tcfunction_1,\tcfunction_2)$, defined as the optimal choice of $\cvalue$ in Equation~\eqref{eq:tradeoff-linkedfunction}. See Figure~\ref{fig:tradeoff-c-function} for an example. (Note that we did not formally prove that the value $\cvalue$ is distinct in the general case, but we may as well pick the \emph{minimum} value $\cvalue$ that fulfills Equation~\eqref{eq:tradeoff-linkedfunction} to ensure that $\cfunction$ is well-defined.)
Presuming that $\tcfunction_1$ and $\tcfunction_2$ are convex, the following Lemma~\ref{lem:tradeoff-increasingcfunction} shows that both $\cfunction$ and the value $\atime-\cfunction(\atime)$ increase \wrt~$\atime\in[\leftintervalborder_1+\leftintervalborder_2,\rightintervalborder_1+\rightintervalborder_2]$.

\begin{figure}[t]
 \centering%
 \ConstrainedExampleDrawRealisticTradeoffFunctionLinkedFunctionWithCFunctions
 \caption{The functions $\cfunction$ and~$\atime-\cfunction$~(red), indicating how the available time is shared between the functions $\tcfunction_1$ and~$\tcfunction_2$ from Figure~\ref{fig:tradeoff-example-linking}, in order to obtain the function~$\tcfunction$.}%
\label{fig:tradeoff-c-function}%
\end{figure}

\setcounter{lemma}{1}
\renewcommand{\thelemma}{\Alph{section}-\arabic{lemma}}

\begin{lemma}
 \label{lem:tradeoff-increasingcfunction}
 Let~$\cfunction\colon[\leftintervalborder_1+\leftintervalborder_2,\rightintervalborder_1+\rightintervalborder_2]\to\posreals$ denote the $\cvalue$-function for two convex consumption functions~$\tcfunction_1$ and~$\tcfunction_2$ with corresponding minimum and maximum driving times~$\leftintervalborder_1$,~$\rightintervalborder_1$,~$\leftintervalborder_2$, and~$\rightintervalborder_2$. Moreover, let $\cfunctionsecond(\atime):=\atime-\cfunction(\atime)$ be defined on the domain of~$\cfunction$. Then both $\cfunction$ and $\cfunctionsecond$ are continuous and increasing.
\end{lemma}

\begin{proof}
 We begin by showing that $\cfunction$ is increasing. Assume for contradiction that this is not the case, \ie, for some value $\atime$ in the domain of~$\cfunction$, there are values~$\varepsilon>0$ and~$\delta>0$ such that $\cvalue=\cfunction(\atime)>\cfunction(\atime+\varepsilon)=\cvalue-\delta$.
 First of all, note that the inequality $\tcfunction_1(\cvalue)+\tcfunction_2(\atime-\cvalue)\le\tcfunction_1(\cvalue-\delta)+\tcfunction_2(\atime-\cvalue+\delta)$ must hold, since $\cvalue$ minimizes this term by definition of~$\cfunction$. Further, $\cvalue=\cfunction(\atime)$ is the \emph{smallest} among all values that minimize the term by definition, so plugging in $\cvalue-\delta<\cvalue$ actually yields a \emph{strictly} greater result.
 Analogously, we have $\tcfunction_1(\cvalue-\delta)+\tcfunction_2(\atime+\varepsilon-\cvalue+\delta)\le\tcfunction_1(\cvalue)+\tcfunction_2(\atime+\varepsilon-\cvalue)$, as this term is minimized by~$\cvalue-\delta$. Therefore, we obtain
 \begin{align*}
  \tcfunction_1(\cvalue - \delta) - \tcfunction_1(\cvalue) &\le \tcfunction_2(\atime + \varepsilon - \cvalue) - \tcfunction_2(\atime + \varepsilon - \cvalue + \delta)\\
  &\le \tcfunction_2(\atime - \cvalue) - \tcfunction_2(\atime - \cvalue + \delta)\\
  &< \tcfunction_1(\cvalue - \delta) - \tcfunction_1(\cvalue),
 \end{align*}
 a contradiction. Here, we exploit the fact that $\tcfunction_2$ is convex and decreasing and hence, $\tcfunction_2(\atime-\cvalue)-\tcfunction_2(\atime-\cvalue+\delta)$ must be decreasing \wrt $\atime$ for fixed values $\cvalue$ and~$\delta$ (the gap between two function values with constant difference on the x-axis must decrease if $\atime$ increases).

 Regarding $\cfunctionsecond$, monotonicity follows from a very similar argument. As before, assume for contradiction that $\cfunctionsecond(\atime)>\cfunctionsecond(\atime+\varepsilon)$ for some~$\varepsilon>0$, so the inequality $\atime-\cfunction(\atime)>\atime+\varepsilon-\cfunction(\atime+\varepsilon)$ holds. We plug in $\cvalue=\cfunction(\atime)$ and $\cvalue+\delta=\cfunction(\atime+\varepsilon)$ to
 obtain $\delta>\varepsilon>0$.
 As in the first case, we get $\tcfunction_1(\cvalue)+\tcfunction_2(\atime-\cvalue)\le\tcfunction_1(\cvalue+\delta)+\tcfunction_2(\atime-\cvalue-\delta)$ and $\tcfunction_1(\cvalue+\delta)+\tcfunction_2(\atime+\varepsilon-\cvalue-\delta)<\tcfunction_1(\cvalue)+\tcfunction_2(\atime+\varepsilon-\cvalue)$ by the definition of $\cfunction$.
 Along the lines of the first case, this yields a contradiction.
 
 Finally, we show that $\cfunction$ and $\cfunctionsecond$ are continuous. Since the sum of $\cfunction(\atime)$ and $\atime-\cfunction(\atime)$ always equals~$\atime$, their sum increases by exactly $\varepsilon$ for any $\atime+\varepsilon$ with $\varepsilon>0$ in the neighborhood of~$\atime$. Hence, a discontinuity in either function would imply that at least one of the two terms must decrease, but we showed before that both $\cfunction(\atime)$ and $\atime-\cfunction(\atime)$ increase \wrt~$\atime$.
\end{proof}

Note that in order to show that a consumption function $\tcfunction$ with minimum driving time $\leftintervalborder$ and maximum driving time $\rightintervalborder$ is convex, it is sufficient to consider the interval~$[\leftintervalborder,\rightintervalborder)$, since we already know that $\tcfunction$ is decreasing and that for all $\atime\in[\rightintervalborder,\infty)$, the value $\tcfunction(\atime)=\tcfunction(\rightintervalborder)$ is constant.
The following Lemma~\ref{lem:tradeoff-convexity} proves this property and hence, implies that linking two convex consumption functions indeed yields a decreasing and convex function. Consequently, Lemma~\ref{lem:tradeoff-increasingcfunction} applies to all consumption functions of the general form.

\begin{lemma}
 \label{lem:tradeoff-convexity}
 Given two consumption functions $\tcfunction_1$ and~$\tcfunction_2$ that are convex on their subdomains $[\leftintervalborder_1,\rightintervalborder_1]$ and $[\leftintervalborder_2,\rightintervalborder_2]$ of admissible driving times, the function $\tcfunction:=\linkop(\tcfunction_1,\tcfunction_2)$ is convex on the interval $[\leftintervalborder_1+\leftintervalborder_2,\rightintervalborder_1+\rightintervalborder_2)$.
\end{lemma}

\begin{proof}
 Assume for the sake of contradiction that $\tcfunction=\linkop(\tcfunction_1,\tcfunction_2)$ is not convex on the indicated interval. We use the previous Lemma~\ref{lem:tradeoff-increasingcfunction}, which implies that the $\cvalue$-function \wrt $\tcfunction_1$ and $\tcfunction_2$ is increasing. Moreover, both (right) derivatives $\tcfunction'_1$ and $\tcfunction'_2$ are increasing on their respective subdomains $[\leftintervalborder_1,\rightintervalborder_1)$ and $[\leftintervalborder_2,\rightintervalborder_2)$ by assumption, as $\tcfunction_1$ and $\tcfunction_2$ are decreasing and convex. Given that $\tcfunction$ is not convex, its (right) derivative $\tcfunction'$ must be decreasing on some subinterval of~$[\leftintervalborder_1+\leftintervalborder_2,\rightintervalborder_1+\rightintervalborder_2)$. Thus, there exist values $\atime\in[\leftintervalborder_1+\leftintervalborder_2,\rightintervalborder_1+\rightintervalborder_2)$ and $\varepsilon>0$ such that $\atime+\varepsilon<\rightintervalborder_1+\rightintervalborder_2$ and we get
 \begin{align*}
  \tcfunction'(\atime) &> \tcfunction'(\atime + \varepsilon)\\
   &= \tcfunction'_1(\cfunction(\atime + \varepsilon)) + \tcfunction'_2(\atime + \varepsilon - \cfunction(\atime + \varepsilon))\\
   &\ge \tcfunction'_1(\cfunction(\atime)) + \tcfunction'_2(\atime - \cfunction(\atime))\\
   &= \tcfunction'(\atime),
 \end{align*}
 a contradiction. This completes the proof.
\end{proof}

\paragraph{A Linear-Time Algorithm for Linking Consumption Functions.}
Given that the functions $\cfunction$ and $\cfunctionsecond$ (as defined in Lemma~\ref{lem:tradeoff-increasingcfunction}) of an arbitrary pair of consumption functions are continuous and increasing, we are able to perform the link operation in a single coordinated linear scan, in which we keep track of~$\cfunction$ and~$\cfunctionsecond$.
For two piecewise functions~$\tcfunction_1$ and~$\tcfunction_2$, let $\tcinner_1^1,\dots,\tcinner_1^k$ and $\tcinner_2^1,\dots,\tcinner_2^\ell$ denote their defining tradeoff functions, given in increasing order of their subdomains. For some subfunction~$\tcinner_1^{\smash{j}}$ with $j\in\{1,\dots,k\}$ or $\tcinner_2^{\smash{j}}$ with~$j\in\{1,\dots,\ell\}$, we denote by $[\leftintervalborder_i^{\smash{j}},\rightintervalborder_i^{\smash{j}})$ its subdomain and by $\tcfunction_i^{\smash{j}}$ its induced consumption function.
The linear-time link operation proceeds as follows. First, it links the consumption functions $\tcfunction_1^1$ and $\tcfunction_2^1$ induced by the two tradeoff functions $\tcinner_1^1$ and $\tcinner_2^1$ with least admissible driving times.
This results in a new convex consumption function $\linkop(\tcfunction_1^1,\tcfunction_2^1)$, which is defined by at most three tradeoff functions. Let $\cfunction$ and $\cfunctionsecond$ be the $\cvalue$-functions associated with this link operation.
We determine the next pair of consumption functions that are linked. To this end, we consider the points $\atime_1^1:=\cfunction^{-1}(\rightintervalborder_1^1)$ and $\atime_2^1:=\cfunctionsecond^{-1}(\rightintervalborder_2^1)$ at which the induced consumption functions $\tcfunction_1^1$ and $\tcfunction_2^1$ reach their maximum driving time in the linked function.
If~$\atime_1^1<\atime_2^1$, we continue with $\linkop(\tcfunction_1^2,\tcfunction_2^1)$, so $\cfunction$ can take values greater than $\rightintervalborder_1^1$. Similarly, if $\atime_1^1>\atime_2^1$ holds, we compute $\linkop(\tcfunction_1^1,\tcfunction_2^2)$ next, so that $\cfunctionsecond$ may exceed~$\rightintervalborder_2^1$. Finally, in the special case~$\atime_1^1=\atime_2^1$ we proceed with~$\linkop(\tcfunction_1^2,\tcfunction_2^2)$.
We continue this procedure until we reach the maximum driving time and link the consumption functions induced by $\tcinner_1^k$ and~$\tcinner_2^\ell$. The lower envelope of the (linear number of) computed consumption functions yields the desired result~$\linkop(\consumptionfunction_1,\consumptionfunction_2)$. Obviously, the running time of this procedure is in~$\bigO(k+\ell)$.

\subsection{Omitted Details from Section~\ref{sec:operations:dominance}}\label{app:linking:dominance}

The unique extreme point of the difference $\tcinner_1^i-\tcinner_2^{\smash{j}}$ of the subfunctions $\tcinner_1^i$ and $\tcinner_2^{\smash{j}}$ of two consumption functions (\cf Section~\ref{sec:operations:dominance}) is given by
\begin{align*}
 \atime=\frac{\bparam_2-\bparam_1\sqrt[3]{\frac{\aparam_2}{\aparam_1}}}{1-\sqrt[3]{\frac{\aparam_2}{\aparam_1}}}\text{.}
\end{align*}

\section{Omitted Details and Proofs from Section~\ref{sec:approach}}\label{app:approach}

Below, Appendix~\ref{app:approach:paths} mentions technical details from Section~\ref{sec:approach:paths}, whereas Appendix~\ref{app:approach:lemma} proves Lemma~\ref{lem:poly-label-size-heuristic} from Section~\ref{sec:approach:heuristic}.

\subsection{Omitted Details from Section~\ref{sec:approach:paths}}\label{app:approach:paths}

As mentioned in Section~\ref{sec:approach:paths}, we store a function $\cfunctionsecond$ with each label~$\tcfunction$ to enable path and speed retrieval after the search has terminated. The function $\cfunctionsecond$ is given as $\cfunctionsecond(\atime)=\atime-\cfunction(\atime)$, \wrt the link operation that resulted in the label~$\tcfunction$, and it is a piecewise function with subfunctions that have the general form
\begin{align*}
 \cfunctionsecond(\atime)=\atime-\frac{\atime-\pparam}{\qparam}\text{,}
\end{align*}
with nonnegative coefficients $\pparam\in\posreals$ and~$\qparam\in\strictposreals$; \cf Equation~\eqref{eq:tradeoff-cfunction} in Section~\ref{sec:operations:linking}.

\subsection{Proof of Lemma~\ref{lem:poly-label-size-heuristic}}\label{app:approach:lemma}

Lemma~\ref{lem:poly-label-size-heuristic} proves that the number of settled labels per vertex of the heuristic approach described in Section~\ref{sec:approach:heuristic} is upper bounded by~$\lceil1/\varepsilon\rceil+1$, which implies that the algorithm itself runs in polynomial time.

\setcounter{lemma}{2}
\renewcommand{\thelemma}{\arabic{lemma}}

\begin{lemma}
 The number of settled labels contained in the set $\labelset_\text{set}(\vertex)$ of each vertex $\vertex\in\vertices$ is at most $\lceil1/\varepsilon\rceil+1$ when running \gls*{tfp} with exact improved dominance checks.
\end{lemma}

\begin{proof}
 Let $k:=\lceil1/\varepsilon\rceil+1$ and assume for contradiction that after running~\gls*{tfp}, the label set $\labelset_\text{set}(\vertex)$ of some vertex~$\vertex\in\vertices$ contains at least $k+1$ consumption functions. We denote these functions by $\tcfunction_1,\dots,\tcfunction_{k+1}$ and without loss of generality, we assume they are given in increasing order of minimum driving time (and hence, were inserted into the label set in this order). Since exact dominance checks are applied, we know that $\tcfunction_2$ yields an improvement over $\tcfunction_1$ by at least $\varepsilon\maxbattery$ at its minimum driving time~$\leftintervalborder_2\in\strictposreals$. Consequently, we have~$\tcfunction_2(\leftintervalborder_2)\le\tcfunction_1(\leftintervalborder_2)-\varepsilon\maxbattery$. We can apply the same argument to any function $\tcfunction_i$ with~$i\in\{2,\dots,k+1\}$ and it follows that~$\tcfunction_{i}(\leftintervalborder_{i})\le\tcfunction_{i-1}(\leftintervalborder_{i})-\varepsilon\maxbattery$. Thus, we obtain
 \begin{align*}
  \tcfunction_{k+1}(\leftintervalborder_{k+1})&\le\tcfunction_k(\leftintervalborder_{k+1})-\varepsilon\maxbattery\\
                                              &\le\tcfunction_1(\leftintervalborder_2)-k\varepsilon\maxbattery\\
                                              &\le\maxbattery-\left(\left\lceil\frac{1}{\varepsilon}\right\rceil+1\right)\varepsilon\maxbattery\\
                                              &<0\text{.}
 \end{align*}
 This contradicts the fact that $\tcfunction_{k+1}$ must have nonnegative consumption for all admissible driving times, since the algorithm ensures that battery constraints are not violated for a consumption function before running any dominance checks with it.
\end{proof}

\section{Omitted Details and Proofs from Section~\ref{sec:astar}}\label{app:astar}

This appendix contains omitted details about the potential function based on piecewise linear functions, introduced in Section~\ref{sec:astar:pwl-functions}. We first provide pseudocode and a brief description of the backward search to compute the potential function $\convexpotential$ (Appendix~\ref{app:astar:backward-search}). Afterwards, we formally prove that the resulting potential function is consistent (Appendix~\ref{app:astar:lemma}).

\subsection{Pseudocode of the Label-Correcting Backward Search}\label{app:astar:backward-search}

Algorithm~\ref{alg:tradeoff-potential-backward-search} shows pseudocode of the search that is executed from the target $\target\in\vertices$ in order to compute the functions representing the vertex potential $\convexpotential\colon\vertices\times[0,\maxbattery]\cup\{-\infty\}\to\posreals\cup\{\infty\}$.
Recall that each vertex stores a \emph{single} label consisting of its (tentative) lower bound function.
The search is initialized with a function $\convexlowerboundfunction_\target$ at the target~$\target\in\vertices$ that evaluates to~$0$ for arbitrary nonnegative~\gls*{soc}, represented by the single breakpoint~$(0,0)$; see lines~\ref{line:tfp-pot-initialize-begin}--\ref{line:tfp-pot-initialize-end} of Algorithm~\ref{alg:tradeoff-potential-backward-search}.
In the main loop, scanning an outgoing arc~$(\vertexa,\vertexb)\in\arcs$ of some vertex $\vertexa\in\vertices$ corresponds to \emph{linking} the two lower bounds representing the label at~$\vertexa$ and the arc~$(\vertexa,\vertexb)$, respectively.
We first \emph{convert} the consumption function~$\tcfunction_{(\vertexa,\vertexb)}$ mapping driving time to energy consumption to a piecewise linear function $\convexlowerboundfunction_{(\vertexa,\vertexb)}$ mapping \gls*{soc} to a lower bound on driving time (line~\ref{line:tfp-pot-convert} of Algorithm~\ref{alg:tradeoff-potential-backward-search}). Afterwards, the algorithm \emph{merges} the result with the existing label at $\vertexb$ to obtain a new convex function (line~\ref{line:tfp-pot-merge}).

\begin{algorithm}[t]
  \caption{Pseudocode of the function-propagating potential search for \gls*{tfp}}\label{alg:tradeoff-potential-backward-search}%
  \tcp{initialize labels}%
  \ForEach{$\vertex\in\vertices$}{\label{line:tfp-pot-initialize-begin}%
    $\convexlowerboundfunction_\vertex\assign\emptyset$\;%
  }%
  $\convexlowerboundfunction_\target\assign[(0,0)]$\;%
  $\queue$.\queueInsert{$\target,0$}\;\label{line:tfp-pot-initialize-end}%
  \BlankLine
  \tcp{run main loop}%
  \While{$\queue$.\queueIsNotEmpty{}}{%
    \BlankLine
    $\vertexa\assign\queue$.\queueDeleteMin{}\;%
    \BlankLine
    \tcp{scan outgoing arcs}%
    \ForEach{$(\vertexa,\vertexb)\in\backwardarcs$}{%
      $\convexlowerboundfunction_{(\vertexa,\vertexb)}\assign\convertop(\tcfunction_{(\vertexa,\vertexb)})$\;\label{line:tfp-pot-convert}%
      $\convexlowerboundfunction\assign\linkop(\convexlowerboundfunction_\vertexa,\convexlowerboundfunction_{(\vertexa,\vertexb)})$\;%
      \If{$\exists \atime\in\reals\colon\convexlowerboundfunction(\atime)<\convexlowerboundfunction_\vertexb(\atime)$}{%
        $\convexlowerboundfunction_\vertexb\assign\mergeop(\convexlowerboundfunction_\vertexb,\convexlowerboundfunction)$\;\label{line:tfp-pot-merge}%
        $\queue$.\queueUpdate{$\vertexb,\queuekey(\vertexb,\convexlowerboundfunction_\vertexb)$}\;%
      }%
    }%
  }%
\end{algorithm}

\subsection{Potential Function Consistency}\label{app:astar:lemma}

Below, Lemma~\ref{lem:tradeoff-convex-feasible-potential} formally proves that the potential function $\convexpotential$ is consistent when setting $\convexpotential(\vertex,\soc):=\convexlowerboundfunction_\vertex(\soc)$ for all $\vertex\in\vertices$ and $\soc\in[0,\maxbattery]$, using the piecewise linear functions computed as described above.

\setcounter{lemma}{0}
\renewcommand{\thelemma}{\Alph{section}-\arabic{lemma}}

\begin{lemma}
 \label{lem:tradeoff-convex-feasible-potential}
 The potential function $\convexpotential$ is consistent.
\end{lemma}

\begin{proof}
For some arc $(\vertexa,\vertexb)\in\arcs$, let $\convexlowerboundfunction_\vertexa$ and $\convexlowerboundfunction_\vertexb$ be the piecewise linear functions at $\vertexa$ and~$\vertexb$, respectively, after the backward search has terminated.
Let $\convexlowerboundfunction_{(\vertexa,\vertexb)}$ denote the lower bound function of~$(\vertexa,\vertexb)$.
We know that $\convexlowerboundfunction_\vertexa$ is upper bounded by the function computed by~$\linkop(\convexlowerboundfunction_\vertexb,\convexlowerboundfunction_{(\vertexa,\vertexb)})$, since this function was merged into~$\convexlowerboundfunction_\vertexa$ at some point during the backward search.
For an arbitrary driving time~$\atime\in\posreals$, let~$\soc^*:=\tcfunction_{(\vertexa,\vertexb)}(\atime)$ denote the corresponding energy consumption on the arc~$(\vertexa,\vertexb)$. Then, we obtain for all $\soc\in[0,\maxbattery]$ that
 \begin{align*}
  \atime+\convexpotential(\vertexb,\socprofile_{(\vertexa,\vertexb)}(\atime,\soc))&=\atime+\convexlowerboundfunction_\vertexb(\socprofile_{(\vertexa,\vertexb)}(\atime,\soc))\\
  &\ge\atime+\convexlowerboundfunction_\vertexb(\soc-\tcfunction_{(\vertexa,\vertexb)}(\atime))\\
  &=\tcfunction_{(\vertexa,\vertexb)}^{-1}(\soc^*)+\convexlowerboundfunction_\vertexb(\soc-\soc^*)\\
  &\ge\convexlowerboundfunction_{(\vertexa,\vertexb)}(\soc^*)+\convexlowerboundfunction_\vertexb(\soc-\soc^*)\\
  &\ge\min_{\soc^*\in\reals}\convexlowerboundfunction_{(\vertexa,\vertexb)}(\soc^*)+\convexlowerboundfunction_\vertexa(\soc-\soc^*)\\
  &\ge\convexpotential(\vertexa, \soc)\text{,}
 \end{align*}
which implies that the term $\atime-\convexpotential(\vertexa, \soc)+\convexpotential(\vertexb,\socprofile_{(\vertexa,\vertexb)}(\atime,\soc))$ is nonnegative. This completes the proof.
\end{proof}

\section{Omitted Details and Proofs from Section~\ref{sec:ch}}\label{app:ch}

The following Appendix~\ref{app:ch:lemmas} provides proofs from Section~\ref{sec:ch} that were omitted from the main manuscript. Afterwards, Appendix~\ref{app:ch:upper-bounds} describes how upper bound functions represented by a single tradeoff subfunction are derived for the witness search presented in Section~\ref{sec:ch:witness-search}.

\subsection{Omitted Proofs from Section~\ref{sec:ch:simple-soc-functions} and Section~\ref{sec:ch:shortcut-comparison}}\label{app:ch:lemmas}

Lemma~\ref{lem:tradeoff-nonnegative-shortcut} formally proves that the bivariate \gls*{soc} function of a path consisting solely of nonnegative consumption functions can be represented by a single univariate (consumption) function.

\setcounter{lemma}{3}
\renewcommand{\thelemma}{\arabic{lemma}}

\begin{lemma}
 Let $\apath = [\vertex_1, \dots, \vertex_k]$ be a path in~$\graph$ and let $\tcfunction_i$ denote the consumption function of the arc~$(\vertex_i,\vertex_{i+1})$ for~$i\in\{1,\dots,k-1\}$. If all consumption functions are nonnegative, \ie, $\tcfunction_i(\atime)\ge0$ holds for all $\atime\in\posreals$ and~$i\in\{1,\dots,k-1\}$, the \gls*{soc} function of~$\apath$ evaluates to
 \begin{align*}
  \socprofile(\atime, \soc)=
  \begin{cases}
   -\infty     & \mbox{if } \soc < \tcfunction(\atime), \\
   \soc - \tcfunction(\atime) & \mbox{otherwise,}
  \end{cases}
 \end{align*}
 where $\tcfunction$ denotes the function obtained after iteratively linking the functions $\tcfunction_1,\dots,\tcfunction_{k-1}$.
\end{lemma}

\begin{proof}
More formally, let~$\tcfunction:=\linkop(\dots\linkop(\dots\linkop(\tcfunction_1,\tcfunction_2),\dots),\tcfunction_{k-1})$.
First, consider the \gls*{soc} function $\socprofile_i$ of the consumption function~$\tcfunction_i$ of an individual arc $(\vertex_i,\vertex_{i+1})$ with $i\in\{1,\dots,k-1\}$. It equals
\begin{align*}
 \socprofile_i(\atime, \soc)=
\begin{cases}
 -\infty     & \mbox{if } \soc < \tcfunction_i(\atime), \\
 \soc - \tcfunction_i(\atime) & \mbox{otherwise,}
\end{cases}
\end{align*}
since $\tcfunction_i$ has only nonnegative values by assumption, so the \gls*{soc} at $\vertex_{i+1}$ never exceeds~$\maxbattery$.
We prove the lemma by induction.
Assume that for some $i\in\{1,\dots,k-2\}$, we are given the result $\tcfunction_{1,\dots,i}:=\linkop(\dots\linkop(\dots\linkop(\tcfunction_1,\tcfunction_2),\dots),\tcfunction_i)$ of linking all arcs in $[\vertex_1,\dots,\vertex_{i+1}]$ and the corresponding \gls*{soc} function is
\begin{align*}
 \socprofile_{1,\dots,i}(\atime,\soc)=
 \begin{cases}
  -\infty     & \mbox{if } \soc < \tcfunction_{1,\dots,i}(\atime), \\
  \soc - \tcfunction_{1,\dots,i}(\atime) & \mbox{otherwise.}
 \end{cases}
\end{align*}
We construct the \gls*{soc} function $\socprofile_{1,\dots,i+1}$ using~$\tcfunction_{1,\dots,i+1}:=\linkop(\tcfunction_{1,\dots,i},\tcfunction_{i+1})$. Since $\tcfunction_{i+1}(\atime)$ is nonnegative for all~$\atime\in\posreals$, we obtain $\tcfunction_{1,\dots,i}(\atime)\le\tcfunction_{1,\dots,i+1}(\atime)$ for arbitrary driving times~$\atime\in\posreals$. Hence, the path is infeasible for a pair of driving time $\atime\in\posreals$ and initial \gls*{soc} $\soc\in[0,\maxbattery]$ if and only if $\soc<\max\{\tcfunction_{1,\dots,i}(\atime),\tcfunction_{1,\dots,i+1}(\atime)\}=\tcfunction_{1,\dots,i+1}(\atime)$. Otherwise, the function $\tcfunction_{1,\dots,i+1}$ minimizes consumption on the path by definition of the link operation (still, no recuperation is possible). We obtain the function
\begin{align*}
 \socprofile_{1,\dots,i+1}(\atime, \soc)=
\begin{cases}
 -\infty     & \mbox{if } \soc < \tcfunction_{1,\dots,i+1}(\atime), \\
 \soc - \tcfunction_{1,\dots,i+1}(\atime) & \mbox{otherwise,}
\end{cases}
\end{align*}
which completes the proof.
\end{proof}

Lemma~\ref{lem:tradeoff-profile-dominance}, restated below, is used in Section~\ref{sec:ch:shortcut-comparison} to derive a simple method for comparing bivariate \gls*{soc} function when inserting new shortcuts. It is proven formally here.

\begin{lemma}
 Given a nonpositive or discharging \gls*{soc}~function $\socprofile_1$ and a discharging \gls*{soc}~function $\socprofile_2$, such that their respective consumption functions are~$\tcfunction_1^+$,~$\tcfunction_1^-$,~$\tcfunction_2^+$, and~$\tcfunction_2^-$, let the value $\varepsilon\ge0$ be defined as described above.
 If $\tcfunction^+_1(\atime^+)+\varepsilon\le\tcfunction^+_2(\atime^+)$ holds for some~$\atime^+\in\posreals$, any solution where $\atime^+$ is the (optimal) amount of time spent for $\tcfunction^+_2$ is dominated by $\socprofile_1$, \ie, we obtain either $\tcfunction_2^+(\atime^+)=\infty$ or $\socprofile_1(\atime^++\atime^-,\soc)\ge\socprofile_2(\atime^++\atime^-,\soc)=\soc-(\tcfunction_2^+(\atime^+)+\tcfunction_2^-(\atime^-))$ for all $\atime^-\in\posreals$ and~$\soc\in[0,\maxbattery]$.
\end{lemma}

\begin{proof}
 Assume for the sake of contradiction that there exists some value $\atime^+\in\posreals$ such that~$\tcfunction^+_1(\atime^+)+\varepsilon\le\tcfunction^+_2(\atime^+)$ and for some time $\atime^-\in\posreals$ and \gls*{soc} $\soc\in[0,\maxbattery]$, the value $\soc-(\tcfunction^+_2(\atime^+)+\tcfunction^-_2(\atime^-))$ is a feasible solution that is not dominated by $\socprofile_1(\atime^++\atime^-,\soc)$.
 Since~$\varepsilon\ge0$, we know that $\tcfunction^+_1(\atime^+) \le \tcfunction^+_2(\atime^+)$ holds. This implies that $\tcfunction^+_1(\atime^+) + \tcfunction^-_1(\atime^-)$ is a feasible solution for an \gls*{soc} of~$\soc$ (recall that the minimum driving time of $\tcfunction^-_1$ is~$0$). Finally, we know that $\tcfunction^+_1(\atime^+)\le\tcfunction^+_2(\atime^+)-\varepsilon$ holds by assumption and $\tcfunction^-_1(\atime^-)\le\tcfunction^-_2(\atime^-)+\varepsilon$ holds by the definition of $\varepsilon$. This yields
 \begin{align*}
  \socprofile_1(\atime^++\atime^-,\soc)&\ge\soc-\big(\tcfunction^+_1(\atime^+)+\tcfunction^-_1(\atime^-)\big)\\
  &\ge\soc-\big(\tcfunction^+_2(\atime^+)-\varepsilon+\tcfunction^-_2(\atime^-)+\varepsilon\big)\text{,}
 \end{align*}
 which contradicts our assumption and completes the proof.
\end{proof}

\subsection{Upper Bound Functions Described by a Constant Number of Coefficients}\label{app:ch:upper-bounds}

In Section~\ref{sec:ch:witness-search}, we mentioned that our witness search can utilize upper bounds represented by single tradeoff subfunctions to enable faster operations and better data locality. To derive such bounds, consider a piecewise-defined consumption function $\tcfunction$ with minimum and maximum driving time $\leftintervalborder\in\strictposreals$ and~$\rightintervalborder\in\strictposreals$, respectively, that is defined by several tradeoff subfunctions~$\tcinner_1,\dots,\tcinner_k$. We seek a tradeoff function $\tcfunctionupperbound$ that has the general form $\tcfunctionupperbound(\atime)=\aparam/(\atime-\bparam)^2+\cparam$ as in Equation~\eqref{eq:tradeoff-functioninnerform} for all $\atime\in[\leftintervalborder,\rightintervalborder]$ in the interval of admissible driving times. Further, we demand that $\tcfunctionupperbound(\atime)\ge\tcfunction(\atime)$ holds for all~$\atime\in[\leftintervalborder,\rightintervalborder]$.
To achieve this, we first set
\begin{align}
 \bparam &:= \min_{i\in\{1,\dots,k\}}\bparam_i\label{eq:upper-bound-parameters:beta}
\end{align}
where $\bparam_i$ denotes the coefficient of the tradeoff subfunction~$\tcinner_i$ of $\tcfunction$ (see Equation~\eqref{eq:tradeoff-functioninnerform}).
Then, we can fix the function values $\tcfunctionupperbound(\leftintervalborder):=\tcfunction(\leftintervalborder)$ and $\tcfunctionupperbound(\rightintervalborder):=\tcfunction(\rightintervalborder)$ at its domain borders to uniquely define the two remaining coefficients $\aparam$ and $\cparam$ from Equation~\eqref{eq:tradeoff-functioninnerform}. In particular, we obtain
\begin{align}
 \cparam &= \frac{\tcfunction(\rightintervalborder)(\bparam - \rightintervalborder)^2 - \tcfunction(\leftintervalborder)(\bparam - \leftintervalborder)^2}{(\bparam - \rightintervalborder)^2 - (\bparam - \leftintervalborder)^2},\label{eq:upper-bound-parameters:gamma}\\
 \aparam &= (\tcfunction(\leftintervalborder)-\cparam)(\leftintervalborder-\bparam)^2.\label{eq:upper-bound-parameters:alpha}
\end{align}
Lemma~\ref{lem:tradeoff-upper-bound} shows that the resulting function is in fact an upper bound on~$\tcfunction$.
Upper bounds are also \emph{robust} towards incremental linking in the sense that the error does not increase if we recompute the upper bound whenever linking several bound functions results in a bound consisting of multiple subfunctions. This is due to the fact that the bounds are uniquely defined by their (minimum) coefficient $\beta$ and the domain borders of the original functions, which remain unchanged in the upper bound.

During witness search, whenever linking two bound functions results in a function defined by more than one tradeoff subfunction, we compute and store the upper bound instead. Note that we do not even have to link functions explicitly, but simply compute the new coefficient $\bparam$ in a linear scan that simulates the link operation.

\setcounter{lemma}{0}
\renewcommand{\thelemma}{\Alph{section}-\arabic{lemma}}

\begin{lemma}
 \label{lem:tradeoff-upper-bound}
 The function $\tcfunctionupperbound$ defined by Equations~\eqref{eq:upper-bound-parameters:beta}--\eqref{eq:upper-bound-parameters:alpha} is an upper bound on the original consumption function $\tcfunction$ within the interval $[\leftintervalborder,\rightintervalborder]$, \ie, $\tcfunctionupperbound(\atime)\ge\tcfunction(\atime)$ holds for all~$\atime\in[\leftintervalborder,\rightintervalborder]$.
\end{lemma}

\begin{proof}
 Let $\tcinner_1,\dots,\tcinner_k$ denote the tradeoff subfunctions defining~$\tcfunction$ and without loss of generality, assume that these subfunctions are given in increasing order of their admissible driving times. First, we argue that it is sufficient to prove the lemma for the case $k=2$. To show this, we define an operation $\boundop\colon\functionspace\times\functionspace\to\functionspace$ that takes as input two consumption functions, each defined by a \emph{single} tradeoff subfunction, and computes an upper bound as described above. 
 For~$i\in\{1,\dots,k-1\}$, consider two consumption functions $\tcfunction_i$ and~$\tcfunction_{i+1}$ induced by two consecutive tradeoff functions $\tcinner_i$ and~$\tcinner_{i+1}$ (see Section~\ref{sec:operations:linking} for the definition of induced consumption functions). Let their corresponding minimum and maximum driving times be~$\leftintervalborder_i$,~$\rightintervalborder_i=\leftintervalborder_{i+1}$, and~$\rightintervalborder_{i+1}$.
 The $\boundop$ operation computes the consumption function $\tcfunction_{i,i+1}:=\boundop(\tcfunction_i,\tcfunction_{i+1})$ with minimum driving time~$\leftintervalborder_i$, maximum driving time~$\rightintervalborder_{i+1}$, and a single tradeoff subfunction~$\tcfunctionupperbound_{i,i+1}$. According to Equations~\eqref{eq:upper-bound-parameters:beta}--\eqref{eq:upper-bound-parameters:alpha}, the coefficients of this tradeoff function depend only on the values $\bparam=\min\{\bparam_i,\bparam_{i+1}\}$, the driving times $\leftintervalborder_i$ and~$\rightintervalborder_{i+1}$, as well as the consumption values $\tcfunction_i(\leftintervalborder_i)=\tcfunctionupperbound_{i,i+1}(\leftintervalborder_i)$ and $\tcfunction_{i+1}(\rightintervalborder_{i+1})=\tcfunctionupperbound_{i,i+1}(\rightintervalborder_{i+1})$ at the domain borders of~$\tcfunctionupperbound_{i,i+1}$.
 Linking the consumption function $\tcfunction_{i,i+1}$ with another consecutive (induced) consumption function yields a new function defined by the same corresponding values.
 Consequently, the result $\boundop(\dots\boundop(\dots\boundop(\tcfunction_1,\tcfunction_2),\dots),\tcfunction_k)$ of iteratively applying the $\boundop$ operation to the $k$ induced consumption functions of $\tcfunction$ is the function that is defined by the coefficient~$\bparam=\min_{i\in\{0,\dots,k\}}\bparam_i$ (see Equation~\eqref{eq:upper-bound-parameters:beta}), the minimum driving time~$\leftintervalborder=\leftintervalborder_1$, the maximum driving time~$\rightintervalborder=\rightintervalborder_k$, the maximum consumption~$\tcfunction(\leftintervalborder)=\tcfunction_1(\leftintervalborder_1)$, and the minimum consumption~$\tcfunction(\rightintervalborder)=\tcfunction_k(\rightintervalborder_k)$.
 This is exactly the function $\tcfunctionupperbound$ defined above.
 Thus, we can construct $\tcfunctionupperbound$ by iteratively applying the $\boundop$ operation to consumption functions induced by the tradeoff subfunctions of~$\tcfunction$. To prove the lemma, we now show that each function constructed by the operation $\boundop$ is in fact an upper bound on its \emph{two} input functions. Observe that this implies that $\tcfunctionupperbound$ is an upper bound on $\tcfunction$ within the interval~$[\leftintervalborder,\rightintervalborder]$.

 In the remainder of the proof, let $\tcfunction$ be a consumption function defined by two tradeoff subfunctions $\tcinner_1$ and~$\tcinner_2$, which induce two consumption functions $\tcfunction_1$ and~$\tcfunction_2$. We prove that the function $\tcfunctionupperbound\colon\strictposreals\to\reals$ computed by $\boundop(\tcfunction_1,\tcfunction_2)$ yields an upper bound on $\tcfunction$ on the interval~$[\leftintervalborder,\rightintervalborder]$.
 Let the subdomains of $\tcinner_1$ and $\tcinner_2$ be $[\leftintervalborder,\intervalborder)$ and~$[\intervalborder,\rightintervalborder)$, respectively. By continuity of~$\tcfunction$ on the interval~$[\leftintervalborder,\rightintervalborder]$ and by continuity of both $\tcinner_1$ and $\tcinner_2$ on~$\strictposreals$, we know that~$\tcinner_1(\intervalborder)=\tcinner_2(\intervalborder)$.
 To prove the lemma, we use the following three claims.
 \begin{enumerate}
  \item The inequality $\tcfunctionupperbound(\intervalborder)\ge\tcinner_1(\intervalborder)=\tcinner_2(\intervalborder)$ holds.
  \item The slopes (\ie, the derivatives) of $\tcfunctionupperbound$ and $\tcinner_1$ are equal at $\leftintervalborder$ if and only if $\tcfunctionupperbound\equiv\tcinner_1\equiv\tcinner_2$. Otherwise, the slope of $\tcfunctionupperbound$ is \emph{greater} at this point.
  \item The slopes of $\tcfunctionupperbound$ and $\tcinner_2$ are equal at $\rightintervalborder$ if and only if $\tcfunctionupperbound\equiv\tcinner_1\equiv\tcinner_2$. Otherwise, the slope of $\tcfunctionupperbound$ is \emph{smaller} at this point.
 \end{enumerate}
 Then, $\tcfunctionupperbound(\intervalborder)\ge\tcinner_1(\intervalborder)$ holds by our first claim, $\tcfunctionupperbound$ and $\tcinner_1$ intersect at $\leftintervalborder$ by construction, and $\tcfunctionupperbound(\leftintervalborder+\varepsilon)\ge\tcinner_1(\leftintervalborder+\varepsilon)$ holds for $\varepsilon>0$ in the neighborhood of $\leftintervalborder$ by our second claim.
 This implies that $\tcfunctionupperbound$ must be an upper bound on $\tcinner_1$ on the interval~$[\leftintervalborder,\intervalborder]$, because the functions $\tcfunctionupperbound$ and $\tcinner_1$ can intersect at most twice in this interval unless~$\tcfunctionupperbound\equiv\tcinner_1$. This is easy to verify by determining the number of zeros of $\tcfunctionupperbound-\tcinner_1$ within the considered interval~$[\leftintervalborder,\intervalborder]$.
 A similar argument holds for $\tcfunctionupperbound$ and $\tcinner_2$ on the interval~$[\intervalborder,\rightintervalborder]$. Hence, the lemma follows after proving the three claims made above. We detail the rather technical proofs of these claims below.

 Assume that the functions $\tcinner_1$ and $\tcinner_2$ are given as~$\tcinner_i(\atime)=\aparam_i/(\atime-\bparam_i)^2+\cparam_i$ for all $\atime\in\strictposreals$ and for~$i\in\{1,2\}$.
 For the sake of simplicity and without loss of generality, we presume that $\min\{\bparam_1,\bparam_2\}=0$. Note that we can always enforce this property by \emph{shifting} both functions (and their subdomains) along the x-axis. Afterwards, we obtain the same function~$\tcfunctionupperbound$ on the shifted subdomains. By a similar argument, we presume that $\cparam_1=0$ holds.
 Below, we consider the case $\bparam_1=0$ and~$\bparam_2\ge0$. The case $\bparam_1\ge0$ and $\bparam_2=0$ is analogous. Since $\cparam_2\in\reals$ is allowed to become negative, no further case distinction is necessary.

 To prove the first claim, we have to show that $\tcfunctionupperbound(\intervalborder)\ge\tcinner_1(\intervalborder)=\tcinner_2(\intervalborder)$ holds.
 For the sake of contradiction, assume $\tcfunctionupperbound(\intervalborder)<\tcinner_1(\intervalborder)$.
 As mentioned before, continuity of the consumption function~$\tcfunction$ on the interval~$[\leftintervalborder,\rightintervalborder]$ implies that $\tcinner_1(\intervalborder)=\tcinner_2(\intervalborder)$, \ie,
 \begin{align}
  \label{eq:tradeoff-bound-proof-intersection-vals}
  \frac{\aparam_1}{\intervalborder^2}=\frac{\aparam_2}{(\intervalborder-\bparam_2)^2}+\cparam_2\text{.}
 \end{align}
 Furthermore, we know that $\tcfunction$ is a \emph{convex} function (see Lemma~\ref{lem:tradeoff-convexity}), so when evaluating the derivatives of $\tcinner_1$ and $\tcinner_2$ at $\intervalborder$, we get the inequality
 \begin{align}
  \label{eq:tradeoff-bound-proof-intersection-slopes}
  -\frac{2\aparam_1}{\intervalborder^3}\le-\frac{2\aparam_2}{(\intervalborder-\bparam_2)^3}\text{.}
 \end{align}
 Finally, we know that the inequalities $0<\leftintervalborder<\intervalborder<\rightintervalborder$,~$\bparam_2\ge0$,~$\aparam_1>0$,~$\aparam_2>0$, and~$\intervalborder>\bparam_2$ hold by definition for consumption functions composed of multiple tradeoff subfunctions.
 We now show that altogether, these inequalities yield a contradiction.
 First, we plug the values of $\aparam\in\posreals$ from Equation~\eqref{eq:upper-bound-parameters:alpha} and $\cparam\in\reals$ from Equation~\eqref{eq:upper-bound-parameters:gamma} into the term~$\tcfunctionupperbound(\intervalborder)=\aparam/\intervalborder^2+\cparam$.
 Afterwards, we replace $\cparam_2=\aparam_1/\intervalborder^2-\aparam_2/(\intervalborder-\bparam_2)^2$ according to Equation~\eqref{eq:tradeoff-bound-proof-intersection-vals} and exploit that the inequality $\aparam_1\ge\aparam_2\intervalborder^3/(\intervalborder-\bparam_2)^3>0$ holds by Equation~\eqref{eq:tradeoff-bound-proof-intersection-slopes} to obtain
 \begin{align*}
  && \tcfunctionupperbound(\intervalborder)=\frac{\tcinner_1(\leftintervalborder)\leftintervalborder^2(\rightintervalborder^2-\intervalborder^2)+\tcinner_2(\rightintervalborder)\rightintervalborder^2(\intervalborder^2-\leftintervalborder^2)}{\intervalborder^2(\rightintervalborder^2-\leftintervalborder^2)}&<\tcinner_1(\intervalborder)\\
  \Leftrightarrow && \tcinner_1(\leftintervalborder)\leftintervalborder^2(\rightintervalborder^2-\intervalborder^2)-\tcinner_1(\intervalborder)\intervalborder^2(\rightintervalborder^2-\leftintervalborder^2)+\tcinner_2(\rightintervalborder)\rightintervalborder^2(\intervalborder^2-\leftintervalborder^2)&<0\\
  \Leftrightarrow && (\intervalborder^2-\leftintervalborder^2)\Big(\aparam_1(\rightintervalborder^2-\intervalborder^2)(\rightintervalborder-\bparam_2)^2(\intervalborder-\bparam_2)^2+\aparam_2\intervalborder^2\rightintervalborder^2\big((\intervalborder-\bparam_2)^2-(\rightintervalborder-\bparam_2)^2\big)\Big)&<0\\
  \Rightarrow && \frac{\aparam_2\intervalborder^2}{\intervalborder-\bparam_2}\Big(\intervalborder(\rightintervalborder^2-\intervalborder^2)(\rightintervalborder-\bparam_2)^2+\rightintervalborder^2(\intervalborder-\bparam_2)\big((\intervalborder-\bparam_2)^2-(\rightintervalborder-\bparam_2)^2\big)\Big)&<0\\
  \Leftrightarrow && \bparam_2(\rightintervalborder^2-\intervalborder^2)(2\rightintervalborder\intervalborder-2\rightintervalborder\bparam_2+\rightintervalborder^2-\intervalborder\bparam_2)&<0\\
  \Leftrightarrow && 2\rightintervalborder\intervalborder-2\rightintervalborder\bparam_2+\rightintervalborder^2-\intervalborder\bparam_2\hphantom{)}&<0\text{.}
 \end{align*}
 This yields a contradiction, because we know that $0\le\bparam_2<\intervalborder<\rightintervalborder$ holds. Thus, both $2\rightintervalborder\intervalborder-2\rightintervalborder\bparam_2$ and $\rightintervalborder^2-\intervalborder\bparam_2$ are positive terms and their sum cannot be negative.

 For the second claim, we examine the slopes of $\tcinner_1$ and $\tcfunctionupperbound$ at the domain border $\leftintervalborder$.
 Let the parameter $\aparam\in\posreals$ of $\tcfunctionupperbound$ be defined as in Equation~\eqref{eq:upper-bound-parameters:alpha}. Plugging in the coefficient $\bparam=0$ and the value of $\cparam\in\reals$ according to Equation~\eqref{eq:upper-bound-parameters:gamma}, we obtain
 \begin{align*}
  \aparam=\frac{\big(\tcinner_1(\leftintervalborder)-\tcinner_2(\rightintervalborder)\big)\leftintervalborder^2\rightintervalborder^2}{\rightintervalborder^2-\leftintervalborder^2}\text{.}
 \end{align*}
 As before, we use $\cparam_2=\aparam_1/\intervalborder^2-\aparam_2/(\intervalborder-\bparam_2)^2$ and the inequality~$\aparam_1\ge\aparam_2\intervalborder^3/(\intervalborder-\bparam_2)^3>0$. For the difference between the derivatives $\tcfunctionupperbound'$ and $\tcinner'_1$ at~$\leftintervalborder$, this yields
 \begin{align*}
  \tcfunctionupperbound'(\leftintervalborder)-\tcinner'_1(\leftintervalborder)&=\frac{2\aparam_1}{\leftintervalborder^3}-\frac{2\aparam}{\leftintervalborder^3}\\
  &=\frac{2\left(\tcinner_2(\rightintervalborder)\rightintervalborder^2\leftintervalborder^2-\aparam_1\leftintervalborder^2\right)}{\leftintervalborder^3\left(\rightintervalborder^2-\leftintervalborder^2\right)}\\
  &=\frac{2\leftintervalborder^2}{\leftintervalborder^3\left(\rightintervalborder^2-\leftintervalborder^2\right)}\left(\aparam_2\rightintervalborder^2\frac{(\intervalborder-\bparam_2)^2-(\rightintervalborder-\bparam_2)^2}{(\rightintervalborder-\bparam_2)^2(\intervalborder-\bparam_2)^2}+\aparam_1\left(\frac{\rightintervalborder^2}{\intervalborder^2}-1\right)\right)\\
  &\ge\frac{2\aparam_2\leftintervalborder^2}{\leftintervalborder^3(\rightintervalborder^2-\leftintervalborder^2)}\left(\rightintervalborder^2\frac{(\intervalborder-\bparam_2)^2-(\rightintervalborder-\bparam_2)^2}{(\rightintervalborder-\bparam_2)^2(\intervalborder-\bparam_2)^2}+\frac{\intervalborder^3}{(\intervalborder-\bparam_2)^3}\left(\frac{\rightintervalborder^2}{\intervalborder^2}-1\right)\right)\\
  &=\frac{2\aparam_2\bparam_2\leftintervalborder^2(\rightintervalborder-\intervalborder)^2(2\rightintervalborder\intervalborder-2\rightintervalborder\bparam_2+\rightintervalborder^2-\intervalborder\bparam_2)}{\rightintervalborder^3(\rightintervalborder^2-\leftintervalborder^2)(\rightintervalborder-\bparam_2)^2(\intervalborder-\bparam_2)^3}\text{.}
 \end{align*}
 As in the proof of the first claim, we observe that each term in the product of the numerator is nonnegative, while each term in the product of the denominator is positive. Moreover, the numerator is equal to $0$ if and only if $\beta_2=0$ holds. Using Equation~\eqref{eq:tradeoff-bound-proof-intersection-vals}, it is easy to verify that this implies $\aparam_1=\aparam_2$ and~$\cparam_2=0$, which corresponds to the case where the three functions $\tcfunctionupperbound$,~$\tcinner_1$, and~$\tcinner_2$ are equivalent.

 Finally, we deal with the slopes of $\tcinner_2$ and $\tcfunctionupperbound$ at $\rightintervalborder$ to prove the third claim. Below, we first replace the values~$\aparam$, $\aparam_2$, and~$\cparam_2$ as in our proof of the second claim. Afterwards, we exploit the fact that $(\intervalborder^3-\atime)(\rightintervalborder-\bparam_2)^3-(\rightintervalborder^3-\atime)(\intervalborder-\bparam_2)^3$ decreases with increasing~$\atime\in\posreals$, since its derivative \wrt $\atime$ is $(\intervalborder-\bparam_2)^3-(\rightintervalborder-\bparam_2)^3<0$.
 After some further rearrangements, we obtain
 \begin{align*}
  \tcinner'_2(\rightintervalborder)-\tcfunctionupperbound'(\rightintervalborder)&=\frac{2\aparam}{\rightintervalborder^3}-\frac{2\aparam_2}{(\rightintervalborder-\bparam_2)^3}\\
  &\ge\frac{2\aparam_2\big((\intervalborder^3-\leftintervalborder^2\bparam_2)(\rightintervalborder-\bparam_2)^3-(\rightintervalborder^3-\leftintervalborder^2\bparam_2)(\intervalborder-\bparam_2)^3\big)}{\rightintervalborder(\rightintervalborder^2-\leftintervalborder^2)(\rightintervalborder-\bparam_2)^3(\intervalborder-\bparam_2)^3}\\
  &\ge\frac{2\aparam_2\big((\intervalborder^3-\intervalborder^2\bparam_2)(\rightintervalborder-\bparam_2)^3-(\rightintervalborder^3-\intervalborder^2\bparam_2)(\intervalborder-\bparam_2)^3\big)}{\rightintervalborder(\rightintervalborder^2-\leftintervalborder^2)(\rightintervalborder-\bparam_2)^3(\intervalborder-\bparam_2)^3}\\
  &=\frac{2\aparam_2\bparam_2(\rightintervalborder-\intervalborder)^2(\rightintervalborder\intervalborder-\rightintervalborder\bparam_2+\rightintervalborder\intervalborder-\intervalborder\bparam_2+\intervalborder^2-\intervalborder\bparam_2)}{\rightintervalborder(\rightintervalborder^2-\leftintervalborder^2)(\rightintervalborder-\bparam_2)^3(\intervalborder-\bparam_2)^2}\text{.}
 \end{align*}
 Again, we end up with products for which all factors are nonnegative (and strictly positive in case of the denominator). As before, the numerator equals $0$ if and only if~$\tcfunctionupperbound\equiv\tcinner_1\equiv\tcinner_2$.
 Hence, all three claims hold and the proof is complete.
\end{proof}

\section{Implementation Details}\label{app:implementation}

Our implementation stores graphs as adjacency arrays~\cite{Cor01}, following the dynamic data structures of Delling~\cite{Del09d} for efficient insertion and deletion of shortcuts during the preprocessing routine of~\gls*{ch}.
All algorithms use $k$-heaps~\cite{Cor01,Joh75} as priority queue, where $k=4$ except for unsettled label sets, which use~$k=2$. Compared to Fibonacci heaps~\cite{Fre87}, these heaps have a higher worst-case complexity, but are faster on sparse graphs (such as road networks) in practice~\cite{Che96}.
Implementation details of our speedup techniques are given below.

\paragraph{A* Search.}
When computing the potential function~$\convexpotential$ (see Section~\ref{sec:astar:pwl-functions}), the number of breakpoints of lower bounds can become quite large.
Therefore, we reduce it as follows (while slightly deteriorating the quality of the bounds).
Before applying Graham's scan, we replace consecutive pairs of breakpoints in the piecewise linear function by a single one if they are close to each other, \ie, their difference \wrt driving time or \gls*{soc} is below a certain threshold $\Delta_\atime\in\posreals$ or~$\Delta_\soc\in\posreals$, respectively. Two such points $\pointa=(\soc_\pointa,\atime_\pointa)$ and $\pointb=(\soc_\pointb,\atime_\pointb)$ are replaced by~$\pointc:=(\min\{\soc_\pointa,\soc_\pointb\},\min\{\atime_\pointa,\atime_\pointb\})$.
Furthermore, if two consecutive segments $\pointa\pointb$ and $\pointb\pointc$ with $\pointa=(\soc_\pointa,\atime_\pointa)$,~$\pointb=(\soc_\pointb,\atime_\pointb)$, and~$\pointc=(\soc_\pointc,\atime_\pointc)$ have similar slopes $\slope_{\pointa\pointb}\approx\slope_{\pointb\pointc}$ (\ie, the difference $|\slope_{\pointa\pointb}-\slope_{\pointb\pointc}|$ is below some threshold $\Delta_\slope\in\posreals$), we replace them by a single segment from $(\soc_\pointa,\atime_\pointa)$ to $(\soc^*,\atime_\pointc)$ with slope $\min\{\slope_{\pointa\pointb},\slope_{\pointb\pointc}\}$, which uniquely defines the value $\soc^*\in\reals$. Clearly, the modified function remains a lower bound. Moreover, consistency of the potential is maintained, as function values can only decrease and changes in the function are propagated by the search. Thus, all steps in the proof of Lemma~\ref{lem:tradeoff-convex-feasible-potential} still apply.

Graham's scan and the breakpoint reduction step are performed on-the-fly during the merge operation. Moreover, we convert consumption functions $\tcfunction_\arc$ of all arcs $\arc\in\arcs$ to their corresponding lower bounds $\convexlowerboundfunction_\arc$ during preprocessing for faster query times.
The thresholds $\varepsilon$ to determine lower bound errors, $\Delta_\atime$ and $\Delta_\soc$ for close points, and $\Delta_\slope$ for similar slopes are tuning parameters.
Smaller thresholds increase accuracy of bounds, but also slow down the backward search. Therefore, we set above thresholds to $2^{\delta-\lfloor\log\maxbattery\rfloor}$ in our experiments, where $\delta\in\naturals$ is a constant and $\maxbattery$ is the battery capacity (assumed to be given in kWh). Hence, bounds are more accurate for higher capacities (where the forward search becomes more expensive). The value of $\delta$ is again a tuning parameter.
In our experiments, we use $\delta=10$ for~$\Delta_\atime$ (the resulting threshold is measured in seconds), $\delta=17$ for~$\Delta_\soc$, and $\delta=15$ for~$\varepsilon$ (both measured in~Wh).
For example, a battery capacity of 16\,kWh yields $\Delta_\atime=64$~(seconds),~$\Delta_\soc=2^{13}$~(Wh), and~$\varepsilon=2^{11}$~(Wh).
The value $\Delta_\slope=2^{-4}$ is constant and chosen independently of~$\maxbattery$ (all parameters were determined in preliminary experiments).

\paragraph{Contraction Hierarchies.}
During preprocessing, we determine the next vertex to be contracted using the measures \gls*{ed} and \gls*{cq} according to Geisberger et~al.~\cite{Gei12b}.
To reflect the complexity of \gls*{soc} functions, we add another term \emph{\gls*{sc}}, which is defined as $|\tcfunction^+|+k|\tcfunction^-|$ for the \gls*{soc} function of a given shortcut candidate, where $|\tcfunction^+|$ and $|\tcfunction^-|$ denote the number of tradeoff subfunctions that define the positive and negative part of a shortcut, respectively, and~$k\in\naturals$ is a tuning parameter. Using penalized weights for negative parts, we favor earlier contraction of \gls*{soc} functions without a negative part (we use $k=4$ in our experiments).
The priority of a vertex (higher priority means higher importance) is then set to~$64 \operatorname{\gls*{ed}}+\operatorname{\gls*{cq}}+\operatorname{\gls*{sc}}$.
We set the priority of all inactive vertices to~$\infty$.

Further, we employ a \emph{settled node limit}~\cite{Gei12b} of~$128$, which limits the maximum number of queue extractions per witness search for better performance.
If multiple shortcut candidates with the same tail vertex $\vertexa\in\vertices$ are constructed during contraction of a vertex, we save time by running only a \emph{single} multi-target witness search from~$\vertexa$.
Finally, to improve performance of the backward searches during a query (\gls*{bfs} and potential computation), we explicitly construct and store their more lightweight search graphs from the input graph (enriched with shortcuts, but storing less complex cost functions) during preprocessing.

\section{Omitted Plots from Section~\ref{sec:experiments}}\label{app:experiments}

\begin{figure}[t]
 \centering
 \begin{subfigure}[b]{.5\textwidth}%
 \centering%
 \tikzstyle{markSign} = [mark=*]
\tikzstyle{shortenLines} = [shorten <= 3.5pt,shorten >= 3.5pt]

\begin{tikzpicture}[figure]
\pgfplotsset{
    grid style = {dash pattern = on 1pt off 1pt, black15,line width = 0.5pt  }
 }
\pgfplotsset{ every non boxed x axis/.append style={x axis line style=-},
     every non boxed y axis/.append style={y axis line style=-},
     legend image post style={line width=1.5pt}}

\colorlet{plotColor1}{thesisblue}
\colorlet{plotColor2}{thesisred}
     
\begin{axis}[
   height=5cm,
   width=0.95\textwidth,
   xmin=1.75,
   xmax=8.25,
   ymin=1,
   ymax=5000000,
   ymode=log,
   ytick pos=left,
   ytick={10,100,1000,10000,100000,1000000},
   xlabel={Step Size [km/h]},
   ylabel={Time [ms]},
   xtick={1, 2, 3, 4, 5, 6, 7, 8},
   xticklabels={0, 1, 2, 5, 10, 20, $\infty$, cnst.},
   grid=major,
   legend entries={avg. time, max. time},
   legend cell align=left,
   legend style={at={(0.97,0.95)},
   anchor=north east,
   font=\scriptsize}
]

\addlegendimage{legend line with exact mark,plotColor1}
\addlegendimage{legend line with exact mark,plotColor2}

\addplot [color=thesisblue-light,line width=1pt] table {
1.75 6.234645
8.25 6.234645
};
\addplot [color=thesisred-light,line width=1pt] table {
1.75 22.849100
8.25 22.849100
};

\addplot [color=plotColor2,markSign,shortenLines,line width=1pt] table {
   2 3498280.000000
   3 1022610.000000
};
\addplot [color=plotColor2,markSign,shortenLines,line width=1pt] table {
   3 1022610.000000
   4 89553.900000
};
\addplot [color=plotColor2,markSign,shortenLines,line width=1pt] table {
   4 89553.900000
   5 17001.400000
};
\addplot [color=plotColor2,markSign,shortenLines,line width=1pt] table {
   5 17001.400000
   6 4770.600000
};
\addplot [color=plotColor2,markSign,shortenLines,line width=1pt] table {
   6 4770.600000
   7 476.579000
};
\addplot [color=plotColor2,markSign,shortenLines,line width=1pt] table {
   7 476.579000
   8 18.915000
};

\addplot [color=plotColor1,markSign,shortenLines,line width=1pt] table {
   2 93407.18589
   3 24573.46175
};
\addplot [color=plotColor1,markSign,shortenLines,line width=1pt] table {
   3 24573.46175
   4 2848.703090
};
\addplot [color=plotColor1,markSign,shortenLines,line width=1pt] table {
   4 2848.703090
   5 595.527929
};
\addplot [color=plotColor1,markSign,shortenLines,line width=1pt] table {
   5 595.527929
   6 203.225958
};
\addplot [color=plotColor1,markSign,shortenLines,line width=1pt] table {
   6 203.225958
   7 26.731756
};
\addplot [color=plotColor1,markSign,shortenLines,line width=1pt] table {
   7 26.731756
   8 4.374692
};

\end{axis}
\end{tikzpicture}
 \caption{}%
 \label{fig:step-plot:time}%
 \end{subfigure}%
 \begin{subfigure}[b]{.5\textwidth}%
 \centering%
 \tikzstyle{markSign} = [mark=*]
\tikzstyle{shortenLines} = [shorten <= 3.5pt,shorten >= 3.5pt]

\begin{tikzpicture}[figure]
\pgfplotsset{
    grid style = {dash pattern = on 1pt off 1pt, black15,line width = 0.5pt  }
 }
\pgfplotsset{ every non boxed x axis/.append style={x axis line style=-},
     every non boxed y axis/.append style={y axis line style=-},
     legend image post style={line width=1.5pt}}

\colorlet{plotColor1}{thesisgreen}
\colorlet{plotColor2}{thesisyellow}
     
\begin{axis}[
   height=5cm,
   width=0.95\textwidth,
   ymode=log,
   ytick pos=right,
   ymin=0.00001,
   ymax=50.0,
   /pgf/number format/.cd,
   use comma,
   xmin=1.75,
   xmax=8.25,
   xlabel={Step Size [km/h]},
   ylabel={Quality Loss [\%]},
   ytick={0.0001,0.001,0.01,0.1,1,10},
   xtick={1, 2, 3, 4, 5, 6, 7, 8},
   xticklabels={0, 1, 2, 5, 10, 20, $\infty$, cnst.},
   grid=major,
   legend entries={avg. loss, max. loss},
   legend cell align=left,
   legend style={at={(0.97,0.05)},
   anchor=south east,
   font=\scriptsize}
]

\addlegendimage{legend line with exact mark,plotColor1}
\addlegendimage{legend line with exact mark,plotColor2}

\addplot [color=plotColor1,markSign,shortenLines,line width=1pt] table {
   2 0.0001
   3 0.0020
};
\addplot [color=plotColor1,markSign,shortenLines,line width=1pt] table {
   3 0.0020
   4 0.0136
};
\addplot [color=plotColor1,markSign,shortenLines,line width=1pt] table {
   4 0.0136
   5 0.0377
};
\addplot [color=plotColor1,markSign,shortenLines,line width=1pt] table {
   5 0.0377
   6 0.1481
};
\addplot [color=plotColor1,markSign,shortenLines,line width=1pt] table {
   6 0.1481
   7 0.6657
};
\addplot [color=plotColor1,markSign,shortenLines,line width=1pt] table {
   7 0.6657
   8 2.2592
};

\addplot [color=plotColor2,markSign,shortenLines,line width=1pt] table {
   2 0.0051
   3 0.0215
};
\addplot [color=plotColor2,markSign,shortenLines,line width=1pt] table {
   3 0.0215
   4 0.0872
};
\addplot [color=plotColor2,markSign,shortenLines,line width=1pt] table {
   4 0.0872
   5 0.2332
};
\addplot [color=plotColor2,markSign,shortenLines,line width=1pt] table {
   5 0.2332
   6 2.2579
};
\addplot [color=plotColor2,markSign,shortenLines,line width=1pt] table {
   6 2.2579
   7 9.1806
};
\addplot [color=plotColor2,markSign,shortenLines,line width=1pt] table {
   7 9.1806
   8 22.5135
};

\end{axis}
\end{tikzpicture}
 \caption{}%
 \label{fig:step-plot:quality}%
 \end{subfigure}%
 \caption{Trading accuracy for running time when using \gls*{bsp} (\instanceGerNoAux, 2\,kWh). Both plots show query times and result quality for the same set of queries as in Table~\ref{tbl:speed-steps}. (a)~Average and maximum query times for different speed steps. Horizontal lines correspond to the average and maximum query times when running~\gls*{tfp}. (b)~Average and maximum loss in quality compared to the result of \gls*{tfp}, \ie, continuous tradeoffs.}
 \label{fig:step-plot}
\end{figure}
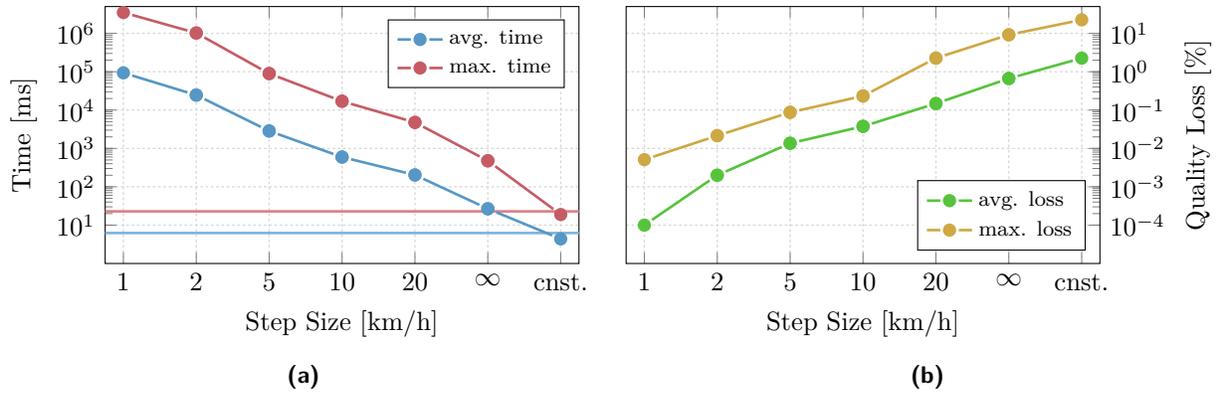

Figure~\ref{fig:step-plot} plots query times against solution qualities for different speed step sizes (visualizing the same figures as in Table~\ref{tbl:speed-steps}).
Figure~\ref{fig:champ-limitplot-avg} shows average running times for the same queries as in Figure~\ref{fig:champ-limitplot-med}, comparing the scalability of different approaches. Recall that a run was aborted if at least one of the 100 queries exceeded an hour of computation time.

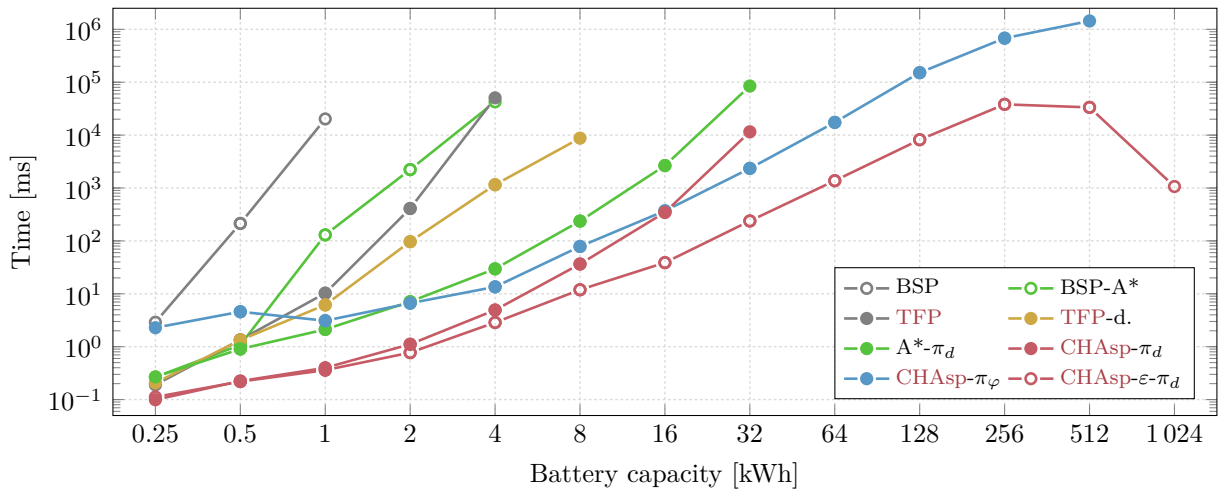
\begin{figure}[t]
  \centering
  \begin{tikzpicture}[figure]
\pgfplotsset{
   grid style = {dash pattern = on 1pt off 1pt, black15,line width = 0.5pt  }
}

\colorlet{plotColor1}{black50} 
\colorlet{plotColor2}{thesisgreen} 
\colorlet{plotColor3}{black50} 
\colorlet{plotColor4}{thesisyellow} 
\colorlet{plotColor5}{thesisgreen} 
\colorlet{plotColor7}{thesisblue} 
\colorlet{plotColor8}{thesisred} 
\colorlet{plotColor9}{thesisred} 

\begin{axis}[
   height=7.0cm,
   width=0.98\textwidth,
   xmin=0.5,
   xmax=13.5,
   ymin=0.05,
   ymax=2500000,
   ymode=log,
   xlabel={Battery capacity [kWh]},
   ylabel={Time [ms]},
   y label style={at={(axis description cs:-0.06,0.5)}},
   /pgf/number format/.cd,
   1000 sep={\,},
   xtick={0, 1, 2, 3, 4, 5, 6, 7, 8, 9, 10, 11, 12, 13},
   xticklabel=\pgfmathparse{2^(\tick-3)}${\pgfmathprintnumber{\pgfmathresult}}$,
   grid=major,
   legend entries={BSP, BSP-A*, \gls{tfp}, \gls{tfp}-d., A*-$\drivingtimepotential$, \gls{champ}-$\drivingtimepotential$, \gls{champ}-$\convexpotential$, \gls{champ}-$\varepsilon$-$\drivingtimepotential$},
   legend cell align=left,
   legend columns=2,
   legend style={at={(0.98,0.04)},
   anchor=south east,
   font=\scriptsize}
]

 \addlegendimage{legend line with heuristic mark,plotColor1}
 \addlegendimage{legend line with heuristic mark,plotColor2}
 \addlegendimage{legend line with exact mark,plotColor3}
 \addlegendimage{legend line with exact mark,plotColor4}
 \addlegendimage{legend line with exact mark,plotColor5}
 \addlegendimage{legend line with exact mark,plotColor8}
 \addlegendimage{legend line with exact mark,plotColor7}
 \addlegendimage{legend line with heuristic mark,plotColor9}

\addplot[color=plotColor1,heuristicMarkSign,shortenLines,line width=1pt] table {
    1 2.901470
    2 213.859676
};
\addplot[color=plotColor1,heuristicMarkSign,shortenLines,line width=1pt] table {
    2 213.859676
    3 20211.224772
};

\addplot[color=plotColor2,heuristicMarkSign,shortenLines,line width=1pt] table {
    1 0.266506
    2 1.077578
};
\addplot[color=plotColor2,heuristicMarkSign,shortenLines,line width=1pt] table {
    2 1.077578
    3 129.981777
};
\addplot[color=plotColor2,heuristicMarkSign,shortenLines,line width=1pt] table {
    3 129.981777
    4 2229.902942
};
\addplot[color=plotColor2,heuristicMarkSign,shortenLines,line width=1pt] table {
    4 2229.902942
    5 42559.306855
};

\addplot[color=plotColor3,exactMarkSign,shortenLines,line width=1pt] table {
    1 0.190112
    2 1.349123
};
\addplot[color=plotColor3,exactMarkSign,shortenLines,line width=1pt] table {
    2 1.349123
    3 10.260272
};
\addplot[color=plotColor3,exactMarkSign,shortenLines,line width=1pt] table {
    3 10.260272
    4 409.848081
};
\addplot[color=plotColor3,exactMarkSign,shortenLines,line width=1pt] table {
    4 409.848081
    5 50581.923936
};

\addplot[color=plotColor4,exactMarkSign,shortenLines,line width=1pt] table {
    1 0.207168
    2 1.333068
};
\addplot[color=plotColor4,exactMarkSign,shortenLines,line width=1pt] table {
    2 1.333068
    3 6.165035
};
\addplot[color=plotColor4,exactMarkSign,shortenLines,line width=1pt] table {
    3 6.165035
    4 96.788332
};
\addplot[color=plotColor4,exactMarkSign,shortenLines,line width=1pt] table {
    4 96.788332
    5 1154.349906
};
\addplot[color=plotColor4,exactMarkSign,shortenLines,line width=1pt] table {
    5 1154.349906
    6 8755.665961
};

\addplot[color=plotColor5,exactMarkSign,shortenLines,line width=1pt] table {
    1 0.268757
    2 0.904524
};
\addplot[color=plotColor5,exactMarkSign,shortenLines,line width=1pt] table {
    2 0.904524
    3 2.121455
};
\addplot[color=plotColor5,exactMarkSign,shortenLines,line width=1pt] table {
    3 2.121455
    4 7.099843
};
\addplot[color=plotColor5,exactMarkSign,shortenLines,line width=1pt] table {
    4 7.099843
    5 29.642054
};
\addplot[color=plotColor5,exactMarkSign,shortenLines,line width=1pt] table {
    5 29.642054
    6 236.698116
};
\addplot[color=plotColor5,exactMarkSign,shortenLines,line width=1pt] table {
    6 236.698116
    7 2655.627588
};
\addplot[color=plotColor5,exactMarkSign,shortenLines,line width=1pt] table {
    7 2655.627588
    8 84719.889459
};

\addplot[color=plotColor7,exactMarkSign,shortenLines,line width=1pt] table {
    1 2.285319
    2 4.581735
};
\addplot[color=plotColor7,exactMarkSign,shortenLines,line width=1pt] table {
    2 4.581735
    3 3.100730
};
\addplot[color=plotColor7,exactMarkSign,shortenLines,line width=1pt] table {
    3 3.100730
    4 6.641269
};
\addplot[color=plotColor7,exactMarkSign,shortenLines,line width=1pt] table {
    4 6.641269
    5 13.524816
};
\addplot[color=plotColor7,exactMarkSign,shortenLines,line width=1pt] table {
    5 13.524816
    6 78.469881
};
\addplot[color=plotColor7,exactMarkSign,shortenLines,line width=1pt] table {
    6 78.469881
    7 370.166582
};
\addplot[color=plotColor7,exactMarkSign,shortenLines,line width=1pt] table {
    7 370.166582
    8 2346.711010
};
\addplot[color=plotColor7,exactMarkSign,shortenLines,line width=1pt] table {
    8 2346.711010
    9 17400.965450
};
\addplot[color=plotColor7,exactMarkSign,shortenLines,line width=1pt] table {
    9 17400.965450
    10 151433.589000
};
\addplot[color=plotColor7,exactMarkSign,shortenLines,line width=1pt] table {
    10 151433.589000
    11 679855.710000
};
\addplot[color=plotColor7,exactMarkSign,shortenLines,line width=1pt] table {
    11 679855.710000
    12 1437867.530000
};

\addplot[color=plotColor8,exactMarkSign,shortenLines,line width=1pt] table {
    1 0.099351
    2 0.223936
};
\addplot[color=plotColor8,exactMarkSign,shortenLines,line width=1pt] table {
    2 0.223936
    3 0.398247
};
\addplot[color=plotColor8,exactMarkSign,shortenLines,line width=1pt] table {
    3 0.398247
    4 1.112500
};
\addplot[color=plotColor8,exactMarkSign,shortenLines,line width=1pt] table {
    4 1.112500
    5 4.929401
};
\addplot[color=plotColor8,exactMarkSign,shortenLines,line width=1pt] table {
    5 4.929401
    6 36.533460
};
\addplot[color=plotColor8,exactMarkSign,shortenLines,line width=1pt] table {
    6 36.533460
    7 346.256005
};
\addplot[color=plotColor8,exactMarkSign,shortenLines,line width=1pt] table {
    7 346.256005
    8 11535.030447
};

\addplot[color=plotColor9,heuristicMarkSign,shortenLines,line width=1pt] table {
    1 0.112021
    2 0.219375
};
\addplot[color=plotColor9,heuristicMarkSign,shortenLines,line width=1pt] table {
    2 0.219375
    3 0.360147
};
\addplot[color=plotColor9,heuristicMarkSign,shortenLines,line width=1pt] table {
    3 0.360147
    4 0.772839
};
\addplot[color=plotColor9,heuristicMarkSign,shortenLines,line width=1pt] table {
    4 0.772839
    5 2.868105
};
\addplot[color=plotColor9,heuristicMarkSign,shortenLines,line width=1pt] table {
    5 2.868105
    6 11.923126
};
\addplot[color=plotColor9,heuristicMarkSign,shortenLines,line width=1pt] table {
    6 11.923126
    7 38.847463
};
\addplot[color=plotColor9,heuristicMarkSign,shortenLines,line width=1pt] table {
    7 38.847463
    8 238.299267
};
\addplot[color=plotColor9,heuristicMarkSign,shortenLines,line width=1pt] table {
    8 238.299267
    9 1374.157484
};
\addplot[color=plotColor9,heuristicMarkSign,shortenLines,line width=1pt] table {
    9 1374.157484
    10 8182.909492
};
\addplot[color=plotColor9,heuristicMarkSign,shortenLines,line width=1pt] table {
    10 8182.909492
    11 38144.922870
};
\addplot[color=plotColor9,heuristicMarkSign,shortenLines,line width=1pt] table {
    11 38144.922870
    12 33447.646600
};
\addplot[color=plotColor9,heuristicMarkSign,shortenLines,line width=1pt] table {
    12 33447.646600
    13 1066.983700
};

\end{axis}
\end{tikzpicture}
  \caption{Average running times for different battery capacities. For the same set of queries as in Figure~\ref{fig:champ-limitplot-med}, this plot shows the corresponding average running time of 100 random in-range queries.}
  \label{fig:champ-limitplot-avg}
\end{figure}

\end{document}